\documentclass[a4paper,11pt]{article}

\usepackage{amsthm}
\usepackage{fullpage}%
\usepackage{graphicx,tikz,enumerate}
\usepackage[colorlinks=true,citecolor=blue,linkcolor=red]{hyperref}
\usepackage[title]{appendix}
\usepackage{amssymb, amsfonts,mathtools,stmaryrd}
\usepackage{xcolor}
\usepackage{enumitem}
\usepackage{etoolbox}

\usepackage{algorithm}
\usepackage{algpseudocode}

\usepackage{dsfont}
\usepackage{bbm}
\newtheorem{definition}{\textbf{Definition}}
\newtheorem{theorem}[definition]{\textbf{Theorem}}
\newtheorem{lemma}[definition]{\textbf{Lemma}}
\newtheorem{corollary}[definition]{\textbf{Corollary}}
\newtheorem{remark}[definition]{\textbf{Remark}}
\newtheorem{proposition}[definition]{\textbf{Proposition}}
\newtheorem{example}[definition]{\textbf{Example}}

\newtheorem{assumption}[definition]{\textbf{Assumption}}

\newtheorem*{proof*}{\textbf{Proof}}
\newcommand{\brck}[1]{\llbracket #1 \rrbracket}
\newcommand{\Ag}{\msf{Ag}}

\newcommand{\fanBrOp}[1]{\langle\!\langle#1\rangle\!\rangle}

\newcommand{\Interval}[1]{\ensuremath{\llbracket #1 \rrbracket}}
\newcommand{\N}{\mathbb{N}}

\newcommand{\msf}[1]{\mathsf{#1}}

\newcommand{\LTL}{\msf{LTL}}
\newcommand{\ML}{\msf{ML}}

\newcommand{\CTL}{\msf{CTL}}
\newcommand{\ATL}{\msf{ATL}}
\newcommand{\SAT}{\msf{SAT}}
\newcommand{\PCTL}{\msf{PCTL}}

\newcommand{\head}[1]{\ensuremath{\msf{hd}(#1)}}
\newcommand{\body}[1]{\ensuremath{\msf{bd}(#1)}}

\newcommand{\kar}[1]{\ensuremath{\msf{k}_{\msf{#1}}}}

\newcommand{\act}{\mathsf{Act}}
\newcommand{\prop}{\mathsf{Prop}}

\DeclareMathOperator{\lF}{\mathbf{F}}
\DeclareMathOperator{\lG}{\mathbf{G}}
\DeclareMathOperator{\lU}{\mathbf{U}}
\DeclareMathOperator{\lX}{\mathbf{X}}

\usepackage{bookmark}

\usepackage{pifont}
\usepackage{authblk}
\usepackage{graphicx}
\usepackage{multirow}
\author[1,2]{Benjamin Bordais}
\author[1,2]{Daniel Neider}
\affil[1]{TU Dortmund University, Dortmund, Germany}
\affil[2]{Center for Trustworthy Data Science and Security, University Alliance Ruhr, Dortmund, Germany}
\date{}
\usepackage{array,booktabs}
\newcolumntype{M}[1]{>{\centering\arraybackslash}m{#1}}

\begin{document}
	\title{A framework for computing upper bounds in passive learning settings}
	\maketitle	
	
	\begin{abstract}
		The task of inferring logical formulas from examples has garnered significant attention as a means to assist engineers in creating formal specifications used in the design, synthesis, and verification of computing systems. Among various approaches, enumeration algorithms have emerged as some of the most effective techniques for this task. These algorithms employ advanced strategies to systematically enumerate candidate formulas while minimizing redundancies by avoiding the generation of syntactically different but semantically equivalent formulas. However, a notable drawback is that these algorithms typically do not provide guarantees of termination%, which poses challenges for their use in real-world applications
		.
		
		This paper develops an abstract framework to bound the size of possible solutions for a logic inference task, thereby providing a termination guarantee for enumeration algorithms through the introduction of a sufficient stopping criterion. The proposed framework is designed with flexibility in mind and is applicable to a broad spectrum of practically relevant logical formalisms, including Modal Logic, Linear Temporal Logic, Computation Tree Logic, Alternating-time Temporal Logic, Probabilistic Computation Tree Logic and even selected inference task for automata. In addition, our approach enabled us to develop a meta algorithm that enumerates over the semantics of formulas rather than their syntactic representations, offering new possibilities for reducing redundancy.
	\end{abstract}

	\section{Introduction}

The goal of formal verification is to provide strong guarantees on the behavior of reactive systems. Formal verification techniques rely both on the use of mathematical models of systems and on formal specifications, which describe the intended behavior of the system. However, constructing formal specifications is no easy task, and doing it manually often leads to errors, which makes the specifications unreliable. The lack of usable and trustworthy specifications is a large impediment on the effectiveness of formal methods \cite{DBLP:conf/vstte/Rozier16}.

To circumvent this issue, a recent research trend is targeted towards automatically generating (or learning) formal specifications, written as logical formulas, %which have often taken the form of the study of how to learn logical formulas 
from examples. This approach has been explored with many kinds of temporal logics, such as Linear Temporal Logic ($\LTL$) \cite{flie,CamachoIKVM19,scarlet,learning-ltl-noisy-data-3,DBLP:conf/cav/ValizadehFB24}, Computation Tree Logic ($\CTL$) \cite{DBLP:conf/fmcad/EhlersGN20,DBLP:conf/ijcar/PommelletSS24,DBLP:conf/fm/BordaisNR24}, Alternating-time Temporal Logic ($\ATL$) \cite{DBLP:conf/fm/BordaisNR24,DBLP:journals/corr/abs-2408-04486}, 
Signal Temporal Logic ($\msf{STL}$) \cite{dtmethod,MohammadinejadD20}, Past Time LTL (PLTL) \cite{DBLP:conf/fmcad/ArifLERCT20}, the Property Specification Language ($\msf{PSL}$) \cite{0002FN20}, Metric Temporal Logic ($\msf{MTL}$) \cite{DBLP:conf/vmcai/RahaRFNP24}, etc. Note that learning from examples has also been studied in other context than temporal logics, see e.g. \cite{Angluin78,DBLP:journals/iandc/Gold78} for learning finite automata and regular expressions. 

Learning logical formulas from examples is often done in a passive learning setting where, given a finite set of positive and negative examples, the goal is to synthesize --- or decide the existence of --- a separating formula, i.e., a formula satisfied by all positive models, and rejected by all the negative ones. There are three main techniques used in the literature to solve the passive learning problem: (1) constraint-solving \cite{flie,CamachoM19,Riener19,learning-ltl-noisy-data-1,DBLP:conf/ilp/IeloLFRGR23,DBLP:conf/fm/BordaisNR24}, which translates the learning problem into one or more constraint satisfaction problems and applies off-the-shelf solvers to find a solution; (2) neuro-symbolic techniques \cite{learning-ltl-noisy-data-3,DBLP:conf/aaai/WanLDLYP24}, which encode the learning problem into an input that (graph) neural networks can process to output a separating formula; and (3) enumerative search algorithms \cite{scarlet,DBLP:conf/cav/ValizadehFB24,MohammadinejadD20} which syntactically enumerate candidate formulas --- possibly with the help of handcrafted templates \cite{Chan00,WasylkowskiZ11} --- until a separating formula is found. 

While this latter enumeration technique is the most efficient in practice \cite{scarlet,DBLP:conf/cav/ValizadehFB24,MohammadinejadD20}, many algorithms proposed in the literature lack theoretical groundings. Clever enumeration algorithms and correctness proofs are provided, but termination arguments are rarely given. The main goal of this paper is to provide a general framework from which one can derive upper bounds on the minimal size of separating formulas in various passive learning settings, which gives a termination condition for enumeration-based algorithm. A version of this theorem was established in \cite[Theorem 2]{DBLP:conf/fm/BordaisNR24} with $\CTL$- and $\ATL$-formulas, we extend it here to a much wider class of logical formalisms.

There are several already-existing size-related results for passive learning, but they all focus on a specific setting, e.g. it is folklore that polynomial-size automata are sufficient to separate sets of positive and negative finite words (since any finite language is regular); in \cite{arXivFijalkow}, a whole section is dedicated to $\LTL$-fragments (evaluated on finite words) for which separating formulas may have polynomial size; in \cite{DBLP:conf/time/GorankoK16}, the authors study the size of (temporal logic) formulas distinguishing non-bisimilar transitons systems; in \cite{funk2019concept}, the minimal size of separating concepts (which derive from descriptive logic) is studied.

\textbf{Our contributions.} In Section~\ref{sec:Def}, we define an abstract logical formalism that can be instantiated with many different concrete logical formalisms. It is 
deliberately designed to be simple and easy to instantiate and consists of a set of formula types and a set of operators.
In our running example of Modal Logic ($\ML$) formulas evaluated on Kripke structures, there is a single formula type, while the set of operators is composed of e.g. $\neg,\wedge,\vee,[\cdot],\langle\cdot\rangle$. 

In Section~\ref{sec:results}, we introduce the main result of this paper, which hinges on the notion of semantic values, to which logical formulas are mapped, which capture the semantics of the logics. For our running example, the semantic values are the set of states of the Kripke structure satisfying a modal logic formula. Our main result Theorem~\ref{thm:size_separation} shows that the minimal size of a separating formulas, assuming one exists, is upper bounded by the total number of different semantic values. 

In addition to the upper bounds that it provides, the proof of Theorem~\ref{thm:size_separation} also suggests a promising semantic-based enumeration algorithm. Indeed, one of the common pitfalls of the enumeration algorithms of the literature, that clever techniques attempt to avoid as much as possible, is to generate syntactically different but semantically equivalent formulas. Following the proof of Theorem~\ref{thm:size_separation}, we exhibit a meta algorithm --- which can be instantiated with various concrete logics --- which bypasses the above-mentioned issue by design as it enumerates semantic values instead of formulas. When instantiating this meta algorithm to modal logic, we obtain an exponential time algorithm, while formula-enumeration algorithms may have a doubly-exponential time complexity. This is discussed in Subsection~\ref{subsec:algo}. 

In Section~\ref{sec:examples}, we demonstrate the applicability of Theorem~\ref{thm:size_separation} to a wide range of logical formalisms% to exhibit exponential upper bounds on the minimal size of separating formulas
, including: our running example $\ML$-formulas evaluated on Kripke structures, $\LTL$-formulas evaluated on finite and infinite words, $\CTL$-formulas evaluated on Kripke structures, $\ATL$-formulas evaluated on concurrent game structures %(these two use-cases show that Theorem~\ref{thm:size_separation} does generalize \cite[Theorem 2]{DBLP:conf/fm/BordaisNR24})
, $\PCTL$-formulas evaluated on Markov chains (these last three cases are very similar to the case of $\ML$-formulas), and $\LTL$-formulas evaluated on Kripke structures. For this latter case, we are unable to apply Theorem~\ref{thm:size_separation} to the full logic, but we use the assumptions of this theorem to find a related logic for which we can apply Theorem~\ref{thm:size_separation}. Finally, we also apply Theorem~\ref{thm:size_separation} to establish upper bounds on the minimal size of words separating automata, thus showing that out framework does only apply to logical formalisms. %Note that this different from automata learning.

In Section~\ref{sec:lower_bound}, we argue that the upper bounds exhibited in the previous section are not absurd. Specifically, we investigate two logical fragments --- of $\LTL$-formulas (resp. $\ML$-formulas) evaluated on ultimately periodic words (resp. Kripke structures) --- where our upper bounds asymptotically (almost) match a lower bound. 

Many technical details are postponed to the Appendix.

\section{Definitions}
\label{sec:Def}
We let $\N$ (resp. $\N_1$) denote the set of (resp. positive) integers. For all $i \leq j \in \N$, we let $\Interval{i,j} \subseteq \N$ denote the set $\Interval{i,j} := \{ k \in \N \mid i \leq k \leq j\}$.

For all non-empty sets $Q$, we let $2^Q := \{ A \subseteq Q \}$ denote the set of subsets of $Q$. Furthermore, for all sets $\msf{A} = (\msf{A}_q)_{q \in Q}$ indexed by $Q$ and tuples $\msf{S} \in (\msf{A}_q)_{q \in Q}$, for all $q \in Q$, we let $\msf{S}[q] \in \msf{A}$ denote the element of $\msf{S}$ corresponding to $q \in Q$. %In addition, for all 

\subsection{Modal logic and Kripke structures}
In the following, we are going to use modal logic \cite{DBLP:books/cu/BlackburnRV01} as a running example to give the intuition behind the various notions that we will define. Thus, let us first introduce the syntax and semantics of modal logic formulas. 

\begin{definition}[Modal logic syntax]
	\label{def:modal_logic_syntax}
	Consider a non-empty set of propositions $\prop$ and a non-empty set of actions $\act$. Modal logic formulas (abbreviated as $\ML(\prop,\act)$-formulas) are constructed from the grammar:
	\begin{equation*}
		\varphi ::= p \mid \neg \varphi \mid \varphi \vee \varphi \mid \varphi \wedge \varphi \mid \langle a \rangle^{\geq k} \varphi \mid [a] \varphi
	\end{equation*}
	where $p \in \prop$ is a proposition, $a \in \act$ is an action, and $k \in \N_1$. %The above grammar actually corresponds to graded multi-modal logic. Classical modal logic only involves the operator $\langle a \rangle^{\geq k}$ for $k = 1$, with a single action: $|\act| = 1$.
\end{definition}

Modal logics formulas are usually evaluated on Kripke structures, i.e. graphs with proposition-labeled  states and action-labeled transitions.
\begin{definition}[Kripke structures]
	\label{def:kripke_structure}
	A Kripke structure $K$ is defined by a tuple $(Q,I,A,\delta,P,\pi)$ where $Q$ is a non-empty set of states, $I \subseteq Q$ is the non-empty set of initial states, $A$ is a non-empty set of actions, $\delta: Q \times A \rightarrow 2^Q$ is the transition function, $P$ is a set of propositions, and $\pi: Q \rightarrow 2^P$ maps every state to the set of propositions satisfied at that state. Given a non-empty set of propositions $\prop$ and a non-empty set of actions $\act$, we let $\mathcal{K}(\prop,\act)$ denote the set of Kripke structures $\mathcal{K}(\prop,\act) := \{ K = (Q,I,A,\delta,P,\pi) \mid A \subseteq \act,P \subseteq \prop \}$.
	
	Unless otherwise stated, a Kripke structure $K$ refers to the tuple $(Q,I,A,\delta,P,\pi)$. %The size $|K|$ of a Kripke structure $K$ is defined by $|K| := |Q| \cdot |A| + |P|$.
\end{definition}

Modal logic formulas are evaluated on Kripke structures using the semantics below.
\begin{definition}[Modal logic semantics]
	\label{def:semantics_modal_logic}
	Consider a non-empty set of propositions $\prop$, a non-empty set of actions $\act$, and a Kripke structure $K \in \mathcal{K}(\prop,\act)$. Given a state $q \in Q$, and an $\ML(\prop,\act)$-formula $\varphi$, we define when $\varphi$ is satisfied in $q$ inductively as follows: 
	
	\vspace{-0.65cm}
	\hspace{-1.3cm}
	\begin{minipage}[t]{0.25\textwidth}
		\begin{align*}
			& q \models p & \text{ iff } &p \in \pi(q) \\
			& q \models \varphi_1 \wedge \varphi_2 & \text{ iff } & q \models \varphi_1 \text{ and } q \models \varphi_2\\
			& q \models \varphi_1 \vee \varphi_2 & \text{ iff } & q \models \varphi_1 \text{ or } q \models \varphi_2
		\end{align*}
	\end{minipage}
	\begin{minipage}[t]{0.05\textwidth}
		\begin{align*}
			\mid \\
			\mid \\
			\mid \\
		\end{align*}
	\end{minipage}
	\begin{minipage}[t]{0.45\textwidth}
		\begin{align*}
			& q \models \neg \varphi & \text{ iff } & q \not\models \varphi \\
			& q \models [a]  \varphi & \text{ iff } & \delta(q,a) \subseteq \{q' \in Q \mid q' \models \varphi\} \\
			& q \models \langle a \rangle^{\geq k} \varphi & \text{ iff } & |\delta(q,a) \cap \{q' \in Q \mid q' \models \varphi\}| \geq k
		\end{align*}
	\end{minipage}
	
	Then, a Kripke structure $K \in \mathcal{K}(\prop,\act)$ satisfies an $\ML(\prop,\act)$-formula $\varphi$ (denoted $K \models \varphi$) if and only if, for all $q \in I$, we have $q \models \varphi$.
\end{definition}
%As an example, the modal logic formula $\varphi := \neg p \wedge \langle a \rangle^{\geq 2} p$ is satisfied by the states of a Kripke structure that are not labeled by $p$, and that have at least two $a$-successors labeled by $p$.

In all of the examples below on modal logic formulas, we will consider a fixed non-empty set of propositions $\prop$ and non-empty set of actions $\act$. 

\subsection{Abstract logical formalism}
The goal of this paper is to establish results that can be applied to %upper bounds on the size of formulas in a passive learning setting for 
a wide range of logics. Thus, we define an abstract logical formalism that can be instantiated with various concrete logical formalisms. First of all, let us describe the syntax of our abstract logical formalism. 

Consider the modal logic syntax in Definition~\ref{def:modal_logic_syntax}. It is described by the set of operators $\msf{Op}$ that can be used (e.g. $\neg,\wedge$, etc.) along with their arity (e.g. one for $\neg$, two for $\wedge$, etc). In our abstract formalism, we allow logic syntaxes that are a little more involved. Specifically, we assume that formulas may have various types (we use a set $\mathcal{T}$ to collect all such types), some of which are final (in $\mathcal{T}_{\msf{f}} \subseteq \mathcal{T}$). Since there are various types of formulas, each operator $\msf{o} \in \msf{Op}$ also has a type $\tau_{\msf{o}} \in \mathcal{T}$, and the $i$-th argument (for $i \in \{1,2\}$) of this operator $\msf{o}$ may only be of certain types $T(\msf{o},i) \subseteq \mathcal{T}$. This abstract syntax is formally defined below.
\begin{definition}[Abstract syntax]
	\label{def:syntax_logic}
	The syntax of a logic $\msf{L}$ is defined by a tuple $\msf{Stx}_{\msf{L}} = (\mathcal{T},\mathcal{T}_{\msf{f}},\mathsf{Op},(\tau_\msf{o})_{\msf{o} \in \msf{Op}},T)$ where $\mathcal{T} \neq \emptyset$ is the finite set of types of formulas, $\emptyset \neq \mathcal{T}_{\msf{f}} \subseteq \mathcal{T}$ is the set of final types, $\msf{Op} := \msf{Op}_0 \uplus \msf{Op}_1 \uplus \msf{Op}_2$ is the set of operators with $\msf{Op}_0 \neq \emptyset$, with for $0 \leq i \leq 2$, $\msf{Op}_i$ denoting the set of operators of arity $i$. For all operators $\msf{o} \in \msf{Op}$, we let $\kar{o} \in \{0,1,2\}$ be such that $\msf{o} \in \msf{Op}_{\kar{o}}$. Moreover, $\tau_{\msf{o}} \in \mathcal{T}$ is the type of the operator $\msf{o}$; and, for all $i \in \Interval{1,\kar{o}}$, $T(\msf{o},i) \subseteq \mathcal{T}$ is the set of possible types of the $i$-th argument of the operator $\msf{o}$.
	
	The set $\msf{Fm}_\msf{L}$ of $\msf{L}$-formulas is then defined inductively as follows. 
	\begin{itemize}
		\item For all $\msf{o} \in \msf{Op}_0$, $\varphi := \msf{o}$ is an $\msf{L}$-formula of type $\tau_{\msf{o}} \in \mathcal{T}$: $\varphi \in \msf{Fm}_\msf{L}(\tau_{\msf{o}})$.
		\item For all $\msf{o} \in \msf{Op}_1$, $\varphi_1 \in \bigcup_{\tau_1 \in T(\msf{o},1)} \msf{Fm}_\msf{L}(\tau_1)$: $\msf{o}(\varphi_1)%$ is of type $\tau_{\msf{o}} \in \mathcal{T}$: $\varphi 
		\in \msf{Fm}_\msf{L}(\tau_{\msf{o}})$. 
		\item For all $\msf{o} \in \msf{Op}_2$, $(\varphi_1,\varphi_2) \in \bigcup_{\tau_1 \in T(\msf{o},1)} \msf{Fm}_\msf{L}(\tau_1) \times \bigcup_{\tau_2 \in T(\msf{o},2)} \msf{Fm}_\msf{L}(\tau_2)$: $\msf{o}(\varphi_1,\varphi_2)$\footnote{When writing concrete logical formulas, we will use the infix notation $\varphi_1 \; \msf{o} \; \varphi_2$.}% is a formula of type $\tau_{\msf{o}} \in \mathcal{T}$: $\varphi 
		$\in \msf{Fm}_\msf{L}(\tau_{\msf{o}})$. 
	\end{itemize}
	%In the following, we will always assume that, for all types $\tau \in \mathcal{T}$, there are $\msf{L}$-formulas of type $\tau$, i.e. $\msf{Fm}_\msf{L}(\tau) \neq \emptyset$. 
	For all $X \subseteq \mathcal{T}$, we let $\msf{Fm}_\msf{L}(X) := \bigcup_{\tau \in X} \msf{Fm}_\msf{L}(\tau)$. Then, we let $\msf{Fm}_\msf{L} := \msf{Fm}_\msf{L}(\mathcal{T})$ (resp. $\msf{Fm}_\msf{L}^{\msf{f}} := \msf{Fm}_\msf{L}(\mathcal{T}_{\msf{f}})$) be the set of all (resp. final) $\msf{L}$-formulas.
\end{definition}
Unless otherwise stated, whenever we consider a logic $\msf{L}$, its syntax will be assumed to be given by the tuple $(\mathcal{T},\mathcal{T}_{\msf{f}},\mathsf{Op},(\tau_\msf{o})_{\msf{o} \in \msf{Op}},T)$.

%We exemplify this notion below on modal logic formulas.
\begin{example}
	\label{ex:modal_logic}
	Let us encode the modal logic grammar of Definition~\ref{def:modal_logic_syntax} into this abstract formalism. The $\ML(\prop,\act)$-syntax is given by the tuple $(\mathcal{T},\mathcal{T}_\msf{f},\mathsf{Op},(\tau_\msf{o})_{\msf{o} \in \msf{Op}},T)$ where: $\mathcal{T} := \mathcal{T}_\msf{f} := \{ \tau \}$; $\msf{Op}_0 := \{ p\mid p \in \prop \}$, $\msf{Op}_1 := \{ \neg,\langle a \rangle^{\geq k},[a] \mid a \in \act, k \in \N\}$, $\msf{Op}_2 := \{\vee,\wedge\}$; and for all $\msf{o} \in \msf{Op}$, we have $\tau_\msf{o} := \tau$ and, for all $i \in \Interval{1,\kar{o}}$, we have $T(\msf{o},i) := \mathcal{T}$.
\end{example}

A (syntactic) fragment $\msf{L}'$ of a logic $\msf{L}$ is a logic with the same syntax, but potentially fewer operators, as formally defined below.
\begin{definition}[Syntactic fragment]
	\label{def:fragment}
	A logic $\msf{L}' = (\mathcal{T},\mathcal{T}_{\msf{f}},\mathsf{Op}',(\tau_\msf{o})_{\msf{o} \in \msf{Op}},T)$ is a \emph{fragment} of a logic $\msf{L} = (\mathcal{T},\mathcal{T}_{\msf{f}},\mathsf{Op},(\tau_\msf{o})_{\msf{o} \in \msf{Op}},T)$ if $\msf{Op}' \subseteq \msf{Op}$. In the following, whenever we refer to a fragment $\msf{L}'$ of a logic $\msf{L}$, we will denote by $\msf{Op}'$ the set of operators of this fragment $\msf{L}'$% (while $\mathcal{T},\mathcal{T}_\msf{f},(\tau_\msf{o})_{\msf{o} \in \msf{Op}}$ and $T$ are the same)
	. 
\end{definition}

In this paper we are particularly interested in the size of formulas. There are two natural (and inductive) ways to define the size $\msf{sz}(\varphi)$ of a formula $\varphi$: either as the size of its syntax tree --- in which case the size of e.g. $\varphi_1 \wedge \varphi_2$ is equal to one plus the sizes of $\varphi_1$ and $\varphi_2$ --- or as the number of sub-formulas, i.e. the size of its syntax DAG --- in which case the size of $\varphi_1 \wedge \varphi_2$ is equal to one plus the number of different sub-formulas in $\varphi_1$ and $\varphi_2$. For instance, the syntax-tree size of the $\ML$-formula $\varphi := \neg p \wedge \langle a \rangle^{\geq 2} p$ is five; its syntax-DAG size is four, since the set of sub-formulas of $\varphi$ is $\msf{Sub}(\varphi) = \{p,\neg p,\langle a \rangle^{\geq 2} p,\varphi\}$. In this paper, we use the syntax-DAG size since it is very well-suited for the inductive arguments that we will use.
\begin{definition}[Formula size]
	We define the set of sub-formulas of an $\msf{L}$-formula by induction: for all $\psi \in \msf{Op}_0$, $\msf{Sub}(\psi) := \{\psi\}$; for all $\psi = \msf{o}(\varphi)$, $\msf{Sub}(\psi) := \{\psi\} \cup  \msf{Sub}(\varphi)$; for all $\psi = \msf{o}(\varphi_1,\varphi_2)$, $\msf{Sub}(\psi) := \{\psi\} \cup \msf{Sub}(\varphi_1) \cup \msf{Sub}(\varphi_2)$. We then define the size of $\varphi$: $\msf{sz}(\varphi) := |\msf{Sub}(\varphi)|$.
\end{definition}

A model is a mathematical structure that gives meaning to the logic. As in the concrete setting, we assume that there is a satisfaction relation expressing when some model satisfies a formula. We will see later the important properties that satisfaction relations may enjoy. %We do this abstraction because we want to able to evaluate the same formulas on the models. 
%
%Logical formulas are evaluated on models, e.g. modal logic formulas are evaluated on Kripke structures. The main theorem of this paper relies on several assumptions made on the behavior of this satisfaction relation. However, for now, in our abstract formalism, we only assume that there exists a satisfaction relation between logical formulas and models, as defined below.
%Let us now evaluate these logical formulas on models, i.e. let us consider the semantics of the logics. However, here, we simply assume that the semantics is given by a satisfaction relation between models and formulas, as described below. We will consider in more details later on properties enjoyed by this satisfaction relation.	
%between the cross-product ... such that . 
\begin{definition}[Satisfaction relation]
	\label{def:evaluation_formulas_models}
	Let $\msf{L}$ be a logic. A structure $M$ is an $\msf{L}$-model if there exists a satisfaction relation $\models$ between $M$ and each $\msf{L}$-formula: for all $\varphi \in \msf{Fm}_{\msf{L}}^{\msf{f}}$, $M \models \varphi$ means that the model $M$ satisfies the formula $\varphi$, while $M \not\models \varphi$ means that $M$ does not satisfy the formula $\varphi$. A set $\mathcal{C}$ is a class of $\msf{L}$-models if each structure $M$ in $\mathcal{C}$ is an $\msf{L}$-model. 
	%if there is relation $\models \subseteq \mathcal{C} \times \msf{Fm}_{\mathcal{L}}^{\msf{f}}$ such that, for all $M \in \mathcal{C}$, and $\varphi \in \msf{Fm}_{\msf{L}}^{\msf{f}}$, $M \models \varphi$ refers to that fact that the model $M$ satisfies the formula $\varphi \in \msf{Fm}_{\msf{L}}^{\msf{f}}$. %In that case, $M$ is called an $\msf{L}$-\emph{model}. 
	%On the other hand, the notation $M \not\models \varphi$ refers to the fact that $M \models \varphi$ does not hold. 
	%In that case, each element $M \in \mathcal{C}$ is called an $\msf{L}$-\emph{model}. 
	%
	%For each $\msf{L}$-\emph{model} $M \in \mathcal{C}$, for complexity issues, we also assume that there is an integer, denoted $|M| \in \N$, that corresponds to the size of the model $M$, i.e. how costly it is to represent $M$. 
\end{definition}

\begin{example}
	\label{ex:Kripke_structures}
	The set of $\ML(\prop,\act)$-models is equal to $\mathcal{K}(\prop,\act)$, i.e. the set of Kripke structures $K$ such that $A \subseteq \act$ and $P \subseteq \prop$. 
\end{example}

We can now define the notion of separating formulas and the passive learning problem%, it consists of the following: given a sample $\mathcal{S}$ of a positive set $\mathcal{P}$ and a negative set $\mathcal{N}$ of models, decide if a separating formula --- accepting the positive models and rejecting the negative ones --- exists
.
\begin{definition}[Separating formula and passive learning problem]
	Consider a logic $\msf{L}$, an $\msf{L}$-fragment $\msf{L}'$, and a class of $\msf{L}$-models $\mathcal{C}$. A $\mathcal{C}$-sample $\mathcal{S}$ is a pair $\mathcal{S} = (\mathcal{P},\mathcal{N})$ where $\mathcal{P},\mathcal{N} \subseteq \mathcal{C}$ are two finite sets of $\msf{L}$-models. This sample $\mathcal{S}$ is $\msf{L}'$-\emph{separable} if there is a $\msf{L}'$-formula $\varphi \in \msf{Fm}_{\msf{L}'}^{\msf{f}}$ such that: for all $M \in \mathcal{P}$, we have $M \models \varphi$; and for all $M \in \mathcal{N}$, we have $M \not\models \varphi$. 
	
	In that case, the formula $\varphi$ is called a $\mathcal{S}$-\emph{separating} (final) formula.
	
	We denote by $\msf{PvLn}(\msf{L}',\mathcal{C})$ the decision problem that takes as input a $\mathcal{C}$-sample $\mathcal{S}$, and outputs yes if and only if the sample $\mathcal{S}$ is $\msf{L}'$-separable.
\end{definition}
Unless otherwise stated, $\mathcal{C}$-samples $\mathcal{S}$ refer to the pair $\mathcal{S} = (\mathcal{P},\mathcal{N})$. With an abuse of notation, we will also identify the sample $\mathcal{S} = (\mathcal{P},\mathcal{N})$ with the set $\mathcal{P} \cup \mathcal{N}$.

\iffalse
Let us start by stating a proposition instantiating this result with modal logic formulas evaluated on Kripke structures.
\begin{proposition}
	\label{prop:size_separation_modal_logic}
	Consider a non-empty set of propositions $\prop$ and a non-empty set of actions $\act$. For all fragments $\msf{L}'$ of the modal logic $\msf{ML}(\prop,\act)$, a $\mathcal{K}(\prop,\act)$-sample $\mathcal{S}$ is $\msf{L}'$-separable if and only if there is an $\mathcal{S}$-separating $\msf{L}'$-formula of size at most $2^n$, where $n := \sum_{K \in \mathcal{S}} |Q_K|$.
\end{proposition}
\fi

\section{Main results}
\label{sec:results}
%Our approach does not work with every satisfaction relation. We will cover various of such logical formulas where it satisfies ... 
%
The main goal of this paper is to establish an upper bound on the minimal size of separating formulas, assuming they exist. Our approach does not work with every satisfaction relation, thus we need to assume that it satisfies some conditions. These conditions are met in various use-cases, several of which we detail in Section~\ref{sec:examples}. In this section, we formally define the assumptions that we make on the satisfaction relation; we state the main theorem of this paper and we provide a detailed proof sketch. We then discuss a meta enumeration algorithm solving the passive learning problem derived from the proof of this theorem.

%In this section we formally state the main result of this paper, along with a detailed proof sketch. Before doing so, we need to formally define some assumptions on the satisfaction relation $\models$ between the models and logical formulas upon which this main result relies. 

\subsection{Assumptions on the satisfaction relation}
Consider our running example of modal logic formulas evaluated on Kripke structures. %When seeking separating formulas in a passive learning setting
It is clear that whether a $\ML$-formula accepts a Kripke structure entirely depends on the set of states satisfying the formula. Hence, to find a separating $\ML$-formula in a passive learning setting, we may not consider the exact syntactic shape of formulas, and instead focus on their semantic value, i.e. the set of states satisfying these formulas. 
%their semantic values, in the case of $\ML$-formulas evaluated on Kripke structures, that is the set of states satisfying these formulas. 
We proceed similarly in our abstract logical formalism and we consider a finite set $\msf{SEM}$ of semantic values, and a semantic function $\msf{sem}: \msf{Fm}_{\msf{L}} \to \msf{SEM}$ mapping each $\msf{L}$-formula to a semantic value.
\begin{definition}
	\label{def:semantic_set_function}
	For a logic $\msf{L}$ and an $\msf{L}$-model $M$, an $(\msf{L},M)$\emph{-pair} $(\msf{SEM}_M,\msf{sem}_M)$ is s.t.:
	\begin{itemize}
		\item $\msf{SEM}_M = \bigcup_{\tau \in \mathcal{T}} \msf{SEM}_M(\tau)$ is a finite set of semantical values, where for all types $\tau,\tau' \in \mathcal{T}$, we have $\msf{SEM}_M(\tau) \cap \msf{SEM}_M(\tau') = \emptyset$. 
		For all $T \subseteq \mathcal{T}$, we let $\msf{SEM}_M(T) := \bigcup_{\tau \in T} \msf{SEM}_M(\tau)$.
		\item $\msf{sem}_M: \msf{Fm}_\msf{L} \rightarrow \msf{SEM}_M$ is the semantic function such that, for all types $\tau \in \mathcal{T}$ and formulas $\varphi \in \msf{Fm}_\msf{L}(\tau)$, we have $\msf{sem}_M(\varphi) \in \msf{SEM}_M(\tau)$.
	\end{itemize} 	
	Unless otherwise stated, an $(\msf{L},M)$-pair $\Theta_M$ refers to the pair $\Theta_M = (\msf{SEM}_M,\msf{sem}_M)$.
\end{definition}

\begin{example}
	\label{ex:Kripke_structures_SEM}
	Consider a Kripke structure $K \in \mathcal{K}(\prop,\act)$. %For all $q \in Q$, we let $K_q := (Q,\{q\},A,\delta,P,\pi)$ denote the Kripke structure identical to $K$ except that $q$ is its single initial state. Then, w
	We let $\Theta_K := (\msf{SEM}_K,\msf{sem}_K)$ be the $(\ML(\prop,\act),K)$-pair, such that $\msf{SEM}_K := 2^Q$ and $\msf{sem}_K: \msf{Fm}_{\ML(\prop,\act)} \rightarrow 2^Q$ maps each %$\ML(\prop,\act)$-
	formula $\varphi$ to the set of states satisfying it: $\msf{sem}_K(\varphi) := \{ q \in Q \mid q \models \varphi \} \in \msf{SEM}_K$.
\end{example}

Consider the above $(\ML(\prop,\act),K)$-pair $\Theta_K$. There are two crucial properties that this pair satisfies. The first one --- which justifies the terminology \textquotedblleft{}semantic value\textquotedblright{} --- relates to capturing the behavior of $\ML(\prop,\act)$-formulas w.r.t. the satisfaction relation $\models$ in $K$. Indeed, for any $\ML(\prop,\act)$-formula $\varphi$, given $\msf{sem}_K(\varphi)$, one can decide if $K \models \varphi$: it holds if and only if $\msf{sem}_K(\varphi) \subseteq I$. In particular, this implies that if two $\ML(\prop,\act)$-formulas are mapped to the same semantic value in $\msf{SEM}_K$%by the function $\msf{sem}_K$
, then one is satisfied by $K$ if and only if the other is. In such a case, we say that the pair $\Theta_K$ captures the $\ML(\prop,\act)$-semantics%, as formally defined below% for arbitrary logical formalisms
.
%For now, given a logic $\msf{L}$ and an $\msf{L}$-model $M$, an $(\msf{L},M)$-pair could be defined arbitrarily and is of no use. Let us define some properties that we will assume that $(\msf{L},M)$-pairs satisfy. First, we say that an $(\msf{L},M)$-pair $\Theta_M$ captures the $\msf{L}$-semantics if the fact that a final $\msf{L}$-formula $\varphi$ is satisfied by the model $M$ depends entirely on the image of $\varphi$ by the function $\msf{sem}$. This is formally defined below.
\begin{definition}[Capturing the $\msf{L}$-semantics]
	\label{def:capturing_semantics}
	Consider a logic $\msf{L}$, an $\msf{L}$-model $M$, and an $(\msf{L},M)$-pair $\Theta_M$. We say that this pair $\Theta_M$ \emph{captures the $\msf{L}$-semantics} if:
	\begin{equation*}
		\forall \tau \in \mathcal{T},\; \forall \varphi,\varphi' \in \msf{Fm}^{\msf{f}}_\msf{L}(\tau):\; \msf{sem}_M(\varphi) = \msf{sem}_M(\varphi') \implies M \models \varphi \text{ iff }M \models \varphi'
	\end{equation*}
	%for all of $\msf{L}$-formulas $\varphi,\varphi' \in \msf{Fm}_\msf{L}$, if $\msf{sem}_M(\varphi) = \msf{sem}_M(\varphi')$, then $M \models \varphi$ if and only if $M \models \varphi'$. 
	In that case, there is some subset $\msf{SAT}_{M} \subseteq \msf{SEM}_M$ such that, for all $\msf{L}$-formulas $\varphi \in \msf{Fm}_\msf{L}^{\msf{f}}$, we have $M \models \varphi$ if and only if we have $\msf{sem}_M(\varphi) \in \msf{SAT}_M$. %Informally, this means that, to know whether an $\msf{L}$-formula is satisfied by $M$, it suffices to consider its image by the function $\msf{sem}_M$.
\end{definition}

\begin{lemma}[Proof~\ref{proof:lem_modal_logic_pair_captures_semantics}]
	\label{lem:modal_logic_pair_captures_semantics}
	For all $K \in \mathcal{K}(\prop,\act)$, the $(\ML(\prop,\act),K)$-pair $\Theta_K$ from Example~\ref{ex:Kripke_structures_SEM} captures the $\ML(\prop,\act)$-semantics.
\end{lemma}

The $(\ML(\prop,\act),K)$-pair $\Theta_K$ from Example~\ref{ex:Kripke_structures_SEM} satisfies a second crucial property, which relates to the fact that the semantic function $\msf{sem}_K$ can be computed in an inductive way. Consider for instance the $\ML(\prop,\act)$-formula $\varphi := \varphi_1 \wedge \varphi_2$. The semantic value $\msf{sem}_K(\varphi) \in \msf{SEM}_K$ is equal to the set of states in $K$ satisfying the formula $\varphi$. By definition of the operator $\wedge$, this exactly corresponds to the set of states in $K$ that satisfy both formulas $\varphi_1$ and $\varphi_2$. This implies that the semantic value $\msf{sem}_K(\varphi)$ can be computed from $\msf{sem}_K(\varphi_1)$ and $\msf{sem}_K(\varphi_2)$, regardless of what the formulas $\varphi_1$ and $\varphi_2$ actually are. In fact, this holds for all $\ML(\prop,\act)$-operators, not only $\wedge$, e.g. the semantic value of the $\ML(\prop,\act)$-formula $\varphi := [a] \varphi'$ is equal to the set of states whose $a$-successors are all in $\msf{sem}_K(\varphi')$. In such a case, we say that this pair $\Theta_K$ satisfies the inductive property. %This is formally defined below% for arbitrary logical formalisms
\begin{definition}
	\label{def:inductive_semantics}
	For a logic $\msf{L}$ and an $\msf{L}$-model $M$, an $(\msf{L},M)$-pair $\Theta_M$ satisfies the \emph{inductive property} if the following holds, for all types $\tau \in \mathcal{T}$. For all $\msf{o} \in \msf{Op}_1(\tau)$, there is a $\Theta_M$-compatible function $\msf{sem}^{\msf{o}}_M: \msf{SEM}_M(T(\msf{o},1)) \rightarrow \msf{SEM}_M(\tau)$ such that, for all $\varphi_1 \in \msf{Fm}_\msf{L}(T(\msf{o},1))$:
	\begin{equation*}
		\msf{sem}_M(\msf{o}(\varphi_1)) = \msf{sem}^{\msf{o}}_M(\msf{sem}_M(\varphi_1)) \in \msf{SEM}_M(\tau)
	\end{equation*}
	
	In addition, for all $\msf{o} \in \msf{Op}_2(\tau)$, there is a $\Theta_M$-compatible function $\msf{sem}^{\msf{o}}_M: \msf{SEM}_M(T(\msf{o},1)) \times \msf{SEM}_M(T(\msf{o},2)) \rightarrow \msf{SEM}_M(\tau)$ such that, for all $(\varphi_1,\varphi_2) \in \msf{Fm}_\msf{L}(T(\msf{o},1)) \times \msf{Fm}_\msf{L}(T(\msf{o},2))$:
	\begin{equation*}
		\msf{sem}_M(\msf{o}(\varphi_1,\varphi_2)) = \msf{sem}^{\msf{o}}_M(\msf{sem}_M(\varphi_1),\msf{sem}_M(\varphi_2)) \in \msf{SEM}_M(\tau) 
	\end{equation*}
\end{definition}

\begin{lemma}[Proof~\ref{proof:lem_modal_logic_pair_inductive_property}]
	\label{lem:modal_logic_pair_inductive_property}
	For all $K \in \mathcal{K}(\prop,\act)$, the $(\ML(\prop,\act),K)$-pair $\Theta_K$ from Example~\ref{ex:Kripke_structures_SEM} satisfies the inductive property.
\end{lemma}

In the following, we will focus on those classes of $\msf{L}$-models $\mathcal{C}$ for which, for all models $M \in \mathcal{C}$, there are $(\msf{L},M)$-pairs satisfying the two above properties.
\begin{definition}
	\label{def:finitely_inductively_capture_semantics}
	Consider a logic $\msf{L}$ and an $\msf{L}$-model $M$. An $(\msf{L},M)$-pair $\Theta_M$ \emph{inductively capture the $\msf{L}$-semantics} if it captures the $\msf{L}$-semantics and satisfies the inductive property.
	%For all $\msf{L}$-models $M$, given any $(\msf{L},M)$-tuple $\Theta$ that finitely captures the $\msf{L}$-semantics, for all types $\tau \in \mathcal{T}$, we let $n_\Theta(\tau) := |\msf{SEM}(\tau)| \in \N$.
	%
	Given a class of $\msf{L}$-models $\mathcal{C}$, an $(\msf{L},\mathcal{C})$-pair $\Theta$ is such that: $\Theta = (\Theta_M)_{M \in \mathcal{C}}$ where, for all $M \in \mathcal{C}$, $\Theta_M= (\msf{SEM}_M,\msf{sem}_M)$ is an $(\msf{L},M)$-pair that inductively captures the $\msf{L}$-semantics.%; and $f = (f_\tau)_{\tau \in \mathcal{T}}$ is such that, for all types $\tau \in \mathcal{T}$, we have $f_\tau: \N \rightarrow \N$ such that, for all models $M \in \mathcal{C}$, we have $|\msf{SEM}_M(\tau)| \leq f_\tau(|M|)$. %We denote by $g_f: \N \to \N$ the function $g_f := \max_{\tau \in \mathcal{T}} f_\tau$.
\end{definition}

\subsection{Main theorem and proof sketch}
We can now state the main theorem of this paper.
\begin{theorem}[Proof~\ref{proof:thm_size_separation}]
	\label{thm:size_separation}
	Consider a logic $\msf{L}$, a class of $\msf{L}$-models $\mathcal{C}$ and an $(\msf{L},\mathcal{C})$-pair $\Theta$. For all fragments $\msf{L}'$ of the logic $\msf{L}$, a $\mathcal{C}$-sample $\mathcal{S}$ is $\msf{L}'$-separable if and only if there is an $\mathcal{S}$-separating $\msf{L}'$-formula of size at most: $n^{\mathcal{S}}_{\Theta} := \sum_{\tau \in \mathcal{T}} \prod_{M \in \mathcal{S}} |\msf{SEM}_M(\tau)|$.
	%\begin{equation*}
	%	n^{\mathcal{S}}_{\Theta} := \sum_{\tau \in \mathcal{T}} \prod_{M \in \mathcal{S}} |\msf{SEM}_M(\tau)|
	%\end{equation*}
\end{theorem}

%Note that in the appendix we actually show that this bound holds for a setting more general than the passive learning problem. In that setting, we are given a set of models $S$ and a set of subsets $\mathcal{R}_S \subseteq 2^Q$, and we are seeking formulas whose set of accepted $S$-models belongs to $\mathcal{R}_S$. 
%,  instead of looking for a formulas accepting all positive models and rejecting all negative ones, we seek formulas 

The core idea behind this theorem stems from a pigeonhole argument. Indeed, consider some $\mathcal{C}$-sample $\mathcal{S}$ and assume that an $\mathcal{S}$-separating $\msf{L}'$-formula $\varphi$ has two sub-formulas $\varphi_1$ and $\varphi_2$ of the same type that are mapped to the same semantic value in $\msf{SEM}_M$ by the function $\msf{sem}_M$, for all $M \in \mathcal{S}$. Then the formula $\varphi'$ obtained from $\varphi$ by replacing $\varphi_2$ by $\varphi_1$ --- which can be done since $\varphi_1$ and $\varphi_2$ are of the same type --- and the formula $\varphi$ are mapped to the same semantic value in $\msf{SEM}_M$, for all $M \in \mathcal{S}$. Thus, the formula $\varphi'$ is also $\mathcal{S}$-separating. %Furthermore, the formula $\varphi'$ is of smaller size than $\varphi$. 
By repeating this process, we can obtain an $\mathcal{S}$-separating formula in which there are no two sub-formulas of the same type that are mapped to the same semantic value in $\msf{SEM}_M$, for all $M \in \mathcal{S}$. The bound of Theorem~\ref{thm:size_separation} then follows from the definition of formula size.
%Thus, if there is a separating formula, there is one for which each sub-formula is mapped to a different tuple of semantic values. The formalization of this idea would prove Theorem~\ref{thm:size_separation}, but we venture into a different direction to obtain a proof (sketch) from which we can derive an enumeration algorithm.

The proof (sketch) that we provide below of Theorem~\ref{thm:size_separation} actually ventures into a different direction than the pigeonhole argument presented above. Although this proof (sketch) is slightly more complicated than the pigeonhole argument, it additionally allows to derive a semantic-based enumeration algorithm solving the passive learning problem.
\begin{proof}[Proof sketch]
	Let us consider a logic $\msf{L}$, a class of $\msf{L}$-models $\mathcal{C}$ and an $(\msf{L},\mathcal{C})$-pair $\Theta$. Let $\mathcal{S} = (\mathcal{P},\mathcal{N})$ be a $\mathcal{C}$-sample and $\msf{L}'$ be fragment of the logic $\msf{L}$. To simplify the explanations, we assume here that there is a single type of $\msf{L}$-formulas: $|\mathcal{T}| = 1$ (thus, there is also a single type of $\msf{L}'$-formulas), and that there are no arity-1 $\msf{L}'$-operators: $\msf{Op}_1' = \emptyset$. 
	
	Let us first handle the case where there is a single (positive) $\msf{L}$-model $M$ in $\mathcal{S}$, i.e. $\mathcal{P} = \{M\}$ and $\mathcal{N} = \emptyset$. Our goal is to find an $\msf{L}'$-formula satisfied by this model $M$. %There are two ideas involved in this proof. They both come from our assumption that the $(\msf{L},M)$-pair $\Theta_M = (\msf{SEM}_M,\msf{sem}_M)$ inductively captures the $\msf{L}$-semantics. 
	By assumption, the $(\msf{L},M)$-pair $\Theta_M = (\msf{SEM}_M,\msf{sem}_M)$ both a) captures the $\msf{L}$-semantics and b) satisfies the inductive property. Property a) gives that any final $\msf{L}'$-formula $\varphi \in \msf{Fm}^{\msf{f}}_{\msf{L}'}$ is satisfied by $M$ if and only if $\msf{sem}_{M}(\varphi) \in \msf{SAT}_{M}$. Our goal is thus to find an $\msf{L}'$-formula mapped in $\msf{SAT}_{M}$% by the function $\msf{sem}_{M}$
	. 
	
	To find such a formula, we are going to compute the subset $\msf{sem}_M[\msf{Fm}^{\msf{f}}_{\msf{L}'}] \subseteq \msf{SEM}_M$ of all semantics values in $\msf{SEM}_M$ that $\msf{L}'$-formulas can be mapped to by the function $\msf{sem}_M$. Thanks to Property b), this set can actually be computed inductively. Initially, we set $\msf{SEM}_{M,0}^{\msf{L}'}$ to be the subset of all elements of $\msf{SEM}_M$ that arity-0 $\msf{L}'$-operators can be mapped to by the function $\msf{sem}_M$, formally: $\msf{SEM}_{M,0}^{\msf{L}'} := \{ \msf{sem}_M(\msf{o}) \mid \msf{o} \in \msf{Op}_0' \} \subseteq \msf{SEM}_M$. Then, at step $i \in \N$, we go through all arity-2 $\msf{L}'$-operators $\msf{o} \in \msf{Op}_2'$ and, for all pairs $(\msf{S}^1,\msf{S}^2) \in (\msf{SEM}_{M,i}^{\msf{L}'})^2$, we add the semantical value $\msf{sem}_{M}^{\msf{o}}(S^1,S^2) \in \msf{SEM}_M$ to $\msf{SEM}_{M,i+1}^{\msf{L}'} \supseteq \msf{SEM}_{M,i}^{\msf{L}'}$. Note that the function $\msf{sem}_{M}^{\msf{o}}$ is a $\Theta_M$-compatible function. The process then stops once we reach a fixed point, i.e. when $\msf{SEM}_{M,i+1}^{\msf{L}'} = \msf{SEM}_{M,i}^{\msf{L}'}$ for some $i \in \N$. Then, we let $\msf{SEM}_{M}^{\msf{L}'} := \msf{SEM}_{M,i}^{\msf{L}'}$, and we claim that $\msf{SEM}_{M}^{\msf{L}'} = \msf{sem}_M^{\msf{L}'}[\msf{Fm}^{\msf{f}}_{\msf{L}'}]$. This equality can be proved by a double-inclusion: for the right-to-left inclusion, we show by induction on $\varphi \in \msf{Fm}^{\msf{f}}_{\msf{L}'}$ that $\msf{sem}_M^{\msf{L}'}(\varphi) \in \msf{SEM}_{M}^{\msf{L}'}$. For the left-to-right inclusion, we
	\iffalse
	we show by induction on $i \in \N$ that $\msf{SEM}_{M,i}^{\msf{L}'} \subseteq \msf{sem}_M^{\msf{L}'}[\msf{Fm}^{\msf{f}}_{\msf{L}'}]$; 
	
	Theorem~\ref{thm:size_separation} exhibits an upper bound on the minimal size of separating formulas. What we have proved above gives us that if there is an $\mathcal{S}$-separating $\msf{L}'$-formula $\varphi$, then $\msf{SEM}_{M}^{\msf{L}'} \cap \msf{SAT}_M \neq \emptyset$. This cannot be used directly to prove Theorem~\ref{thm:size_separation}. However, when establishing the above left-to-right inclusion, we can actually also 
	\fi
	show that, for all semantic values $X \in \msf{SEM}_{M}^{\msf{L}'}$, there is a $\msf{L}'$-formula $\varphi_X$ satisfying $\msf{sem}_M(\varphi_X) = X$ %(which means that $X \in \msf{sem}_M[\msf{Fm}^{\msf{f}}_{\msf{L}'}]$, as we have already established above) 
	while ensuring that all of its sub-formulas (in $\msf{Sub}(\varphi_X)$) are mapped to a different value in $\msf{SEM}_M$ by the function $\msf{sem}_M$. This implies that $\msf{sz}(\varphi_X) = |\msf{Sub}(\varphi_X)| \leq |\msf{SEM}_M|$. Overall, if there is a $\mathcal{S}$-separating $\msf{L}'$-formula, then there is some $X \in \msf{SEM}_{M}^{\msf{L}'} \cap \msf{SAT}_M$. 
	Considering a $\msf{L}'$-formula $\varphi_X$ satisfying the two above conditions, we obtain that: $\msf{sz}(\varphi_X) \leq |\msf{SEM}_M|$ and $\msf{sem}_M(\varphi_X) = X \in \msf{SAT}_M$, thus $M \models \varphi_X$. %This is what we wanted to establish. 
	
	Consider now the separation problem in its full generality where the sample $\mathcal{S}$ does not only consist of a single positive model. We follow a similar procedure in this case, except that we manipulate subsets of tuples in $\prod_{M \in \mathcal{S}} \msf{SEM}_M$. As above, we use a fixed-point procedure to obtain the set $\msf{SEM}^{\msf{L}'}_{\mathcal{S}} \subseteq \prod_{M \in \mathcal{S}} \msf{SEM}_M$. Then, a $\msf{L}'$-formula is $\mathcal{S}$-separating if and only if it is mapped by the function $(\msf{sem}_M)_{M \in \mathcal{S}}$ to a tuple $X \in \msf{SEM}^{\msf{L}'}_{\mathcal{S}}$ such that, for all $M \in \mathcal{P}$, we have $X[M] \in \msf{SAT}_M$ and for all $M \in \mathcal{N}$, we have $X[M] \in \msf{SEM}_M \setminus \msf{SAT}_M$. Furthermore, as above, we can show that, for all $X \in \msf{SEM}^{\msf{L}'}_{\mathcal{S}}$, there is a $\msf{L}'$-formula $\varphi_X$ such that, for all $M \in \mathcal{S}$ we have $\msf{sem}_M(\varphi_X) = X[M]$ and all sub-formulas of $\varphi_X$ are mapped to a different tuple in $\msf{SEM}^{\msf{L}'}_{\mathcal{S}} \subseteq \prod_{M \in \mathcal{S}} \msf{SEM}_M^{\msf{L}}$ by the function $(\msf{sem}_M)_{M \in \mathcal{S}}$. In turn, the size of such a formula is bounded from above by $|\prod_{M \in \mathcal{S}} \msf{SEM}_M| = \prod_{M \in \mathcal{S}} |\msf{SEM}_M|$. Theorem~\ref{thm:size_separation} follows. 
\end{proof}

\subsection{A semantic-based meta algorithm}
\label{subsec:algo}

Many enumeration algorithms in the literature use sophisticated techniques to avoid generating syntactically different but semantically identical formulas. However, this proves to be a difficult task as deciding formula equivalence is hard. Thus, these algorithms often resort to using heuristics which do not entirely prevent enumerating semantically identical formulas. 

%This is a very hard thing to do, deciding equivalence is hard. Do not really check in a complete (computationally really expensive), some heuristics, typically cannot prevent enumerating equivalent formulas.

The above proof sketch suggests a meta algorithm that circumvents this difficulty by not enumerating formulas syntactically, but semantically. It consists in enumerating over the possible semantic values of the formulas, instead of the formulas themselves. This meta algorithm is described as Algorithm~\ref{algo:decide_passive_learning} in pseudo-code: in Line 1-7, the set $\msf{SEM}_{\msf{S}}^{\msf{L}}$ is computed via a fixed point computation; in Line 8-10, it is checked that there is some $X \in \msf{SEM}_{\msf{S}}^{\msf{L}}$ such that for all $M \in \mathcal{P}$, we have $X[M] \in \msf{SAT}_M$ and for all $M \in \mathcal{N}$, we have $X[M] \in \msf{SEM}_M \setminus \msf{SAT}_M$. From the proof (sketch), we immediately obtain the following result.

\begin{algorithm}[t]
	\caption{$\msf{PassiveLearning}_{\msf{L},\Theta}$: Decides if an input $\mathcal{C}$-sample $\mathcal{S}$ is $\msf{L}$-separable}
	\label{algo:decide_passive_learning}
	\textbf{Input}: An $\mathcal{C}$-sample $\mathcal{S}$ of models
	\begin{algorithmic}[1]
		\State $\msf{SEM} \gets \emptyset$, $\msf{SEM}' \gets \{ ((\msf{sem}_M(\msf{o}))_{M \in \mathcal{S}},\tau) \mid \tau \in \mathcal{T},\; \msf{o} \in \msf{Op}_0(\tau) \}$
		\While{$\msf{SEM} \neq \msf{SEM}'$}
		\State $\msf{SEM} \gets \msf{SEM}'$
		\For{$\tau \in \mathcal{T},\; \msf{o} \in \msf{Op}_1(\tau),\; \tau_1 \in T(\msf{o},1),\; (X,\tau_1) \in \msf{SEM}'$}
		\State $\msf{SEM}' \gets \msf{SEM}' \cup \{ ((\msf{sem}^{\msf{o}}_M(X[M]),\tau)_{M \in \mathcal{S}},\tau) \}$
		\EndFor 
		\For{$\tau \in \mathcal{T},\; \msf{o} \in \msf{Op}_2(\tau),\; (\tau_1,\tau_2) \in T(\msf{o},1) \times T(\msf{o},2),\; ((X_1,\tau_1),(X_2,\tau_2)) \in (\msf{SEM}')^2$}
		\State $\msf{SEM}' \gets \msf{SEM}' \cup \{ ((\msf{sem}^{\msf{o}}_M(X_1[M],X_2[M]))_{M \in \mathcal{S}},\tau) \}$
		\EndFor 
		\EndWhile
		\For{$\tau \in \mathcal{T}_{\msf{f}},\; (X,\tau) \in \msf{SEM}$}
		\If{$X \in \prod_{M \in \mathcal{P}} \msf{SAT}_M \times \prod_{M \in \mathcal{N}} (\msf{SEM}_M \setminus \msf{SAT}_M)\;$}
		\Return Accept
		\EndIf
		\EndFor 
		\State
		\Return Reject
	\end{algorithmic}
\end{algorithm}

\begin{theorem}[Proof~\ref{proof:thm_size_separation}]
	\label{thm:enumeration_algo}		
	Consider a logic $\msf{L}$, a class of $\msf{L}$-models $\mathcal{C}$ and an $(\msf{L},\mathcal{C})$-pair $\Theta$. Algorithm~\ref{algo:decide_passive_learning} decides the passive learning problem $\msf{PvLn}(\msf{L},\mathcal{C})$. 
\end{theorem}
%·We prove this Theorem~\ref{thm:complexity_separation} on the fly as we prove Theorem~\ref{thm:size_separation} in Appendix~\ref{proof:thm_size_separation}.

\textbf{Complexity of Algorithm~\ref{algo:decide_passive_learning}.} This meta algorithm is described on an abstract logical framework which can be instantiated with concrete logics. The complexity of the obtained concrete algorithms depends on the upper bound $n^{\mathcal{S}}_{\Theta}$, and on the complexity of a) computing the output of $\Theta_M$-compatible functions $\msf{sem}_M^{\msf{o}}$; b) deciding if a semantic value in $\msf{SEM}_M$ is in $\msf{SAT}_M$; and if the set of operators $\msf{Op}$ is infinite c) computing a finite subset of \textquotedblleft{}relevant operators\textquotedblright{} that is sufficient to range over in Lines 1, 4 and 6. For each individual logic, the bound $n^{\mathcal{S}}_{\Theta}$ is different and the operations a), b), and c) have different complexities. For our running example of modal logic formulas evaluated on Kripke structures, operations a), b), and c) can be done in polynomial time, while the bound $n^{\mathcal{S}}_{\Theta}$ is exponential. Thus, we obtain an exponential time algorithm (see Proposition~\ref{prop:modal_logic_complexity}).

\textbf{Usefulness of Algorithm~\ref{algo:decide_passive_learning}.} %One of the common pitfalls of the enumeration algorithms of the literature, that clever techniques attempt to avoid as much as possible, is to generate syntactically different but semantically equivalent formulas. 
This meta algorithm %bypasses this issue by design as we 
avoids generating syntactically different but semantically equivalent formulas by design as we shift paradigm from enumerating formulas to enumerating semantic values. This change of paradigm may induce a crucial difference complexity-wise. As mentioned above, for our running example, we have an exponential time algorithm. On the other hand, in Section~\ref{sec:lower_bound}, we will exhibit modal logic fragments for which the minimal size of a separating formula is exponential. Thus, a formula-enumeration algorithm would have, in the worst case, a doubly-exponential complexity (as there are doubly exponentially many formulas of exponential size).  

This semantic enumeration algorithm is not novel as some practical papers have implicitly used exactly this approach, even though they did not describe it as such, e.g. \cite{DBLP:journals/pacmpl/ValizadehB23} with regular expressions, or \cite{DBLP:conf/cav/ValizadehFB24} with $\LTL$-formulas (we will discuss again this paper in Section~\ref{subsec:temporal}). We believe that this meta algorithm is best understood as a framework that can help the design of efficient enumeration algorithms for a wide range of logic learning problems, beyond the instantiations already present in the literature.

\section{Use cases}
\label{sec:examples}
The goal of this section is to demonstrate that our abstract formalism is widely applicable. Thus, we instantiate it on a zoo of concrete formalisms, on which we can apply Theorem~\ref{thm:size_separation} to obtain a (exponential) bound on the minimal size of separating formulas in the passive learning problem. We also investigate a case (namely, $\LTL$-formulas evaluated on Kripke structures) where the assumptions of Theorem~\ref{thm:size_separation} are not met on the full logic. In this context, we explore how these assumptions can guide us towards exhibiting a well-behaving related logic. We also show that our abstract formalism is applicable to non-logical settings with words separating automata. 

\subsection{Modal Logic}
\label{subsec:ml}
%Let us start with modal logic since most of the work is already done. %in Example~\ref{ex:Kripke_structures_SEM} (for the definition of $\ML(\prop,\act)$-pairs), and in Lemmas~\ref{lem:modal_logic_pair_captures_semantics} and~\ref{lem:modal_logic_pair_inductive_property} (to prove that these pairs inductively capture the $\ML(\prop,\act)$-semantics). 

%For this case only, our goal is to apply not only Theorem~\ref{thm:size_separation}, but also Theorem~\ref{thm:complexity_separation}. We have already defined the model logic syntax (in Example~\ref{ex:modal_logic}), described the models, i.e. Kripke strictures in $\mathcal{K}(\prop,\act)$ (in Example~\ref{ex:Kripke_structures}), and given $(\ML(\prop,\act),K)$-pairs (in Example~\ref{ex:Kripke_structures_SEM}) for all Kripke structures $K \in \mathcal{K}(\prop,\act)$. There remains to describe when a modal logic formula is satisfied by a Kripke structure, i.e. to give the modal logic semantics.
\begin{corollary}
	\label{coro:modal_logic}
	Consider a non-empty set of propositions $\prop$, a non-empty set of actions $\act$, and any fragment $\msf{L}$ of the modal logic $\ML(\prop,\act)$. For all $\mathcal{K}(\prop,\act)$-samples $\mathcal{S}$, if there is a $\mathcal{S}$-separating $\msf{L}$-formula, there is one of size at most $2^{n}$, with $n := \sum_{K \in \mathcal{S}} |Q_K|$.
\end{corollary}
\begin{proof}
	For all Kripke structures $K \in \mathcal{K}(\prop,\act)$, the $\ML(\prop,\act)$-pair $\Theta_K$ defined in Example~\ref{ex:Kripke_structures_SEM} inductively captures the $\ML(\prop,\act)$-semantics by Lemmas~\ref{lem:modal_logic_pair_captures_semantics} and~\ref{lem:modal_logic_pair_inductive_property}. The result then follows directly from Theorem~\ref{thm:size_separation}, since for all $\mathcal{K}(\prop,\act)$-samples $\mathcal{S}$, we have $\prod_{K \in \mathcal{S}} |\msf{SEM}_K| = \prod_{K \in \mathcal{S}} 2^{|Q_K|} = 2^{\sum_{K \in \mathcal{S}} |Q_K|}$
\end{proof}
For some modal logic fragments, we will provide in Section~\ref{sec:lower_bound} a family of samples for which an exponential bound is asymptotically optimal. Furthermore, the instantiation of Algorithm~\ref{algo:decide_passive_learning} to this context gives an exponential time algorithm.% Theorem~\ref{thm:complexity_separation} also applies to modal logic formulas evaluated in Kripke structures, as stated in the corollary below.
\begin{proposition}[Proof~\ref{proof:prop_modal_logic_complexity}]
	\label{prop:modal_logic_complexity}
	Consider a non-empty set of propositions $\prop$ and a non-empty set of actions $\act$. Algorithm~\ref{algo:decide_passive_learning} instantiated $\ML(\prop,\act)$-formulas and $\mathcal{K}(\prop,\act)$-sample %, with the $(\ML(\prop,\act),\mathcal{K}(\prop,\act))$-tuple $\Theta$ from Example~\ref{ex:Kripke_structures_SEM} 
	has complexity $2^{O(n)}$, where $n := \sum_{K \in \mathcal{S}} |Q_K|$.
\end{proposition}
\iffalse
\begin{proof}
	Consider a Kripke structure $K \in \mathcal{K}(\prop,\act)$. For all operators $\msf{o}$, there is $\Theta_K$-compatible functions $\msf{sem}_K^{\msf{o}}$ (e.g. the functions provided in the proof of Lemma~\ref{lem:modal_logic_pair_inductive_property}) that satisfies the assumptions 1. and 2. of Theorem~\ref{thm:complexity_separation}, while the set $\msf{SEM}_K$ is such that $|\msf{SEM}_K| = |2^{Q_K}| \leq 2^{|K|}$. Furthermore, for any fragment $\msf{L}'$ of $\ML(\prop,\act)$ and $\mathcal{K}(\prop,\act)$-sample $\mathcal{S}$, the set $\msf{RelOp}(\msf{L}',\mathcal{K}(\prop,\act)) := \{ \neg,\wedge,\lor \} \cup \msf{Prop} \cup \{ [a] \mid a \in \act \} \cup \{ \langle a \rangle^{\geq k} \mid a \in \act, k \leq \max_{K \in \mathcal{S}} |Q_K| \} \cap \msf{Op}'$ satisfies both items of assumption 4. of Theorem~\ref{thm:complexity_separation}, and is computable in polynomial time. Thus, Theorem~\ref{thm:complexity_separation} can be applied, which gives the result.
\end{proof}
\fi

\subsection{Temporal logics}
\label{subsec:temporal}
Let us now focus on temporal logics and, more specifically, on the Linear Temporal Logic ($\LTL$) \cite{DBLP:conf/focs/Pnueli77}. Before that, let us introduce some notations on sequences (finite or infinite)% that will be particularly useful in the following
. %We will see that applying Theorem~\ref{thm:size_separation} in this case is straightforward. We then tackle three extensions. First, the logic CTL (Compuation Tree Logic), which is evaluated on Kripke structure. Applying Theorem~\ref{thm:size_separation} is also straightforward, and very similar to the case of modal logic. We then consider LTL formulas themselves evaluated on Kripke structures. This case is actually quite tricky, and, a priori, Theorem~\ref{thm:size_separation} cannot be applied to the full logic. Thus, we seek fragments that are expressive enough to be relevant and that for which we can recover some kind of inductive behavior. Finally, we handle the logic PSL (Property specification language), still evaluated on words, which enhances LTL with trigger operators described by regular expression. For this case applying Theorem~\ref{thm:size_separation} leads to a super-exponential bound.
Consider a non-empty set $Q$. We let $Q^*$, $Q^+$, and $Q^\omega$ denote the set of finite, non-empty finite, and infinite sequences of elements of $Q$, respectively. For all $\rho \in Q^* \cup Q^\omega$, we let $|\rho| \in \N \cup \infty$ denote the number of elements of $\rho$, and for all $i < |\rho|$, we let $\rho[i] \in Q$ denote the element at position $i$ in $\rho$, $\rho[i:] \in Q^* \cup Q^\omega$ denote the suffix of $\rho$ starting at position $i$, and $\rho[:i] \in Q^+$ denote the finite suffix of $\rho$ ending at position $i$. For all $\rho \in Q^+$, we let $\head{\rho} \in Q$ denote the last element of $\rho$, and $\body{\rho} \in Q^*$ denote the sequence $\rho$ without %its last element 
$\head{\rho}$.

\subsubsection{$\LTL$-formulas evaluated on words}
%Let us start with $\LTL$-formulas evaluated on infinite, ultimately periodic, words. 
$\LTL$-formulas may use different temporal operators $\lX \in \msf{Op}_1$ (neXt), $\lF \in \msf{Op}_1$ (Future), $\lG \in \msf{Op}_1$ (Globally), $\lU \in \msf{Op}_2$ (Until) to express properties about future events. We first consider the case where these formulas are evaluated on finite or ultimately periodic words $w$ (or \textquotedblleft{}lasso\textquotedblright{})\footnote{Usually, $\LTL$-formulas are evaluated either only on finite words or only on infinite words. We consider both cases at the same time as it does not change our underlying argument}. On such models, the semantics of the above operators can be informally described as follows: the formula $\lX \varphi$ expresses the fact that the formula $\varphi$ should hold in the next position (which never holds in the last position of a finite word), the formula $\lF \varphi$ (resp. $\lG \varphi$) means that the formula $\varphi$ eventually (resp always) holds, while the formula $\varphi_1 \lU \varphi_2$ means that the formula $\varphi_2$ eventually holds, and until then, the formula $\varphi_1$ holds. We define formally the syntax, models, and semantics of the $\LTL$-logic below. %(Yesterday, once, historically, since)
\begin{definition}[LTL syntax and semantics]
	\label{def:ltl}
	Consider a non-empty finite set of propositions $\prop$. The $\LTL(\prop)$-syntax is as follows, with $p\in \prop$:
	\begin{equation*}
		\varphi ::= p \mid \neg \varphi \mid \varphi \vee \varphi \mid \varphi \wedge \varphi \mid \lX \varphi \mid \lF \varphi \mid \lG \varphi \mid \varphi \lU \varphi %\mid \lY \varphi \mid \lO \varphi \mid \lH \varphi \mid \varphi \lS \varphi
	\end{equation*}
	The $\LTL(\prop)$-models $\mathcal{W}(\prop)$ are the finite (non-empty) and ultimately periodic words whose letters are subsets of propositions in $\prop$. Formally, we have: $\mathcal{W}(\prop) := \{ u \cdot v^\omega \mid u,v \in (2^\prop)^*,\; u \cdot v \in (2^\prop)^+ \}$%(note that $v$ may be equal to the empty word $\epsilon$, in which case $u$ necessarily is not)
	. %The size $|w|$ of a word $w = u \cdot v^\omega \in \mathcal{IW}(\prop)$ is equal to $|w| := |u| + |v|$. In particular, if $w = \epsilon$, then we have $w = u \in (2^\prop)^+$ a finite word with $|w| = |u|$. 
	Given a word $w \in \mathcal{W}(\prop)$, we define when $\msf{LTL}(\prop)$-formulas are satisfied by $w$ inductively as follows: 
	
	\hspace{0cm}
	\begin{minipage}[t]{0.2\textwidth}
		\begin{align*}
			& w \models p & \text{ iff } & p \in w[0] \\
			& w \models \neg \varphi & \text{ iff } & w \not\models \varphi \\
			& w \models \varphi_1 \vee \varphi_2 & \text{ iff } & w \models \varphi_1 \text{ or } w \models \varphi_2\\
			& w \models \varphi_1 \wedge \varphi_2 & \text{ iff } & w \models \varphi_1 \text{ and } w \models \varphi_2
			%& w,i \models_{\msf{pos}} \lY \varphi & \text{ iff } & i > 0,\; \text{and } w,i-1 \models_{\msf{pos}} \varphi \\
		\end{align*}
	\end{minipage}
	\begin{minipage}[t]{0.05\textwidth}
		\begin{align*}
			\mid \\
			\mid \\
			\mid \\
			\mid \\
			\mid
		\end{align*}
	\end{minipage}
	\begin{minipage}[t]{0.5\textwidth}
		\begin{align*}
			& w \models \lX \varphi & \text{ iff } & |w| \geq 2 \text{ and }w[1:] \models \varphi \\
			& w \models \lF \varphi & \text{ iff } & \exists j < |w|,\; w[j:] \models \varphi \\
			%& w,i \models_{\msf{pos}} \lO \varphi & \text{ iff } & \exists j \leq i,\; w,j \models_{\msf{pos}} \varphi \\
			& w \models \lG \varphi & \text{ iff } & \forall j < |w|,\; w[j:] \models \varphi \\
			%& w,i \models_{\msf{pos}} \lH \varphi & \text{ iff } & \forall j \leq i,\; w,j \models \varphi \\
			& w \models \varphi_1 \lU \varphi_2 & \text{ iff } & \exists j < |w|,\; w[j:] \models \varphi_2,\; \\
			& & & \forall 0 \leq k \leq j-1,\; w[k:] \models \varphi_1
			%& w,i \models_{\msf{pos}} \varphi_1 \lS \varphi_2 & \text{ iff } & \exists j \leq i,\; w,j \models_{\msf{pos}} \varphi_2,\; \\
			%& & & \forall j+1 \leq k \leq i,\; w,k \models_{\msf{pos}} \varphi_1\\
		\end{align*}
	\end{minipage}
\end{definition}

With Theorem~\ref{thm:size_separation}, we obtain an exponential bound on the minimal size of a separating formula in the passive learning problem, as formally stated below.
\begin{corollary}[Proof~\ref{proof:corollary_ltl}]
	\label{coro:ltl_logic}
	Consider a non-empty set of propositions $\prop$, and any $\LTL$-fragment $\msf{L}$ of $\msf{LTL}(\prop)$. For all $\mathcal{W}(\prop)$-samples $\mathcal{S}$, if there is an $\mathcal{S}$-separating $\msf{L}$-formula, there is one of size at most $2^{n}$, with $n := \sum_{w \in \mathcal{S}} |w|$.
\end{corollary}
\begin{proof}[Proof sketch]
	Consider a set of propositions $\prop$. Consider some $w = u \cdot v^\omega \in \mathcal{W}(\prop)$. We let $||w|| := |u| + |v|$. In particular, if $v = \epsilon$, then we have $w = u \in (2^\prop)^+$ a finite word with $||w|| = |u|$. Then, we consider the $(\msf{LTL}(\prop),w)$-pair $(\msf{SEM}_w,\msf{sem}_w)$ where $\msf{SEM}_w := 2^{\msf{Pos}(w)}$ with $\msf{Pos}(w) := \{0,\ldots,||w||-1\}$ and $\msf{sem}_w: \msf{Fm}_{\msf{LTL}(\prop)} \to \msf{SEM}_w$ is such that, for all $\varphi \in \msf{Fm}_{\msf{LTL}(\prop)}$, we have $\msf{sem}_w(\varphi) := \{ i \in \msf{Pos}(w) \mid w[i:] \models \varphi \}$. Such a pair straightforwardly captures the $\LTL(\prop)$-semantics. It also satisfies the inductive property. There are two main reasons for that: first, as can be seen in the semantic definition above, the positions at which a formula holds entirely depends on the positions at which sub-formulas hold; second, even if $w = u \cdot v^\omega \in (2^\prop)^\omega$ is an ultimately periodic word, it is enough to only track positions in $\{0,\ldots,||w||-1\}$ instead of all $\{0,\ldots,|w|-1\}  = \N$ because, for all $i,j \in \N$, the infinite words $w[|u|+i:] \in (2^\prop)^\omega$ and $w[|u|+i + j \cdot |v|:] \in (2^\prop)^\omega$ are equal. In fact, the $(\msf{LTL}(\prop),w)$-pair $(\msf{SEM}_w,\msf{sem}_w)$ inductively captures the $\LTL(\prop)$-semantics. Thus, Corollary~\ref{coro:ltl_logic} follows directly from Theorem~\ref{thm:size_separation}.
\end{proof}
\begin{remark}
	\label{rem:LTL}
	When considering the full logic $\msf{LTL}(\prop)$ evaluated on ultimately periodic words, there is actually a polynomial upper bound on the minimal size of separating formulas. However, the exponential upper bound established above holds for all fragments of $\msf{LTL}(\prop)$, including some for which we establish a sub-exponential lower bound in Section~\ref{sec:lower_bound}.
\end{remark}

\begin{remark}
	\label{rem:LTL_implem}
	As mentioned when discussing the meta algorithm, in the practical paper \cite{DBLP:conf/cav/ValizadehFB24} focusing on the passive learning problem of $\LTL$-formulas evaluated on finite words, the search space is composed of a set of characteristic matrices representing the evaluation of $\LTL$ formulas at each word position. The enumeration of these characteristic matrices involves directly applying operators to them. For instance, propositional operators correspond to bitwise operations, the $\lX$ operator shifts matrices one bit to the left, the $\lF$ operator (eventually) performs a disjunction over leftward shifts, and so on. This corresponds to what we describe in the proof of Corollary~\ref{coro:ltl_logic} in Appendix~\ref{proof:corollary_ltl}.
\end{remark}
%As we will discuss in Section~\ref{sec:lower_bound}, for the full logic $\msf{LTL}(\prop)$, there is actually a polynomial lower bound on the minimal size of separating formulas. However, we

\subsection{$\LTL$ formulas evaluated on Kripke structures}
\label{subsec:LTL_case_study}
%TODO: $\CTL^*$, $\ATL^*$
%In all of the use cases that we have considered so far, the application of Theorem~\ref{thm:size_separation} is straightforward. Let us now consider a logical formalism where it is not the case. Specifically, we 
Let us now consider the case of $\LTL$-formulas evaluated on actionless Kripke structures. In this setting, an $\LTL$-formula accepts a Kripke structure if all of the infinite paths that could occur from an initial state satisfy the formula. As we will see below, the universal quantification over (infinite) paths makes the application of Theorem~\ref{thm:size_separation} much trickier. 

Let us formally define this semantics below.
\begin{definition}[$\LTL$ semantics on Kripke structures]
	Consider a non-empty finite set of propositions $\prop$. The models that we consider are the actionless Kripke structures, i.e. Kripke structures $K = (Q,I,A,\delta,P,\pi)$ with $|A| = 1$. Actionless Kripke structures are in fact described with a simpler tuple $K = (Q,I,\delta,P,\pi)$ where $\delta: Q \to 2^Q$. In such Kripke structures, for all states $q \in Q$, we let $\msf{Path}(q) := \{ \rho \in q \cdot Q^\omega \mid \forall i \in \N, \rho[i+1] \in \delta(\rho[i]) \}$. We will also restrict ourselves to non-blocking Kripke structures, i.e. such that all states have at least one successor. Thus, we consider a set of models:
	\begin{equation*}
		\mathcal{T}(\prop) := \{ K = (Q,I,\delta,P,\pi) \mid P \subseteq \prop,\; \forall q \in Q,\; \delta(q) \neq \emptyset \}
	\end{equation*}
	
	Consider an actionless non-blocking Kripke structure $K = (Q,I,\delta,P,\pi) \in \mathcal{T}(\prop)$. An $\LTL$-formula $\varphi$ accepts a state $q \in Q$, denoted $q \models_{\msf{s}} \varphi$, if for all $\rho \in \msf{Paths}(q)$, we have $\rho \models \varphi$. Then, the $\LTL$-formula $\varphi$ accepts $K \in \mathcal{T}(\prop)$, if for all $q \in I$, we have $q \models_{\msf{s}} \varphi$.
\end{definition}

As a first attempt, we may try to use the same semantic values and semantic function that proved successful with $\ML$-formulas. This is discussed in the example below.
\begin{example}
	\label{ex:discussion_sub_logic}
	Consider the set of propositions $\prop := \{a,b\}$ and the Kripke structure $K = (Q,I,\delta,P,\pi) \in \mathcal{T}(\prop)$ depicted in Figure~\ref{fig:Kripke_example}. Let us consider as semantic values $\msf{SEM}_K := 2^{Q}$ and as semantic function $\msf{sem}_K: \msf{Fm}_{\LTL(\prop)} \to \msf{SEM}_K$ such that, for all $\varphi \in \msf{Fm}_{\LTL(\prop)}$, we have $\msf{sem}_K(\varphi) := \{ q \in Q \mid q \models_{\msf{s}} \varphi \}$. The $(\LTL(\prop),\mathcal{T}(\prop))$-pair $(\msf{SEM}_K,\msf{sem}_K)$ clearly captures the $\LTL(\prop)$-semantics. However, it does not satisfy the inductive property. More precisely, %the logical operator $\wedge$ and the temporal operators $\lX$ and $\lG$ are compatible with the universal quantification over paths, and therefore 
	the $(\LTL(\prop),\mathcal{T}(\prop))$-pair $(\msf{SEM}_K,\msf{sem}_K)$ satisfies the inductive property w.r.t. the logical operator $\wedge$ and the temporal operators $\lX$ and $\lG$% (we will formally argue this below)
	; however, it is not the case of the propositional operators $\neg,\lor$ and of the temporal operators $\lF,\lU$. Indeed, we have $\msf{sem}_K(\lX a) = \{q_2\} = \msf{sem}_K(b)$, while: $\msf{sem}_K(\neg \lX a) = \emptyset \neq \{ q_1,q_3 \} = \msf{sem}_K(\neg b)$; $\msf{sem}_K(\lX a \vee \lX b) = Q \neq \{q_2\} = \msf{sem}_K(b \vee \lX b)$; and $\msf{sem}_K(\lF \lX a) = Q \neq \{q_2\} = \msf{sem}_K(\lF b)$.
\end{example}

\begin{figure}
	\centering
	\resizebox{0.6\linewidth}{!}{
		\includegraphics{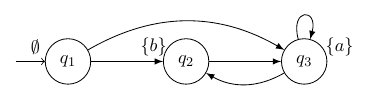}}
	\caption{A depiction of a Kripke structure $K$ where the labeling function $\pi$ is described by the of propositions next to every state.}
	\label{fig:Kripke_example}
\end{figure}

In fact, as we discuss in Appendix~\ref{appen:discussion_inexistence}, it is rather straightforward to show that there cannot exist an $(\LTL(\{a,b\}),K)$-pair (for $K$ the Kripke structure of Figure~\ref{fig:Kripke_example}) that inductively captures the $\LTL$-semantics. (This is due to the fact that, by assumption, an $(\LTL(\{a,b\}),K)$-pair has finitely-many semantic values.) To circumvent this issue, we may restrict the use of the operators $\neg,\vee,\lF,\lU$ with simple-enough formulas, specifically formulas that do not use temporal operators. This defines a new logic described below. 
\begin{definition}[$\LTL_P$ syntax and semantics]
	For a non-empty finite set of propositions $\prop$, the $\LTL_P(\prop)$-syntax is as follows, with $p\in \prop$ (both types of formulas are final):
	\begin{align*}
		\varphi_P & ::=  p \mid \neg \varphi_P \mid \varphi_P \wedge \varphi_P \mid \varphi_P \vee \varphi_P \\
		\varphi & ::= \varphi_P \mid \varphi \wedge \varphi \mid  \varphi_P \vee \varphi \mid \lX \varphi \mid \lG \varphi \mid \lF \varphi_P \mid \varphi \lU \varphi_P
	\end{align*}
	Although $\LTL_P(\prop)$ is not a fragment of the logic $\LTL(\prop)$ (since there are more types of formulas, recall Definition~\ref{def:fragment}), in the following, we abusively consider $\LTL_P(\prop)$-formulas to be $\LTL(\prop)$-formulas.
	%
	%In the following, we abusively consider $\LTL_P(\prop)$-formulas as $\LTL(\prop)$ 
	%With an abuse 
	%Although this logic is not a fragment of the logic $\LTL(\prop)$ (since there are more types of formulas, recall Definition~\ref{def:fragment}), there is a direct translation of $\LTL_P(\prop)$-formulas into $\LTL(\prop)$. This translation gives the $\LTL_P(\prop)$-semantics over Kripke structures in $\mathcal{T}(\prop)$.
\end{definition}

As we will see below in Corollary~\ref{coro:ltl_P_logic}, we have obtained a logic for which Theorem~\ref{thm:size_separation} can be applied% by considering a semantic function mapping each formula to the set of states that it satisfies
. However, this came at the cost of expressivity, since there are many $\LTL$-formulas that have no equivalent $\LTL_P$-formulas. Nonetheless, there are many $\LTL$-formulas that do, and it is in particular the case of $\LTL$-formulas whose only temporal operator used is $\lX$. %This is formally stated below.
\begin{proposition}[Proof~\ref{proof:prop_LTL_X_into_LTL_P}]
	\label{prop:LTL_X_into_LTL_P}
	Consider a non-empty set $\prop$, and the $\LTL(\prop)$-fragment $\msf{L}_{\lX}$ that forbids the use of the operators $\lG,\lF,\lU$. Then, for all $\varphi \in \msf{L}_{\lX}$, there is some $\varphi_P \in \LTL_P(\prop)$ such that, for all $K \in \mathcal{T}(\prop)$, we have: $K \models \varphi$ if and only if $K \models \varphi_P$. 
\end{proposition}

Let us now apply Theorem~\ref{thm:size_separation} to this new logic $\LTL_P(\prop)$. %state the corollary akin to those of the previous cases, that we establish with a proof similar to the $\ML$- and $\CTL$- cases. 
\begin{corollary}[Proof~\ref{proof:coro_ltl_P_logic}]
	\label{coro:ltl_P_logic}
	Consider a non-empty set of propositions $\prop$, and any $\msf{LTL}_P(\prop)$-fragment $\msf{L}$. For all $\mathcal{T}(\prop)$-samples $\mathcal{S}$, if there is a $\mathcal{S}$-separating $\msf{L}$-formula, there is one of size at most $2^{n+1}$, with $n := \sum_{K \in \mathcal{S}} |Q_K|$.
\end{corollary}
As a side remark, %we would like to mention that 
since the logic $\LTL_P$ has an inductive behavior, the model checking problem (i.e. the problem of deciding if an $\LTL_P$-formulas satisfies a Kripke structure $K$) can be decided in polynomial time, whereas the complexity for the full logic $\LTL$ is $\msf{PSAPCE}$-complete \cite{sistla1985complexity}, and the complexity for the fragment $\msf{L}_{\lX}$ is $\msf{NP}$-hard \cite{DBLP:journals/tocl/BaulandM0SSV11}% that the model checking problem of $\LTL$-formulas on Kripke structures where the only temporal operator is $\lX$ is $\msf{NP}$-hard
. %This is discussed more thoroughly in Appendix~\ref{proof:prop_model_checking}.

\begin{proposition}[Proof~\ref{proof:prop_model_checking}]
	\label{prop:model_checking}
	Consider a non-empty set $\prop$. Deciding if an $\LTL_P$-formula $\varphi$ accepts a Kripke structure $K$ can be done in time polynomial in $\msf{sz}(\varphi)$ and $|Q_K|$. 
\end{proposition}

\subsubsection{Computation Tree Logic ($\CTL$)}
The logic $\LTL$ is intrinsically linear as it only expresses properties in a single possible future, without branching operator. On the other hand, the Computation Tree Logic ($\CTL$) uses all the temporal operators of the logic $\LTL$ and extends it with path quantifiers (existential or universal). Such $\CTL$-formulas are usually evaluated on Kripke structures with a single action, which we will refer to as actionless Kripke structures. This is defined formally below.
\begin{definition}[$\CTL$ syntax and semantics]
	For a non-empty finite set of propositions $\prop$, the $\CTL(\prop)$-syntax is as follows, with $p\in \prop$, $\mathcal{Q} \in \{\exists, \forall\}$:
	\begin{equation*}
		\varphi ::= p \mid \neg \varphi \mid \varphi \vee \varphi \mid \varphi \wedge \varphi \mid \mathcal{Q} \lX \varphi \mid \mathcal{Q} \lF \varphi \mid \mathcal{Q} \lG \varphi \mid \mathcal{Q} (\varphi \lU \varphi)
	\end{equation*}
	The models on which $\CTL(\prop)$-formulas are evaluated are actionless Kripke structures:
	\begin{equation*}
		\mathcal{K}(\prop) := \{ K = (Q,I,\delta,P,\pi) \mid P \subseteq \prop \}
	\end{equation*}
	%The size $|K|$ of an actionless Kripke structure $K = (Q,I,\delta P,\pi) \in \mathcal{K}(\prop)$ is equal to $|K| := |Q| + |P|$. 
	
	Let us now define the $\CTL$-semantics. Given an actionless Kripke structure $K = (Q,I,\delta,P,\pi) \in \mathcal{K}(\prop)$ and a $\CTL(\prop)$-formula $\varphi$, we define when $\varphi$ is satisfied by a state $q$ inductively as follows, for all $\mathcal{Q} \in \{\exists, \forall\}$:
	
	\hspace*{-0.5cm}
	\begin{minipage}[t]{0.2\textwidth}
		\begin{align*}
			& q \models p & \text{ iff } & p \in \pi(q) \\
			& q \models \neg \varphi & \text{ iff } & q \not\models \varphi \\
			& q \models \varphi_1 \vee \varphi_2 & \text{ iff } & q \models \varphi_1 \text{ or } \\
			&&& q \models \varphi_2\\
			& q \models \varphi_1 \wedge \varphi_2 & \text{ iff } & q \models \varphi_1 \text{ and } \\
			&&& q \models \varphi_2\\
		\end{align*}
	\end{minipage}
	\begin{minipage}[t]{0.05\textwidth}
		\begin{align*}
			\mid \\
			\mid \\
			\mid \\
			\mid \\
			\mid \\
			\mid \\
		\end{align*}
	\end{minipage}
	\begin{minipage}[t]{0.5\textwidth}
		\begin{align*}
			& q \models \mathcal{Q} \lX \varphi & \text{ iff } & \mathcal{Q} \; \rho \in \msf{Path}(q),\; \rho[1:] \models \varphi \\
			& q \models \mathcal{Q} \lF \varphi & \text{ iff } & \mathcal{Q} \; \rho \in \msf{Path}(q),\; \exists i \in \N,\; \rho[i:] \models \varphi \\
			& q \models \mathcal{Q} \lG \varphi & \text{ iff } & \mathcal{Q} \; \rho \in \msf{Path}(q),\; \forall i \in \N,\; \rho[i:] \models \varphi \\
			& w \models \mathcal{Q} (\varphi_1 \lU \varphi_2) & \text{ iff } & \mathcal{Q} \; \rho \in \msf{Path}(q),\; \\
			&&& \; \; \exists i \in \N,\; \rho[i:] \models \varphi_2,\; \\
			&&& \; \; \forall j \leq i-1,\; \rho[j:] \models \varphi_1\\
		\end{align*}
	\end{minipage}
	
	A $\CTL$-formula $\varphi$ accepts a Kripke structure $K = (Q,I,\delta,P,\pi)$, if for all $q \in I$, we have $q \models \varphi$.
\end{definition}

As for $\ML$- and $\LTL$-formulas, Theorem~\ref{thm:size_separation} can be applied to $\CTL$-formulas to obtain an exponential bound. This is identical to the case of $\ML$-formulas. 
\begin{corollary}[Proof~\ref{proof:corollary_ctl}]
	\label{coro:ctl_logic}
	Consider a non-empty set of propositions $\prop$, and any $\msf{CTL}(\prop)$-fragment $\msf{L}$. For all $\mathcal{K}(\prop)$-samples $\mathcal{S}$, if there is a $\mathcal{S}$-separating $\msf{L}$-formula, there is one of size at most $2^{n}$, with $n := \sum_{K \in \mathcal{S}} |Q_K|$.
\end{corollary}
In Appendix~\ref{proof:corollary_ctl}, we actually generalize Corollary~\ref{coro:ctl_logic} to $\ATL$-formulas evaluated on concurrent (multi-player) game structures. The logic $\ATL$ extends the logic $\CTL$ by replacing the path quantifiers $\exists$ and $\forall$ with strategic operators $\fanBrOp{ \cdot }$; while actionless Kripke structures can be seen as one-player concurrent game structures. Proving this generalization of Corollary~\ref{coro:ctl_logic} is actually not more involved than proving Corollary~\ref{coro:ctl_logic} itself (this is only presented in the appendix due to space constraints). Furthermore, with this generalization, we recover the result proved in \cite{DBLP:conf/fm/BordaisNR24}, and thus show that the abstract formalism that we have introduced in this paper does capture and generalize the initial idea developed in \cite{DBLP:conf/fm/BordaisNR24}.

%\subsubsection{LTL on Kripke structures}

\subsubsection{Probabilistic computation tree logic ($\PCTL$)}
\label{subsubsec:pctl}
The logic $\CTL$ uses existential and universal quantifiers overs paths. A probabilistic extension of this logic, called Probabilistic computation tree logic ($\PCTL$) \cite{DBLP:journals/fac/HanssonJ94}, instead uses probabilistic quantification expressing properties about the likelihood that a temporal property is satisfied. Contrary to $\CTL$-formulas, $\PCTL$-formulas are evaluated on Markov chains. This is defined formally below.
\begin{definition}[$\PCTL$ syntax and semantics]
	For a non-empty finite set of propositions $\prop$, the $\PCTL(\prop)$-syntax is as follows, with $p\in \prop$, $\mathcal{\bowtie} \in \{\geq,>,\leq,<,=,\neq\}$ and $r \in \mathbb{Q}$:
	\begin{equation*}
		\varphi ::= p \mid \neg \varphi \mid \varphi \vee \varphi \mid \varphi \wedge \varphi \mid \mathbb{P}_{\bowtie r}(\lX \varphi) \mid \mathbb{P}_{\bowtie r}(\lF \varphi) \mid \mathbb{P}_{\bowtie r}(\lG \varphi) \mid \mathbb{P}_{\bowtie r}(\varphi \lU \varphi)
	\end{equation*}
	The models on which $\PCTL(\prop)$-formulas are evaluated are 
	Markov chains, i.e. tuples $N = (Q,I,\mathbb{P}_N,P,\pi)$ where $Q$ the non-empty set of states, $I \subseteq Q$ is the set of initial states, $\mathbb{P}_N: Q \to \mathcal{D}(Q)$, maps every state to a probability distribution over $Q$, $P$ is a set of propositions, and $\pi: Q \rightarrow 2^P$ maps every state to the set of propositions satisfied at that state. The set of $\PCTL(\prop)$-models $\mathcal{N}(\prop)$ is equal to:
	\begin{equation*}
		\mathcal{N}(\prop) := \{ N = (Q,I,\mathbb{P}_N,P,\pi) \mid P \subseteq \prop \}
	\end{equation*}
	%The size $|K|$ of an actionless Kripke structure $K = (Q,I,\delta P,\pi) \in \mathcal{K}(\prop)$ is equal to $|K| := |Q| + |P|$. 
	
	Consider a Markov chain $N = (Q,I,\mathbb{P}_N,P,\pi) \in \mathcal{N}(\prop)$. We let $\mathsf{Borel}(Q) \subseteq 2^{Q^\omega}$ denote the set of Borel sets on $Q$. In the following, we use (omega-)regular notations on subset of states in $Q$ to describe Borel sets in $\mathsf{Borel}(Q)$. Then, we let $\overline{\mathbb{P}_{N}}: Q \times \mathsf{Borel}(Q) \to [0,1]$ denote the unique probability function, mapping each pair of a state $q \in Q$ and a Borel subset $S \in \mathsf{Borel}(Q) \subseteq 2^{Q^\omega}$ to its probability $\overline{\mathbb{P}_{N}}(q,S)$ to occur from $q$, such that, for all $q,q' \in Q$, we have: $\overline{\mathbb{P}_{N}}(q,\{q\} \cdot \{q'\}) = \mathbb{P}_{N}(q,q') \in [0,1]$. Given a Markov chain $N = (Q,I,\mathbb{P}_N,P,\pi) \in \mathcal{N}(\prop)$ and a $\PCTL(\prop)$-formula $\varphi$, we define when $\varphi$ is satisfied by a state $q$ inductively as follows:
	
	\hspace*{-1cm}
	\begin{minipage}[t]{0.2\textwidth}
		\begin{align*}
			& q \models p & \text{ iff } & p \in \pi(q) \\
			& q \models \neg \varphi & \text{ iff } & q \not\models \varphi \\
			& q \models \varphi_1 \vee \varphi_2 & \text{ iff } & q \models \varphi_1 \text{ or }q \models \varphi_2\\
			& q \models \varphi_1 \wedge \varphi_2 & \text{ iff } & q \models \varphi_1 \text{ and } q \models \varphi_2\\
		\end{align*}
	\end{minipage}
	\begin{minipage}[t]{0.05\textwidth}
		\begin{align*}
			\mid \\
			\mid \\
			\mid \\
			\mid \\
		\end{align*}
	\end{minipage}
	\begin{minipage}[t]{0.5\textwidth}
		\begin{align*}
			& q \models \mathbb{P}_{\bowtie r}(\lX \varphi) & \text{ iff } & \overline{\mathbb{P}_{N}}(q,Q \cdot \{q' \in Q \mid q' \models \varphi\}) \bowtie r \\
			& q \models \mathbb{P}_{\bowtie r}(\lF \varphi) & \text{ iff } & \overline{\mathbb{P}_{N}}(q,Q^* \cdot \{q' \in Q \mid q' \models \varphi\}) \bowtie r \\
			& q \models \mathbb{P}_{\bowtie r}(\lG \varphi) & \text{ iff } & \overline{\mathbb{P}_{N}}(q,(\{q' \in Q \mid q' \models \varphi\})^\omega) \bowtie r \\
			& w \models \mathbb{P}_{\bowtie r}(\varphi_1 \lU \varphi_2) & \text{ iff } & \overline{\mathbb{P}_{N}}(q,S(\varphi_1,\varphi_2)) \bowtie r\\
		\end{align*}
	\end{minipage}
	where $S(\varphi_1,\varphi_2) := \{q' \in Q \mid q' \models \varphi_1\}^* \cdot \{q' \in Q \mid q' \models \varphi_2\}$. Then, a $\PCTL$-formula $\varphi$ accepts a Markov chain $N = (Q,I,\mathbb{P}_N,P,\pi)$, i.e. $N\models \varphi$, if for all $q \in I$, we have $q \models \varphi$.
\end{definition}

As for $\ML$- and $\CTL$-formulas, Theorem~\ref{thm:size_separation} can be applied to $\PCTL$-formulas to obtain an exponential bound. 
\begin{corollary}[Proof~\ref{proof:corollary_pctl}]
	\label{coro:pctl_logic}
	Consider a non-empty set of propositions $\prop$, and any $\msf{PCTL}(\prop)$-fragment $\msf{L}$. For all $\mathcal{N}(\prop)$-samples $\mathcal{S}$, if there is a $\mathcal{S}$-separating $\msf{L}$-formula, there is one of size at most $2^{n}$, with $n := \sum_{N \in \mathcal{S}} |Q_N|$.
\end{corollary}
Note that the proof of this corollary is straightforward, it suffices to consider a semantic function mapping each $\PCTL$-formula to the set of states that it satisfies in a Markov chain, exactly like with $\ML$- and $\CTL$-formulas with Kripke structures. The proof is then almost identical to the $\ML$- and $\CTL$-cases, even though there are infinitely many different $\PCTL$-operators (since probabilistic operators are parameterized by a rational threshold). %On the other hand, the fact that they are infinitely many $\PCTL$-operators 
Note that, on the other hand, instantiating Algorithm~\ref{algo:decide_passive_learning} to this setting is not straightforward, since there are infinitely many different $\PCTL$-operators, i.e. operation c) described in Section~\ref{subsec:algo} is tricky. 

\subsection{Minimal size of words separating automata}
\label{subsec:automata}
Let us now depart from logical formalisms and focus on automaton. %The application of Theorem~\ref{thm:size_separation} in this section is both much more direct than in the previous $\LTL$-formulas evaluated on Kripke structures case and quite different from the other use cases that we have seen in Section~\ref{sec:examples}. 
What we establish here is not novel, but it shows that the abstract formalism that we have introduced, and Theorem~\ref{thm:size_separation} itself, can be applied to various kinds of concrete formalisms%, including automata ones
.

The passive learning of automaton is a well-studied subject, dating back from the 70s \cite{DBLP:journals/iandc/Gold78}. In this setting, the goal is, given an input a set of positive and negative finite words, to learn an automaton accepting the positive words, and rejecting the negative ones. Theorem~\ref{thm:size_separation} is useless in this setting. Indeed, it is folklore that if there exists a separating automaton, there is one whose number of states is at most linear in the size of input, whereas an application of Theorem~\ref{thm:size_separation} would yield an exponential bound on the number of states.

However, Theorem~\ref{thm:size_separation} can be applied in an interesting way in the reversed setting where the goal is to find words separating automata, i.e. we are given sets of positive and negative automata, and we want to exhibit a word accepted by the positive automata and rejected by the negative ones. 

Let us first tackle finite words and (non-deterministic) finite automaton.
\begin{definition}[Finite automaton]
	\label{def:finite_automaton}
	Consider an non-empty alphabet $\Sigma$. A finite automaton $A$ is a tuple $A = (Q,\Sigma,I,\delta,F)$ where $Q$ is a non-empty set (of states), $I \subseteq Q$ is the set of initial states, $F \subseteq Q$ is the set of final states, and $\delta: Q \times \Sigma \rightarrow 2^Q$. An $A$-run on a finite word $u \in \Sigma^*$ is a finite path $\rho \in Q^{|u|+1}$ such that $\rho[0] \in I$, for all $i \in \Interval{0,|\rho|-2}$, we have $\rho[i+1] \in \delta(\rho[i],u[i])$. A word $u \in \Sigma^*$ is accepted by an automaton $A$ if there is an $A$-run $\rho$ on $u$ such that $\head{\rho} \in F$.
\end{definition}
\begin{corollary}[Proof~\ref{proof:corollary_automaton}]
	\label{coro:finite_automaton}
	Consider a non-empty alphabet $\Sigma$. For all pairs $(\mathcal{P},\mathcal{N})$ of finite sets of finite automata, if there exists a finite word $u \in \Sigma^*$ accepted by all automata in $\mathcal{P}$ and rejected by all automata in $\mathcal{N}$, there there is one such word $u$ of size at most $2^{n}-1$, with $n := \sum_{\mathcal{A} \in \mathcal{P} \cup \mathcal{N}} |Q_\mathcal{A}|$.
\end{corollary}
\begin{proof}[Proof sketch]
	We see finite words in $\Sigma^*$ as formulas (inductively defined by the grammar $u ::= \varepsilon \mid u \cdot a$ for $a \in \Sigma$) evaluated on finite automata (the models). We consider the set of semantic values $\msf{SEM}_A := 2^Q$, and the semantic function $\msf{sem}_A$ mapping each finite word $u$ to the set of states where an $A$-run on $u$ can end up. Then, one can realize that, for all finite words $u$, we have that $u$ is accepted by the automaton $A$ if  $\msf{sem}_A(u) \cap F \neq \emptyset$ (thus the pair $(\msf{SEM}_A,\msf{sem}_A)$ captures the semantics), and for all $a \in \Sigma$, we have $\msf{sem}_A(u \cdot a) = \cup_{q \in \msf{sem}_A(u)} \delta(q,a)$ (thus the pair $(\msf{SEM}_A,\msf{sem}_A)$ satisfies the inductive property). We can then apply Theorem~\ref{thm:size_separation}.
\end{proof}
The actual proof of this corollary is a tedious, since we need to introduce a finite-word logic and link it to actual finite words. Alternatively, this corollary can also be proved in a straightforward manner by determinizing all the automata involved, and then considering the product of all of the obtained deterministic automata. In addition, although it may seem surprising, an exponential bound is actually asymptotically optimal, even with a fixed alphabet, as we will discuss in the next section.

Let us now tackle ultimately periodic words and parity automaton.
\begin{definition}[Parity automaton]
	\label{def:parity_automaton}
	For an alphabet $\Sigma$, a parity automaton $A$ is a tuple $A = (Q,\Sigma,I,\delta,\pi)$ where $Q$ is a non-empty set, $I \subseteq Q$, $\delta: Q \times \Sigma \to 2^\Sigma$, and $\pi: Q \to \N$. An $A$-run $\rho$ on a infinite word $w \in \Sigma^\omega$ is a infinite path $\rho \in Q^\omega$ such that $\rho[0] \in I$, and for all $i \in \N$, we have $\rho[i+1] \in \delta(\rho[i],u[i])$. An infinite word $w \in \Sigma^\omega$ is accepted by %a parity automaton 
	$A$ if there is an $A$-run $\rho \in Q^\omega$ on $w$ such that $\max \{ n \in \N \mid \forall i \in \N,\; \exists j \geq i,\; \pi(\rho[j]) = n\}$ is even% (i.e. the least priority seen infinitely often is even)
	. 
\end{definition}

\begin{corollary}[Proof~\ref{proof:corollary_parity_automaton}]
	\label{coro:parity_automaton}
	Consider a non-empty alphabet $\Sigma$. For all pairs $(\mathcal{P},\mathcal{N})$ of finite sets of parity automata, if there is an ultimately periodic word $w = u \cdot v^\omega \in \Sigma^\omega$ accepted by all automata in $\mathcal{P}$ and rejected by all automata in $\mathcal{N}$, there there is one such word $w$ of size at most $2^{n}-1+2^k$, with $n := \sum_{A \in \mathcal{P} \cup \mathcal{N}} |Q_A|$, and $k := \sum_{A \in \mathcal{P} \cup \mathcal{N}} |Q_A|^2 \cdot n_A$, where, for all $A \in \mathcal{P} \cup \mathcal{N}$, $n_A := |\pi(Q)|$, for $\pi(Q) := \{\pi(q) \mid q \in Q \}$.
\end{corollary}
\begin{proof}[Proof sketch]
	Let us only argue that if there is a periodic word $w = v^\omega \in \Sigma^\omega$ accepted by all automata in $\mathcal{P}$ and rejected by all automata in $\mathcal{N}$, then there is some of size at most $2^k$. We see periodic words $v^\omega \in \Sigma^\omega$ as formulas (represented by finite words $v \in \Sigma^+$) evaluated on parity automata (the models). We consider the set of semantic values $\msf{SEM}_A := \{ Q \to 2^{Q \times \pi(Q)}\}$, and the semantic function $\msf{sem}_A$ mapping each non-empty finite word $v$ to the function associating to each state $q$ the set of pairs of a state $q' \in Q$ and an priority $n \in \pi(Q)$ for which there is an $A$-finite run $\rho$ on $v$ from $q$ that ends up in $q'$ and such that the maximum priority visited by $\rho$ is $n$. It is actually straightforward that this pair $(\msf{SEM}_A,\msf{sem}_A)$ satisfies the inductive property, it is a little trickier to show that is captures the semantics. The idea is that from the function $\msf{sem}_A(v): Q \to 2^{Q \times \pi(Q)}$, we can recover all possible infinite sequences of priorities $(n_m)_{m \in \N}$ for which there is an $A$-run $\rho \in Q^\omega$ on $v^\omega$ such that, for all $m \in \N$, $n_m$ corresponds to the maximum priority visited between the $m \cdot |v|$ and $(m+1) \cdot |v|$ steps. We can then apply Theorem~\ref{thm:size_separation}.
\end{proof}
%As above, we prove this corollary in Appendix~\ref{proof:corollary_parity_automaton} by applying Theorem~\ref{thm:size_separation}, which is a little trickier than for finite words on finite automata. However, proving Corollary~\ref{coro:parity_automaton} with only automata theoretic arguments is not straightforward either.

\section{Lower bound}
\label{sec:lower_bound}
We have seen several use-cases where Theorem~\ref{thm:size_separation} can be applied to obtain an exponential upper bound on the minimal size of separating formulas, assuming one exists. In this section, we argue that our upper bounds are not orders of magnitude away from lower bounds, in at least two settings: $\LTL$-formulas evaluated on ultimately periodic words, and $\ML$-formulas evaluated on Kripke structures. %In both cases, the family of samples from which we derive the lower bounds are designed from ideas or results coming from automata theory. 

%In all of this section we consider a set of propositions $\prop := \{p,\bar{p}\}$.

\textbf{$\LTL$-formulas evaluated on infinite words.} %Consider a non-empty set of propositions $\prop$. 
Without restrictions on the operators that can be used, given any sample $\mathcal{S}$ of ultimately periodic words, polynomial size $\LTL$-formulas are able to describe the differences between each pair of positive and negative words in $\mathcal{S}$. Thus, by using nested conjunctions and disjunctions, it is easy to build a polynomial size formula that is $\mathcal{S}$-separating. However, the goal of synthesizing a separating formula in a passive learning setting is to capture, in a concise way, the difference between all the positive models and all the negative models. Therefore, when studying the passive learning problem with $\LTL$-formulas, it is usual to restrict the set of operators allowed so that it is not possible to use both conjunctions and disjunctions (see e.g. \cite{arXivFijalkow}). We focus below on such a kind of fragment, called monotone fragments. This is defined both for $\LTL$- and $\ML$- formulas.
\begin{definition}[Monotone fragments]
	\label{def:monotone_fragments}
	Consider a non-empty set $\prop$ and let $\msf{L} \in \{\LTL,\ML\}$. An $\msf{L}(\prop)$-fragment $\msf{L}'$ is \emph{monotone} if $\neg \notin \msf{Op}'$, and $|\{ \vee,\wedge\} \cap \msf{Op}'| \leq 1$.
	%the operator $p$ or the operator $\bar{p}$ is allowed, the operator $\neg$ is disallowed, and so is at least one of the operators $\wedge,\vee$.
\end{definition}
With a monotone $\LTL$-fragment, we asymptotically obtain a sub-exponential lower bound.
%can exhibit a family of samples where the minimal size of separating formulas is asymptotically at least $2^{\sqrt{k}}$, where $k := \sum_{w \in \mathcal{S}} |w|$.
\begin{proposition}[Proof~\ref{proof:lower_bound_ltl}]
	\label{prop:lower_bound_ltl}
	Let $\prop := \{p,\bar{p}\}$. Consider a monotone $\LTL(\prop)$-fragment $\msf{L}'$ with $\lX \in \msf{Op}'$ and $\lU \notin \msf{Op}'$. % that allows the next operator $\lX$ but disallows the until operator $\lU$. 
	For all $n \in \N$, there is a $\mathcal{W}(\prop)$-sample $\mathcal{S}$ such that $k := \sum_{w \in \mathcal{S}} |w| \geq n$, and the minimal size of an  $\mathcal{S}$-separating $\msf{L}$-formula is larger than $2^{\sqrt{k}}$.
\end{proposition}
The family of samples that we exhibit to establish this proposition is constituted of periodic words $(w_i)^\omega$ where $|w_i| = i\in\N$ is a prime number. The size of the $n$-th sample if then roughly equal to the sum of the $n$ first prime numbers, while the minimal size of a formula separating this $n$-th sample is roughly equal to the product of the $n$ first prime numbers. The bound of Proposition~\ref{prop:lower_bound_ltl} then follows from results on prime numbers. Note that similar ideas are used in \cite{DBLP:journals/tcs/Chrobak86}, which deals with minimal size automata.

\textbf{$\ML$-formulas evaluated on Kripke structures.}
We have a higher lower bound with modal logic monotone fragments. To derive this lower bound, we use the theorem below.
\begin{theorem}[Theorems 32 and 10 in  \cite{DBLP:journals/jalc/EllulKSW05}]
	\label{thm:automtaton_least_word_exponential}
	Let $\Sigma := \{a,b,c,d,e\}$. For all $n \geq 3$, there is an automaton $A_n$ with $25n+111$ states, a single initial state and such that a smallest word in $\Sigma^*$ that is not accepted by $A_n$ is of size $(2^n-1) \cdot (n+1) +1$.
\end{theorem}

By turning finite automata into Kripke structures, we establish the proposition below.
\begin{proposition}[Proof~\ref{proof:lower_bound_ml}]
	\label{prop:lower_bound_ml}
	Let $\prop := \{p\}$, $\act := \{a,b,c,d,e\}$, and $\msf{Op}_{[\cdot]} := \{[\alpha] \mid \alpha \in \act\}$ and $\msf{Op}_{\langle\cdot\rangle} := \{\langle\alpha\rangle^{\geq 1} \mid \alpha \in \act\}$. Consider an $\ML(\prop,\act)$-monotone fragment $\msf{L}'$ such that, for all $k \geq 2$ and $\alpha \in \act$, $\langle \alpha\rangle^{\geq k} \notin \msf{Op}'$, $\msf{Op}_{[\cdot]} \subseteq \msf{Op}'$ or $\msf{Op}_{\langle\cdot\rangle} \subseteq \msf{Op}'$, and for all $(*,\theta) \in \{(\wedge,[\cdot]),(\vee,\langle\cdot\rangle)\}$, we have $* \in \msf{Op}'$ implies $\msf{Op}_{\theta} \subseteq \msf{Op}'$. 
	%if $\wedge \in \msf{Op}$, then $\msf{O}_{[]} \subseteq \msf{Op}$ and if $\vee \in \msf{Op}$, then $\msf{O}_{\langle\rangle} \subseteq \msf{Op}$. 
	Then, for all $n \in \N$, there is a $\mathcal{K}(\prop,\act)$-sample $\mathcal{S}$ %that is $\msf{L}$-separable 
	such that $k := \sum_{K \in \mathcal{S}} |Q_K| \geq n$, and the minimal size of an $\mathcal{S}$-separating $\msf{L}$-formula is at least $2^{\frac{k}{25}}$.
\end{proposition}
%The proof requires some technical details, but is relatively straightforward.

\section{Conclusion and future work}
The main contribution of this paper is Theorem~\ref{thm:size_separation}. It provides theoretical groundings to enumeration algorithms as it exhibits a termination criterion. We have easily applied this theorem to various concrete formalisms, and we have also considered the trickier case of $\LTL$-formulas evaluated on Kripke structures. This latter setting showcases one of the strength of Theorem~\ref{thm:size_separation}: even when it cannot be applied to some logic, it can guides us towards finding a related logic to which it can. Following, there may be more efficient enumeration algorithms for that related logic than for the original logic.
%For such a related logic, there may be a more efficient enumeration algorithm%, and possibly of more tractable model checking problem

We also believe that one of the most promising future work is the study, both on the theoretical side and on the experimental side (as mentioned before, it has already been successfully applied for learning regular expressions \cite{DBLP:journals/pacmpl/ValizadehB23} and $\LTL$ formulas \cite{DBLP:conf/cav/ValizadehFB24}), of the meta algorithm discussed in Subsection~\ref{subsec:algo}. The main asset of this algorithm is that its nature prevents it from enumerating semantically-equivalent but syntactically-different formulas, which can be very significant in situations where the number of semantic values is (much) smaller than the number of candidate formulas.

\bibliographystyle{alpha}
\bibliography{main}

\appendix
\section{Proofs on modal logic}
\subsection{Proof of Lemma~\ref{lem:modal_logic_pair_captures_semantics}}
%\texorpdfstring{Proof of Lemma~\ref{lem:modal_logic_pair_captures_semantics}}{Proof of Lemma~\ref{lem:modal_logic_pair_captures_semantics}}
\label{proof:lem_modal_logic_pair_captures_semantics}
\begin{proof}
	Consider a Kripke structure $K \in \mathcal{K}(\prop,\act)$. For all $\ML(\prop,\act)$-formulas $\varphi,\varphi'$, if $\msf{sem}_K(\varphi) = \msf{sem}_K(\varphi')$, then:
	\begin{align*}
		K \models \varphi & \Longleftrightarrow \forall q \in I,\; q \models \varphi \Longleftrightarrow \forall q \in I,\; q \in \msf{sem}_K(\varphi) = \msf{sem}_K(\varphi') \\
		& \Longleftrightarrow \forall q \in I,\; q \models \varphi' \Longleftrightarrow K \models \varphi'
	\end{align*}
	In this case, we have $\msf{SAT}_{K} := \{ S \in 2^Q \mid I \subseteq S \}$.
\end{proof}

\subsection{Proof of Lemma~\ref{lem:modal_logic_pair_inductive_property}}
\label{proof:lem_modal_logic_pair_inductive_property}
\begin{proof}
	Consider a Kripke structure $K \in \mathcal{K}(\prop,\act)$. Let us define $\Theta_K$-compatible functions for all operators $\msf{o} \in \msf{Op}_1 \cup \msf{Op}_2$. For all $S,S_1,S_2 \in \msf{SEM}_K = 2^Q$, we let:
	\begin{itemize}
		\item $\msf{sem}_K^{\neg}(S) := Q \setminus S \in \msf{SEM}_K$;
		\item for all $a\in \act$ and $k \in \N$, $\msf{sem}_K^{\langle a \rangle^{\geq k}}(S) := \{q \in Q \mid \delta(q,a) \cap S \geq k \}$
		\item for all $a\in \act$, $\msf{sem}_K^{[a]}(S) := \{q \in Q \mid \delta(q,a) \subseteq S \}$
		\item $\msf{sem}_K^{\wedge}(S_1,S_2) := S_1 \cap S_2 \in \msf{SEM}_K$;
		\item $\msf{sem}_K^{\lor}(S_1,S_2) := S_1 \cup S_2 \in \msf{SEM}_K$.
	\end{itemize}
	It is straightforward to check that, with these definitions, all of the above functions do correspond to $\Theta_K$-compatible functions. Let us for instance consider the case of the operator $\wedge \in \msf{Op}_2$. For all $\varphi_1,\varphi_2 \in \msf{Fm}_{\ML(\prop,\act)} \times \msf{Fm}_{\ML(\prop,\act)}$, we have:
	\begin{align*}
		\msf{sem}_K^{\wedge}(\msf{sem}_K(\varphi_1),\msf{sem}_K(\varphi_2)) & = \msf{sem}_K(\varphi_1) \cap \msf{sem}_K(\varphi_2) \\
		& = \{q \in Q \mid q \models \varphi_1 \} \cap \{q \in Q \mid q \models \varphi_2 \} \\
		& = \{ q \in Q \mid q \models \varphi_1 \text{ and }q \models \varphi_2 \} \\
		& = \{ q \in Q \mid q \models \varphi_1 \wedge \varphi_2 \} \\
		& = \msf{sem}_K(\varphi_1 \wedge \varphi_2)
	\end{align*}
	This is similar with the $[a]$ for some $a \in \act$:
	For all $\varphi' \in \msf{Fm}_{\ML(\prop,\act)}$, we have:
	\begin{align*}
		\msf{sem}_K^{[a]}(\msf{sem}_K(\varphi')) & = \{q \in Q \mid \delta(q,a) \subseteq \msf{sem}_K(\varphi') \} \\
		& = \{ q \in Q \mid \delta(q,a) \subseteq \{ q' \in Q \mid q' \models \varphi' \} \} \\
		& = \msf{sem}_K([a] \; \varphi')
	\end{align*}
	
	Since there are $\Theta_K$-compatible functions for all operators $\msf{o} \in \msf{Op}_1 \cup \msf{Op}_2$, it follows that the $(\ML(\prop,\act),K)$-pair $\Theta_K$ inductively captures the $\ML(\prop,\act)$-semantics. 
	
	%Consider now Assumptions~\ref{assump:for_complexity}. Clearly, all of the sets $\msf{SEM}_K(\top),\msf{SEM}_K(\bot),\msf{SEM}_K(p)$ (for all propositions $p \in \prop$) and all of the above functions can be computed in time polynomial in $|M|$. Furthermore, checking that some $S \in \msf{SEM}_K$ is in the set $\msf{SAT}_K$ can also be done in polynomial time. In addition, we have $|\msf{SEM}_K| = 2^{|Q|} \leq 2^{|M|}$. Thus, the $(\ML(\prop,\act),K)$-pair $\Theta_K$ satisfies Assumptions~\ref{assump:for_complexity}.
\end{proof}

\iffalse
\subsection{Proof of Corollary~\ref{coro:modal_logic_complexity}}
\label{proof:coro_modal_logic_complexity}
\begin{proof}
	TODO
\end{proof}
\fi

\section{Proof of Theorems~\ref{thm:size_separation} and~\ref{thm:enumeration_algo}}
\label{proof:thm_size_separation}

We start with the proof of Theorem~\ref{thm:size_separation}, the proof of Theorem~\ref{thm:enumeration_algo} that will follow uses some of statements used to prove Theorem~\ref{thm:size_separation}. 

\subsection{Proof of Theorem~\ref{thm:size_separation}}
We are actually going to prove a (slightly) stronger statement than Theorem~\ref{thm:size_separation} as we consider a (slightly) more general setting than the passive learning problem. This is captured by the notion of learning property, defined below.
\begin{definition}
	\label{def:learning_prop}
	Consider a logic $\msf{L}$ and a finite set of $\msf{L}$-models $\mathcal{S}$. An $(\msf{L},\mathcal{S})$-learning property $\mathcal{R} \subseteq \msf{Fm}_{\msf{L}}^{\msf{f}}$ is a subset of final $\msf{L}$-formulas such that, letting: 
	\begin{equation*}
		\mathcal{R}_{\mathcal{S}} := \{ \{ M \in \mathcal{S} \mid M \models \varphi \} \mid \varphi \in \mathcal{R} \}
	\end{equation*}
	we have:
	\begin{equation*}
		\forall \varphi \in \msf{Fm}_{\msf{L}}^{\msf{f}}:\; \{ M \in \mathcal{S} \mid M \models \varphi \} \in 	\mathcal{R}_{\mathcal{S}} \Longleftrightarrow \varphi \in \mathcal{R}
	\end{equation*}
	In other words, whether or not a final $\msf{L}$-formula $\varphi$ belongs to $\mathcal{R}$ entirely depends on the set of models in $\mathcal{S}$ satisfying $\varphi$.
\end{definition}

Let us now state the proposition that generalizes Theorem~\ref{thm:size_separation} to the more general case of learning properties.
\begin{proposition}
	\label{prop:size_separation}
	Consider a logic $\msf{L}$, a fragment $\msf{L}'$ of the logic $\msf{L}$, a class of $\msf{L}$-models $\mathcal{C}$ and an $(\msf{L},\mathcal{C})$-pair $\Theta$. Let $\mathcal{S} \subseteq \mathcal{C}$ be a finite set of $\msf{L}$-models, and $\mathcal{R} \subseteq \msf{Fm}_{\msf{L}}^{\msf{f}}$ be an $(\msf{L},\mathcal{S})$-learning property. If there is an $\msf{L}'$-formula in $\mathcal{R}$, then there is one of size at most $n^{\mathcal{S}}_{\Theta} = \sum_{\tau \in \mathcal{T}} \prod_{M \in \mathcal{S}} |\msf{SEM}_M(\tau)|$.
\end{proposition}

The remainder of this subsection is devoted to the of Proposition~\ref{prop:size_separation}. Therefore, we fix a logic $\msf{L}$, a fragments $\msf{L}'$ of the logic $\msf{L}$, a class of $\msf{L}$-models $\mathcal{C}$, an $(\msf{L},\mathcal{C})$-pair $\Theta$, a finite set of $\msf{L}$-models $\mathcal{S} \subseteq \mathcal{C}$, and an $(\msf{L},\mathcal{S})$-learning property $\mathcal{R} \subseteq \msf{Fm}_{\msf{L}}^{\msf{f}}$ for all of this subsection. %In the following, for all $\msf{L}$-formulas $\varphi$, we let $\tau_\varphi \in \mathcal{T}$ denote the type of the formula $\varphi$.

Let us first define iteratively sets of tuples of semantic sets. This is done formally below.
\begin{definition}
	\label{def:inductive_def_semantics_sets}
	We define by induction the sets $(\msf{SEM}_{\mathcal{S},\Theta,i}^{\msf{L}'}(\tau))_{i \in \N, \tau \in \mathcal{T}}$, as follows:
	\begin{itemize}
		\item For all $\tau \in \mathcal{T}$, we define:
		\begin{align*}
			\msf{SEM}_{\mathcal{S},\Theta,0}^{\msf{L}'}(\tau) & := \{ (\msf{sem}_M(\msf{o}))_{M \in \mathcal{S}} \mid \msf{o} \in \msf{Op}_0'(\tau) \} \subseteq \prod_{M \in \mathcal{S}} \msf{SEM}_M(\tau)
		\end{align*}
		%and, for all $\msf{z} := (\msf{cx}_M,\msf{S}_M)_{M \in \mathcal{S}} \in \msf{SEM}_{\mathcal{S},0}^{\msf{L}'}(\tau)$, we let:
		%\begin{equation*}
		%	f_{\mathcal{S}}^{\msf{L}'}(\msf{z}) := \{ \varphi \in \msf{Op}_0 \cap \msf{Fm}_{\msf{L}'}(\tau) \mid \forall M \in \mathcal{S},\; \msf{S}_M = \msf{sem}^M_\tau(\msf{cx}_M,\varphi) \}
		%\end{equation*}
		\item For all $i \geq 0$, for all $\tau \in \mathcal{T}$, we let:
		\begin{align*}
			\msf{Seq}_{\mathcal{S},\Theta,i+1}^{\msf{L}'}(\tau) := \msf{Seq}_{\mathcal{S},\Theta,i}^{\msf{L}'}(\tau) \cup  \msf{Seq}_{\mathcal{S},\Theta,i+1}^{\msf{L}',1}(\tau) \cup \msf{Seq}_{\mathcal{S},\Theta,i+1}^{\msf{L}',2}(\tau)
		\end{align*}
		with 
		\begin{align*}
			\msf{Seq}_{\mathcal{S},\Theta,i+1}^{\msf{L}',1}(\tau) & := \{ (\msf{sem}_{M}^{\msf{o}}(\msf{X}[M]))_{M \in \mathcal{S}} \mid \msf{o} \in \msf{Op}_1'(\tau),\; \msf{X} \in \msf{SEM}_{\mathcal{S},i}^{\msf{L}'}(T(\msf{o},1))  \} \subseteq \prod_{M \in \mathcal{S}} \msf{SEM}_M(\tau)
		\end{align*}
		and
		\begin{align*}
			\msf{Seq}_{\mathcal{S},\Theta,i+1}^{\msf{L}',2}(\tau) & := \{ (\msf{sem}_{M}^{\msf{o}}(\msf{X}^1[M],\msf{X}^2[M]))_{M \in \mathcal{S}} \mid \msf{o} \in \msf{Op}_2'(\tau),\; \\ 	& \phantom{aaaa} (\msf{X}^1,\msf{X}^2) \in \msf{SEM}_{\mathcal{S},i}^{\msf{L}'}(T(\msf{o},1)) \times \msf{SEM}_{\mathcal{S},i}^{\msf{L}'}(T(\msf{o},2)) \} \subseteq \prod_{M \in \mathcal{S}} \msf{SEM}_M(\tau)
		\end{align*}
	\end{itemize}
	where, for all $i \in \N$ and $T \subseteq \mathcal{T}$, $\msf{SEM}_{\mathcal{S},\Theta,i}^{\msf{L}'}(T)$ denotes the set $\bigcup_{\tau \in T} \msf{SEM}_{\mathcal{S},\Theta,i}^{\msf{L}'}(\tau)$. Then:
	\begin{itemize}
		\item for all $\tau \in \mathcal{T}$, we let  $\msf{SEM}_{\mathcal{S},\Theta}^{\msf{L}'}(\tau) := \bigcup_{n \in \N} \msf{SEM}_{\mathcal{S},\Theta,n}^{\msf{L}'}(\tau) \subseteq \prod_{M \in \mathcal{S}} \msf{SEM}_M(\tau)$;
		\item for all $n \in \mathbb{N}$, we let $\msf{SEM}_{\mathcal{S},\Theta,n}^{\msf{L}'} := \bigcup_{\tau \in \mathcal{T}} \msf{SEM}_{\mathcal{S},\Theta,n}^{\msf{L}'}(\tau) \subseteq \prod_{M \in \mathcal{S}} \msf{SEM}_M$;
		\item $\msf{SEM}_{\mathcal{S},\Theta}^{\msf{L}'} := \bigcup_{\tau \in \mathcal{T}} \msf{SEM}_{\mathcal{S},\Theta}^{\msf{L}'}(\tau) = \bigcup_{n \in \N} \msf{SEM}_{\mathcal{S},\Theta,n}^{\msf{L}'} \subseteq \bigcup_{\tau \in \mathcal{T}} \prod_{M \in \mathcal{S}} \msf{SEM}_M(\tau)$. 
	\end{itemize}
\end{definition}

On the other hand, let us define the set of tuples of semantic sets that any $\msf{L}'$-formula can be mapped to by semantic functions.
\begin{definition}
	We let $\msf{sem}_\Theta: \msf{Fm}_{\msf{L}'} \to \prod_{M \in \mathcal{S}} \msf{SEM}_M$ be such that, for all $\varphi \in \msf{Fm}_{\msf{L}'}$, $\msf{sem}_\Theta(\varphi) := (\msf{sem}_M(\varphi))_{M \in \mathcal{S}} \in \prod_{M \in \mathcal{S}} \msf{SEM}_M$. Then, we let:	
	\begin{equation*}
		\msf{AllSEM}_{\mathcal{S},\Theta}^{\msf{L}'} := \msf{sem}_\Theta[\msf{Fm}_{\msf{L}'}] \subseteq \prod_{M \in \mathcal{S}} \msf{SEM}_M
	\end{equation*}
\end{definition}

We claim that these two sets $\msf{AllSEM}_{\mathcal{S},\Theta}^{\msf{L}'}$ and $\msf{SEM}_{\mathcal{S},\Theta}^{\msf{L}'}$ are actually equal, as stated in the two lemmas below.
\begin{lemma}
	\label{lem:allsat_achieved_by_computation}
	We have: $\msf{AllSEM}_{\mathcal{S},\Theta}^{\msf{L}'} \subseteq \msf{SEM}_{\mathcal{S},\Theta}^{\msf{L}'}$.
\end{lemma}
\begin{proof}
	Let us show by induction on $\msf{L}'$-formulas $\varphi \in \msf{Fm}_{\msf{L}'}$ the property $\mathcal{P}(\varphi)$: $\msf{sem}_\Theta(\varphi) \in \msf{SEM}_{\mathcal{S},\Theta}^{\msf{L}'}$. Consider first any $\msf{L}'$-formula $\varphi \in \msf{Op}_0 \cap \msf{Fm}_{\msf{L}'}$. By definition, we have $\msf{sem}_\Theta(\varphi) = (\msf{SEM}_M(\varphi))_{M \in \mathcal{S}} \in \msf{SEM}_{\mathcal{S},\Theta,0}^{\msf{L}'} \subseteq \msf{SEM}_{\mathcal{S},\Theta}^{\msf{L}'}$. Hence, $\mathcal{P}(\varphi)$ holds. 
	
	Now, let $\tau \in \mathcal{T}$ be a type and consider an operator $\msf{o} \in \msf{Op}_2$ such that $\tau_\msf{o} = \tau$. Assume that the properties $\mathcal{P}(\varphi_1),\mathcal{P}(\varphi_2)$ hold for two $\msf{L}'$-formulas $(\varphi_1,\varphi_2) \in \msf{Fm}_{\msf{L}'}(T(\msf{o},1)) \times \msf{Fm}_{\msf{L}'}(T(\msf{o},2))$ such that $\varphi := \msf{o}(\varphi_1,\varphi_2) \in \msf{Fm}_{\msf{L}'}(\tau)$. By $\mathcal{P}(\varphi_1)$ and $\mathcal{P}(\varphi_2)$, we have $\msf{sem}_\Theta(\varphi_1) \in \msf{SEM}_{\mathcal{S},\Theta}^{\msf{L}'}$ and $\msf{sem}_\Theta(\varphi_2) \in \msf{SEM}_{\mathcal{S},\Theta}^{\msf{L}'}$. Thus, there is $n \in \mathbb{N}$ such that $(\msf{sem}_\Theta(\varphi_1),\msf{sem}_\Theta(\varphi_2)) \in \msf{SEM}_{\mathcal{S},\Theta,n}^{\msf{L}'}(T(\msf{o},1)) \times \msf{SEM}_{\mathcal{S},\Theta,n}^{\msf{L}'}(T(\msf{o},2))$, with $\msf{o}(\varphi_1,\varphi_2) \in \msf{Fm}_{\msf{L}'}(\tau)$. Thus, we have:
	\begin{align*}
		\msf{sem}_{\Theta}(\varphi) & = (\msf{sem}_M(\varphi))_{M \in \mathcal{S}} \\
		& = (\msf{sem}^{\msf{o}}_{M}(\msf{sem}_M(\varphi_1),\msf{sem}_M(\varphi_2)))_{M \in \mathcal{S}} \in \msf{SEM}_{\mathcal{S},\Theta,n+1}^{\msf{L}'}(\tau) \subseteq \msf{SEM}_{\mathcal{S},\Theta}^{\msf{L}'}
	\end{align*}
	
	Thus, $\mathcal{P}(\varphi)$ holds. Similar arguments also apply to arity-1 operators $\msf{o} \in \msf{Op}_1$. In fact, $\mathcal{P}(\varphi)$ holds for all $\msf{L}'$-formulas $\varphi \in \msf{Fm}_{\msf{L}'}$. Hence, $\msf{AllSEM}_{\mathcal{S},\Theta}^{\msf{L}'} \subseteq \msf{SEM}_{\mathcal{S},\Theta}^{\msf{L}'}$.
\end{proof}

\begin{lemma}
	\label{lem:achieved_sat_injective_function}
	There is a function $\varphi_\Theta: \msf{SEM}_{\mathcal{S},\Theta}^{\msf{L}'} \rightarrow \msf{Fm}_{\msf{L}'}$ such that, for all $X \in \msf{SEM}_{\mathcal{S},\Theta}^{\msf{L}'}$:
	\begin{enumerate}
		\item $\msf{sem}_\Theta(\varphi_\Theta(X)) = X$; and
		\item for all $\varphi \in \msf{Sub}(\varphi_\Theta(X))$, we have: $\varphi = \varphi_\Theta(\msf{sem}_\Theta(\varphi))$.
	\end{enumerate}
\end{lemma}
\begin{proof}
	We define the function $\varphi_\Theta: \msf{SEM}_{\mathcal{S},\Theta}^{\msf{L}'} \rightarrow \msf{Fm}_{\msf{L}'}$, with $\msf{SEM}_{\mathcal{S},\Theta}^{\msf{L}'} = \bigcup_{i \in \N} \msf{SEM}_{\mathcal{S},\Theta,i}^{\msf{L}'}$, by induction on $i \in \mathbb{N}$. First, consider some $X \in \msf{SEM}_{\mathcal{S},\Theta,0}^{\msf{L}'}$. We let $\varphi_{\Theta}(X) \in \msf{Op}_0'$ be some arity-0 operator such that $\msf{sem}_\Theta(\varphi_{\Theta}(X)) = (\msf{sem}_M(\varphi_{\Theta}(X)))_{M \in \mathcal{S}} = X$. Note that this is indeed possible by definition of $\msf{SEM}_{\mathcal{S},\Theta,0}^{\msf{L}'}$. This satisfies both conditions (1),(2) above by definition.
	
	Now, assume that the function $\varphi_\Theta$ is already defined on $\bigcup_{i \leq n} \msf{SEM}_{\mathcal{S},\Theta,i}^{\msf{L}'}$ for some $n \in \mathbb{N}$, and assume that it satisfies both conditions (1),(2) above on this set. In particular, this implies that, for all $X \in \bigcup_{i \leq n} \msf{SEM}_{\mathcal{S},\Theta,i}^{\msf{L}'}$ and for all $\varphi \in \msf{Sub}(\varphi_\Theta(X))$, the formula $\varphi_\Theta(\msf{sem}_\Theta(\varphi)) \in \msf{Fm}_{\msf{L}'}$ is well-defined.
	
	Consider now any $X \in \msf{SEM}_{\mathcal{S},\Theta,n+1}^{\msf{L}'}$ on which the function $\varphi_\Theta$ is not already defined. Assume that $X \in \msf{SEM}_{\mathcal{S},\Theta,n+1}^{\msf{L}',2}(\tau)$ for some type $\tau \in \mathcal{T}$. (The case $X \in \msf{SEM}_{\mathcal{S},\Theta,n+1}^{\msf{L}',1}(\tau)$ is similar.) By definition, this implies that there is some arity-2 operator $\msf{o} \in \msf{Op}_2'$ and $X^1 = (X^1_M)_{M \in \mathcal{S}} \in \msf{SEM}_{\mathcal{S},\Theta,n}^{\msf{L}'}(T(\msf{o},1))$ and $X^2 = (X^2_M)_{M \in \mathcal{S}} \in \msf{SEM}_{\mathcal{S},\Theta,n}^{\msf{L}'}(T(\msf{o},2))$ such that $X = (\msf{sem}^{\msf{o}}_{M}(X^1_M,X^2_M))_{M \in \mathcal{S}}$. By assumption, the function $\varphi_\Theta$ is defined on $X^1$ and $X^2$. Since we have $\msf{sem}_\Theta(\varphi_\Theta(X^1)) = X^1 \in \msf{SEM}_{\mathcal{S},\Theta,n}^{\msf{L}',n}(T(\msf{o},1))$, it follows that $\varphi_\Theta(X^1) \in \msf{Fm}_{\msf{L}'}(T(\msf{o},1))$. Similarly, we have $\varphi_\Theta(X^2) \in \msf{Fm}_{\msf{L}'}(T(\msf{o},2))$. Hence, $\varphi := \msf{o}(\varphi_\Theta(X^1),\varphi_\Theta(X^2))$ is an $\msf{L}'$-formula. Thus, we let $\varphi_\Theta(X) := \varphi$. Consider any $M \in \mathcal{S}$. Since $\varphi_\Theta$ satisfies condition (1) on $X^1$ and $X^2$, we have $\msf{sem}_M(\varphi_{\Theta}(X^1)) = X^1_M$ and $\msf{sem}_M(\varphi_{\Theta}(X^2)) = X^2_M$. Furthermore, since the function $\msf{sem}^{\msf{o}}_{M}$ is $\Theta_M$-compatible, we have $\msf{sem}_M(\msf{o}(\varphi_\Theta(X^1),\varphi_\Theta(X^2))) = \msf{sem}^{\msf{o}}_{M}(\msf{sem}_M(\varphi_\Theta(X^1)),\msf{sem}_M(\varphi_\Theta(X^2)))$ (recall Definition~\ref{def:inductive_semantics}). Thus, we have:
	\begin{align*}
		\msf{sem}_M(\varphi_\Theta(X)) & = \msf{sem}_M(\msf{o}(\varphi_\Theta(X^1),\varphi_\Theta(X^2))) \\
		& = \msf{sem}^{\msf{o}}_{M}(\msf{sem}_M(\varphi_\Theta(X^1)),\msf{sem}_M(\varphi_\Theta(X^2))) \\
		& = \msf{sem}^{\msf{o}}_{M}(X^1_M,X^2_M) \\
		& = X_M
	\end{align*}
	Therefore, we have $\msf{sem}_\Theta(\varphi_\Theta(X)) = X$. Thus condition (1) holds on $X$. Consider now any $\psi \in \msf{Sub}(\varphi_\Theta(X))$. If $\psi \in \msf{Sub}(\varphi_\Theta(X^1)) \cup  \msf{Sub}(\varphi_\Theta(X^2))$, then by assumption we have $\psi = \varphi_\Theta(\msf{sem}_\Theta(\psi))$. Otherwise, $\psi = \varphi_\Theta(X) = \varphi_\Theta(\msf{sem}_\Theta(\varphi_\Theta(X))) = \varphi_\Theta(\msf{sem}_\Theta(\psi))$. Thus condition (2) holds on $X$ as well. In fact, both of these conditions hold on $X$. We can then define $\varphi_\Theta$ on all of $\msf{SEM}_{\mathcal{S},\Theta,n+1}^{\msf{L}'}$. That way, we can define $\varphi_\Theta: \msf{SEM}_{\mathcal{S},\Theta}^{\msf{L}'} \rightarrow \msf{Fm}_{\msf{L}'}$ while satisfying both conditions (1) and (2).
\end{proof}

The proof of Proposition~\ref{prop:size_separation} now follows.
\begin{proof}
	Assume that $\mathcal{R} \cap \msf{Fm}_{\msf{L}'}^{\msf{f}} \neq \emptyset$. Consider then any $\varphi \in \mathcal{R} \cap \msf{Fm}_{\msf{L}'}$. Let $X = (X_M)_{M \in \mathcal{S}} := \msf{sem}_\Theta(\varphi) \in \msf{AllSEM}_{\mathcal{S},\Theta}^{\msf{L}'}$. By Lemma~\ref{lem:allsat_achieved_by_computation}, we have $\msf{AllSEM}_{\mathcal{S},\Theta}^{\msf{L}'} \subseteq \msf{SEM}_{\mathcal{S},\Theta}^{\msf{L}'}$. Thus, $X \in \msf{SEM}_{\mathcal{S},\Theta}^{\msf{L}'}$. Consider the function $\varphi_\Theta: \msf{SEM}_{\mathcal{S},\Theta}^{\msf{L}'} \rightarrow \msf{Fm}_{\msf{L}'}$ from Lemma~\ref{lem:achieved_sat_injective_function}. Let $\psi := \varphi_\Theta(X) \in \msf{Fm}_{\msf{L}'}$. We claim that $\psi$ is in $\mathcal{R}$ and of size at most $n_\Theta^{\mathcal{S}}$.
	
	Consider any $M \in \mathcal{S}$. Since $\varphi_\Theta$ satisfies condition (1) of Lemma~\ref{lem:achieved_sat_injective_function}, we have $\msf{sem}_M(\psi) = X_M = \msf{sem}_M(\varphi) \in \msf{SEM}_M$. This implies that $\psi$ and $\varphi$ are of the same type, and thus $\psi \in \msf{Fm}_{\msf{L}'}^{\msf{f}}$. Furthermore, because the $(\msf{L},M)$-pair $\Theta_M$ captures the $\msf{L}$-semantics (recall Definition~\ref{def:capturing_semantics}), it follows that $M \models \psi \Longleftrightarrow M \models \varphi$. Thus, we have $\{ M \in \mathcal{S} \mid M \models \psi \} = \{ M \in \mathcal{S} \mid M \models \varphi \} \in \mathcal{R}_{\mathcal{S}}$. Thus, we have $\psi \in \mathcal{R}$.
	
	Furthermore, the function $\msf{sem}_\Theta: \msf{Sub}(\psi) \to \msf{SEM}_{\mathcal{S},\Theta}^{\msf{L}'}$ is injective. Indeed, for all $\varphi,\varphi' \in \msf{Sub}(\psi)$, if $\msf{sem}_\Theta(\varphi) = \msf{sem}_\Theta(\varphi')$, since the function $\varphi_\Theta$ satisfies condition (2) of Lemma~\ref{lem:achieved_sat_injective_function}, it follows that $\varphi = \varphi_\Theta(\msf{sem}_\Theta(\varphi)) = \varphi_\Theta(\msf{sem}_\Theta(\varphi')) = \varphi'$. Therefore, we have $|\msf{Sub}(\psi)| \leq |\msf{SEM}_{\mathcal{S},\Theta}^{\msf{L}'}|$. Hence:
	\begin{equation*}
		\msf{sz}(\psi) \leq |\msf{SEM}_{\mathcal{S},\Theta}^{\msf{L}'}| = |\bigcup_{\tau \in \mathcal{T}} \msf{SEM}_{\mathcal{S},\Theta}^{\msf{L}'}(\tau)|\leq |\bigcup_{\tau \in \mathcal{T}} \prod_{M \in \mathcal{S}} \msf{SEM}_M(\tau)| \leq \sum_{\tau \in \mathcal{T}} \prod_{M \in \mathcal{S}} |\msf{SEM}_M(\tau)|
	\end{equation*}
\end{proof}

The proof of Theorem~\ref{thm:size_separation} is now direct.
\begin{proof}
	Consider a $\mathcal{C}$-sample $\mathcal{S} = (\mathcal{P},\mathcal{N})$. Let $\mathcal{R}$ denote the $(\msf{L},\mathcal{S})$-learning property for which $\mathcal{R}_{\mathcal{S}} = \{\mathcal{P}\}$. (Here, $\mathcal{S}$ is seen as the set $\mathcal{S} = \mathcal{P} \cup \mathcal{N}$). Clearly, for all fragments $\msf{L}'$ of the logic $\msf{L}$, a final $\msf{L}'$-formula $\varphi$ is $\mathcal{S}$-separating if and only if $\varphi \in \mathcal{R}$. We can therefore apply Proposition~\ref{prop:size_separation}.
\end{proof}

\subsection{Proof of Theorem~\ref{thm:enumeration_algo}}

\begin{proof}
	Consider a $\mathcal{C}$-sample $\mathcal{S}$. 
	The set $\msf{SEM}$ after exiting the while loop (Lines 1-7) is equal to the $\msf{SEM}_{\mathcal{S},\Theta}^{\msf{L}'}$ from Definition~\ref{def:inductive_def_semantics_sets}. %Note that, in , we keep track of the type $\tau$ in the pairs $(\msf{X},\tau)$ computed. 
	%
	%Consider now Algorithm~\ref{algo:decide_pvln}. Given a , this algorithm runs through all tuples in  $\msf{SEM}_{\mathcal{S},\Theta}^{\msf{L}'}$ and accepts the input $\mathcal{S}$ if and only if there is some some tuple $X \in \msf{SEM}_{\mathcal{S},\Theta}^{\msf{L}'}$ satisfying $\{M \in \mathcal{S} \mid X[M] \in \msf{SAT}_M\} = \mathcal{P}$ is found (the right to left inclusion is checked in lines 4-8 and 15, while the left-to-right inclusion is checked in lines 10-14 and 15). Therefore:
	Then, given what is done in Lines 8-10, we have that:
	\begin{align*}
		\text{Algorithm~\ref{algo:decide_passive_learning} accepts } \mathcal{S} & \Longleftrightarrow \exists \tau \in \mathcal{T}_{\msf{f}},\; \exists \msf{X} \in \msf{SEM}_{\mathcal{S},\Theta}^{\msf{L}}(\tau),\; \{M \in \mathcal{S} \mid \msf{X}[M] \in \msf{SAT}_M\} = \mathcal{P} \\
		& \Longleftrightarrow \exists \tau \in \mathcal{T}_{\msf{f}},\; \exists \varphi \in \msf{Fm}^{\msf{f}}_{\msf{L}}(\tau),\; \{M \in \mathcal{S} \mid \msf{sat}_M(\varphi) \in \msf{SAT}_M\} = \mathcal{P} \\
		& \Longleftrightarrow \exists \varphi \in \msf{Fm}^{\msf{f}}_{\msf{L}},\; \{M \in \mathcal{S} \mid M \models \varphi\} = \mathcal{P} \\
		& \Longleftrightarrow \text{The $\mathcal{C}$-sample $\mathcal{S}$ is $\msf{L}$-separable}
	\end{align*}
	The second equivalence comes from Lemma~\ref{lem:allsat_achieved_by_computation}, Lemma~\ref{lem:achieved_sat_injective_function}, and the definition of the set $\msf{AllSEM}_{\mathcal{S},\Theta}^{\msf{L}}$ (since, for all $\tau \neq \tau' \in \mathcal{T}$, we have $\msf{SEM}_M(\tau) \cap \msf{SEM}_M(\tau') = \emptyset$). The third equivalence comes from the definition of the set $\msf{SAT}_M$ (recall Definition~\ref{def:capturing_semantics}).		
\end{proof}

\section{Proof of Proposition~\ref{prop:modal_logic_complexity}}
\label{proof:prop_modal_logic_complexity}
\begin{proof}
	Consider a $\mathcal{K}(\prop,\act)$-sample $\mathcal{S}$ of Kripke structures.
	
	First of all, there are infinitely many $\ML$-operators of arity 2, since we have $\msf{Op}_{\langle \cdot \rangle} := \{ \langle a \rangle^{\geq k} \mid a \in \act, k \in \N \} \subseteq \msf{Op}_2$. We let $\msf{RelOp}_{\langle \cdot \rangle}(\mathcal{S}) := \{ \langle a \rangle^{\geq k} \mid a \in \act, k \leq \max_{K \in \mathcal{S}} |Q_K| \} \subseteq \msf{Op}_{\langle \cdot \rangle}$ be a set of \emph{relevant operators} of arity 2. This set satisfies the following property: for all $\msf{o} \in \msf{Op}_{\langle \cdot \rangle}$, there is a relevant operator $\msf{o}' \in \msf{RelOp}_{\langle \cdot \rangle}(\mathcal{S})$ such that, for all $K \in \mathcal{S}$, the $\Theta_K$-compatible functions $\msf{sem}_K^\msf{o}$ and $\msf{sem}_K^{\msf{o}'}$ are equal.
	
	Then, we let $\msf{RelOp}_2(\mathcal{S}) := (\msf{Op}_2 \setminus \msf{Op}_{\langle \cdot \rangle}) \cup \msf{RelOp}_{\langle \cdot \rangle}(\mathcal{S})$. Note that we have $|\msf{RelOp}(\mathcal{S})|$ polynomial in $n = \Sigma_{K \in \mathcal{S}} |Q_K|$. Now, we have that looping over the operators in $\msf{RelOp}_2(\mathcal{S})$ instead of $\msf{Op}_2$ (in Line 6 of Algorithm~\ref{algo:decide_passive_learning}) does not change the output of the algorithm (by the above property satisfied by the set of relevant operators $\msf{RelOp}_{\langle \cdot \rangle}(\mathcal{S})$).
	
	Let us now consider the complexity of the algorithm %the instantiation of Algorithm~\ref{algo:decide_passive_learning} to the case of modal logic formulas evaluated on Kripke structures 
	when looping over $\msf{RelOp}_2(\mathcal{S})$ in Line 6 instead of $\msf{Op}_2$:
	\begin{itemize}
		\item The while loop is exited after at most $|\msf{SEM}_{\mathcal{S},\Theta}^{\msf{L}}| \leq |\prod_{M \in \mathcal{S}} \msf{SAT}_M| = 2^n$ steps. Furthermore, as each step of this while loop, the algorithm runs through all operators in $\msf{Op}_1 \cup \msf{RelOp}_2$, and (pairs of) elements of the set $\msf{SEM}$, and applies the $\Theta_K$-functions $\msf{sat}_K$, for all models $K \in \mathcal{S}$. These functions can all be computed in time polynomial in $|Q_K|$ (as can be seen in the proof in Appendix~\ref{proof:lem_modal_logic_pair_inductive_property} of Lemma~\ref{lem:modal_logic_pair_inductive_property}). 
		In addition, the comparison of the sets $\msf{SEM}$ and $\msf{SEM}'$, of size at most $2^n$, can be done in time $2^{O(n)}$. Therefore, running Lines 1-7 can be done in time $2^{O(n)}$.
		\item As for Line 8-10, it amounts to running through the set $\msf{SEM}$, of size at most $|\msf{SEM}_{\mathcal{S},\Theta}^{\msf{L}}| \leq 2^n$. For each of these elements $X \in \msf{SEM}$, for all Kripke structures $K \in \mathcal{S}$, it is checked whether $X_K \in \msf{SAT}_K$, which can be done in time polynomial in $|Q_K|$ (as this amounts to checking a set inclusion). Thus, executing Lines 8-10 can also be done in time $2^{O(n)}$. 
	\end{itemize}
\end{proof}

\section{Another definition of what it means for a pair to satisfy the inductive property}
We make a remark below that gives a reformulation of what it means for a pair to satisfy the inductive property. 
\begin{remark}
	\label{rmk:sufficient_condition}
	Consider a logic $\msf{L}$ and an $\msf{L}$-model $M$. An $(\msf{L},M)$-pair $\Theta_M$ satisfies the inductive property if and only if the following holds, for all types $\tau$. 
	\begin{itemize}
		\item For all $\msf{o} \in \msf{Op}_1(\tau)$, we have:
		\begin{equation*}
			\forall \varphi,\varphi' \in \msf{Fm}_{\msf{L}}(T(\msf{o},1)):\; \msf{sem}_M(\varphi) = \msf{sem}_M(\varphi') \implies \msf{sem}_M(\msf{o}(\varphi)) = \msf{sem}_M(\msf{o}(\varphi'))
		\end{equation*}
		\item For all $\msf{o} \in \msf{Op}_2(\tau)$, we have:
		\begin{align*}
			& \forall \varphi,\varphi' \in \msf{Fm}_{\msf{L}}(T(\msf{o},1)),\; \forall \psi,\psi' \in \msf{Fm}_{\msf{L}}(T(\msf{o},2)) :\; \\
			& (\msf{sem}_M(\varphi),\msf{sem}_M(\psi)) = (\msf{sem}_M(\varphi'),\msf{sem}_M(\psi')) \implies \msf{sem}_M(\msf{o}(\varphi,\psi)) = \msf{sem}_M(\msf{o}(\varphi',\psi'))
		\end{align*}
	\end{itemize}
\end{remark}

\section{Proof of Corollaries~\ref{coro:ltl_logic}}
\label{proof:corollary_ltl}
\begin{proof}
	Let us apply Theorem~\ref{thm:size_separation}. Consider an infinite word $w = u \cdot v^\omega \in \mathcal{W}(\prop)$, and let $n := ||w||-1$. We let $\msf{Pos}(w) := \Interval{0,n}$, and for all $i \in \msf{Pos}(w)$, we let:
	\begin{itemize}
		\item $\msf{Succ}(i) := \{i+1\}$ if $i < n$; otherwise $\msf{Succ}(i) := \emptyset$ if $|w|-1 = |u|-1 = i$; otherwise $\msf{Succ}(i) := \{|u|\}$. This ensures that, for all $j \in \msf{Succ}(i)$: $w[i :] = w[i] \cdot w[j:] \in (2^{\prop})^{\omega}$.
		\item $\msf{After}(i) := \Interval{\min(i,|u|),n}$. This ensures that: $\{ w[j:] \mid i \leq j \} = \{ w[j:] \mid j \in \msf{After}(i) \}$ (in both cases where $w$ is finite or infinite).
		\item For all $k \in \msf{After}(i)$: $\msf{Between}(i,k) := \Interval{i,k-1}$ if $i \leq k$, and $\msf{Between}(i,k) := \Interval{i,n} \cup \Interval{|u|,k-1}$ otherwise (note that the latter case cannot occur if $w$ is finite). This ensures that, letting $\msf{Between}(i,k) = \{ i_0 < i_1 < \ldots < i_x \}$, we have $w[i:] = w[i_0] \cdot w[i_1] \cdots w[i_x] \cdot w[k:]$.
	\end{itemize}
	
	Now, we let $\msf{SEM}_w := 2^{\msf{Pos}(w)}$, and for all $\LTL(\prop)$-formulas $\varphi \in \msf{Fm}_{\LTL(\prop)}$:
	\begin{equation*}
		\msf{sem}_w(\varphi) := \{ i \in \msf{Pos}(w) \mid w[i:] \models \varphi \}
	\end{equation*}
	Let us show that the $(\LTL(\prop),w)$-pair $(\msf{SEM}_w,\msf{sem}_w)$ inductively captures the $\LTL(\prop)$-semantics. It clearly captures the $\LTL(\prop)$-semantics since we have $w  = w[0:]$. Let us now focus on the inductive property. For all $\LTL(\prop)$-formulas $\varphi \in \msf{Fm}_{\LTL(\prop)}$, we have:
	\begin{itemize}
		\item $\msf{sem}_w(\neg \varphi) = \msf{Pos}(w) \setminus \msf{sem}_w(\varphi)$;
		\item $\msf{sem}_w(\lX \varphi) = \{ i \in \msf{Pos}(w) \mid \msf{Succ}(i) \cap \msf{sem}_w(\varphi) \neq  \emptyset \}$;
		\item $\msf{sem}_w(\lF \varphi) = \{ i \in \msf{Pos}(w) \mid \msf{After}(i) \cap \msf{sem}_w(\varphi) \neq \emptyset \}$;
		\item $\msf{sem}_w(\lG \varphi) = \{ i \in \msf{Pos}(w) \mid \msf{After}(i) \subseteq \msf{sem}_w(\varphi) \}$.
	\end{itemize}
	Thus, for all $\msf{o} \in \msf{Op}_1$, for all $(\varphi,\varphi') \in (\msf{Fm}_{\LTL(\prop)})^2$, we have:
	\begin{equation*}
		\msf{sem}_w(\varphi) = \msf{sem}_w(\varphi') \implies \msf{sem}_w(\msf{o} \; \varphi) = \msf{sem}_w(\msf{o} \; \varphi')
	\end{equation*}
	
	Furthermore, for all pairs of $\LTL(\prop)$-formulas $(\varphi_1,\varphi_2) \in (\msf{Fm}_{\LTL(\prop)})^2$, we have:
	\begin{itemize}
		\item $\msf{sem}_w(\varphi_1 \wedge \varphi_2) =  \msf{sem}_w(\varphi_1) \cap \msf{sem}_w(\varphi_2)$;
		\item $\msf{sem}_w(\varphi_1 \vee \varphi_2) =  \msf{sem}_w(\varphi_1) \cup \msf{sem}_w(\varphi_2)$;
		\item $\msf{sem}_w(\varphi_1 \lU \varphi_2) = \{ i \in \msf{Pos}(w) \mid \exists k \in \msf{After}(i) \cap \msf{sem}_{w}(\varphi_2):\; \msf{Between}(i,k) \subseteq  \msf{sem}_{w}(\varphi_1) \}$.
	\end{itemize}
	Thus, for all $\msf{o} \in \msf{Op}_2$, for all $(\varphi_1,\varphi_1',\varphi_2,\varphi_2') \in (\msf{Fm}_{\LTL(\prop)})^4$, we have:
	\begin{equation*}
		(\msf{sem}_w(\varphi_1),\msf{sem}_w(\varphi_2)) = (\msf{sem}_w(\varphi_1'),\msf{sem}_w(\varphi_2')) \implies \msf{sem}_w(\varphi_1 \; \msf{o} \; \varphi_2) = \msf{sem}_w(\varphi_1' \; \msf{o} \; \varphi_2')
	\end{equation*}
	Hence, by Remark~\ref{rmk:sufficient_condition}, the pair the $(\LTL(\prop),w)$-pair $(\msf{SEM}_w,\msf{sem}_w)$ satisfies the inductive property. Therefore, by Theorem~\ref{thm:size_separation}, for all $\LTL$-fragments $\msf{L}$, for all inputs $\mathcal{S} = (\mathcal{P},\mathcal{N})$ of the decision problem $\msf{PvLn}(\msf{L},\mathcal{W}(\prop))$, if there is an $\msf{L}$-separating formula, there is one of size at most:
	\begin{equation*}
		\prod_{w \in S} |\msf{SEM}_{w}| = \prod_{w \in S} 2^{|w|} = 2^{\sum_{w \in S} |w|}
	\end{equation*}
\end{proof}

\section{Proof of Corollary~\ref{coro:ctl_logic}}
\label{proof:corollary_ctl}
As mentioned in the main part of the paper, we are actually going to show the result of Corollary~\ref{coro:ctl_logic} with $\ATL$-formulas evaluated on concurrent game structures, which are more general than $\CTL$-formulas evaluated on actionless Kripke structures. That way, we recover (in a straightforward manner) the results previously established in \cite{DBLP:conf/fm/BordaisNR24} about the minimal size of separating $\ATL$-formulas, thus illustrating the usefulness of Theorem~\ref{thm:size_separation}. For the remainder of this section, we use the formalism (for concurrent game structures and $\ATL$-formulas) of \cite{DBLP:conf/fm/BordaisNR24}, that we recall almost verbatim below. 

Let us first introduce formally the concurrent game structures on which $\ATL$-formulas will be evaluated. We let $\N_1 := \{ i \in \N \mid i \geq 1\}$.
\begin{definition}
	\label{def:CGS}
	A concurrent game structure is described by a tuple $C = (Q,I,k,P, \pi, d,\delta)$ where $Q$ is the finite set of states, $I \subseteq Q$ is the set of initial states, $k \in \N_1$ is the number of agents, $P$ is the finite non-empty set of propositions (or observations), $\pi: Q\mapsto 2^{P}$ maps each state $q\in Q$ to the set of propositions $\pi(q) \subseteq P$ that hold in $q$, $d: Q \times \Interval{1,k} \rightarrow \N_1$ maps each state and agent to the number of actions available to that agent at that state, and $\delta: Q_\msf{Act} \rightarrow Q$ maps every state and tuple of one action per agent to the successor state, where $Q_\msf{Act} := \cup_{q \in Q} Q_\msf{Act}(q)$ with $Q_\msf{Act}(q) := \{ (q,\alpha_1,\ldots,\alpha_k) \mid \forall a \in \Ag,\; \alpha_a \in \brck{1,d(q,a)} \}$. 
	\iffalse
	\begin{itemize}
		\item $Q$ is the finite set of states;
		\item $I \subseteq Q$ is the set of initial states;
		\item $k \in \N_1$ is the number of agents, we denote by $\Ag := \brck{1,k}$ the set of $k$ agents;
		\item $\prop$ is the finite set of observations or what we call propositions;
		\item $\pi: Q\mapsto 2^{\mathcal{P}}$ %is a function that 
		maps each state $q\in Q$ to the set of propositions $\mathcal{P}'\subseteq \mathcal{P}$ in $q$;
		%\item $\sigma: Q\mapsto \{1,\cdots k\}$ is a function that maps each state $s\in Q$ to a agent $a\in \{1,\cdots, k\}$ that is supposed to play in $s$;
		\item $d: Q \times \brck{1,k} \rightarrow \mathbb{N}^+$ maps each state and agent to the number of actions available to that agent at that state;
		\item $\delta: Q_\msf{Act} \rightarrow Q$ is the function mapping every state and tuple of one action per agent to the next state, where $Q_\msf{Act} := \cup_{q \in Q} Q_\msf{Act}(q)$ with $Q_\msf{Act}(q) := \{ (q,\alpha_1,\ldots,\alpha_i) \mid \forall a \in \Ag,\; \alpha_a \in \brck{1,d(q,a)} \}$. 
	\end{itemize}
	\fi
	
	For all $q \in Q$ and $A \subseteq \Ag$, we let $\msf{Act}_A(q) := \{ \alpha = (\alpha_a)_{a \in A} \in \prod_{a \in A} %\mid \forall a \in \Ag,\; \alpha_a \in 
	\{a\} \times \brck{1,d(q,a)} \}$. Then, for all tuple $\alpha% = (\alpha_a)_{a \in A} 
	\in \msf{Act}_A(q)$ of one action per agent in $A$, we let: 
	\begin{equation*}
		\msf{Succ}(q,\alpha) := \{ q' \in Q \mid \exists \alpha' %= (\alpha'_a)_{a \in \Ag \setminus A} 
		\in \msf{Act}_{\Ag \setminus A}(q),\; \delta(q,(\alpha,\alpha')) = q' \}
	\end{equation*}
\end{definition}
Unless otherwise stated, a CGS $C$ refers to the tuple $C = (Q,I,k,P,\pi,d,\delta)$.

In a concurrent game structure, a strategy for an agent prescribes what to do as a function of the finite sequence of states seen so far. Moreover, given a coalition of agents and a tuple of one strategy per agent in the coalition, we define the set of infinite sequences of states that can occur with this tuple of strategies. This is defined below.
\begin{definition}
	Consider a concurrent game structure $C$ and an agent $a \in \Ag$. A \emph{strategy} for Agent $a$ is a function $s_a: Q^+ \rightarrow \N_1$ such that, for all $\rho = \rho_0 \dots \rho_n \in Q^+$, we have $s_a(\rho) \leq d(\rho_n,a)$. We let $\msf{S}_a$ denote the set of strategies available to Agent $a% \in \brck{1,k}
	$.
	
	Given a coalition (or subset) of agents $A \subseteq \Ag$, a \emph{strategy profile} for the coalition $A$ is a tuple $s = (s_a)_{a \in A}$ of one strategy per agent in $A$. We denote by $\msf{S}_A$ the set of strategy profiles for the coalition $A$. %Given any such strategy profile 
	For all $s \in \msf{S}_A$ %, for all states 
	and $q \in Q$, we let $\msf{O}(q,s) \subseteq Q^\omega$ denote the set of infinite paths $\rho$ compatible with %the strategy profile 
	$s$ from $q$:
	\begin{equation*}
		\msf{O}(q,s) := \{ \rho \in Q^\omega \mid \rho[0] = q,\; \forall i \in \N, \rho[i+1] \in \msf{Succ}(\rho[i],(s_a(\rho[0] \cdots \rho[i]))_{a \in A}) \}
	\end{equation*}
	%Then, we let:
	%\begin{equation*}
	%	\msf{Out}^\prop(q,s) := \{ \pi^\omega(\rho) \in (2^\prop)^\omega \mid \rho \in \msf{Out}^Q(q,s) \}
	%\end{equation*}
	%denote the set of infinite sequences of observations that are compatible with the strategy profile $s$ from $q$.
\end{definition}

Let us define below the syntax, models, and semantics of the logic ATL.
\begin{definition}
	\label{def:atl}
	Consider a non-empty finite set of propositions $\prop$. The $\ATL(\prop)$-syntax is as follows, with $p\in \prop$ and $A \subseteq \N$ is a coalition of players:
	\begin{equation*}
		\varphi ::= p \mid \neg \varphi \mid \varphi \vee \varphi \mid \varphi \wedge \varphi \mid \fanBrOp{A} \lX \varphi \mid \fanBrOp{A} \lF \varphi \mid \fanBrOp{A} \lG \varphi \mid \fanBrOp{A} (\varphi \lU \varphi)
	\end{equation*}
	The models $\mathcal{C}(\prop)$ on which $\ATL(\prop)$-formulas are evaluated are the concurrent game structures $C$ for which $P \subseteq \prop$. Formally:
	\begin{equation*}
		\mathcal{C}(\prop) := \{ C = (Q,I,k,P, \pi, d,\delta) \mid P \subseteq \prop \}
	\end{equation*}
	
	Given a concurrent game structure $C \in \mathcal{C}(\prop)$, and a $\msf{CTL}(\prop)$-formula $\varphi$, we define when $\varphi$ is satisfied by a state $q$ inductively as follows:
	
	\begin{minipage}[t]{0.2\textwidth}
		\begin{align*}
			& q \models p & \text{ iff } & p \in \pi(q) \\
			& q \models \neg \varphi & \text{ iff } & q \not\models \varphi \\
			& q \models \varphi_1 \vee \varphi_2 & \text{ iff } & q \models \varphi_1 \text{ or } q \models \varphi_2\\
			& q \models \fanBrOp{A} \lF \varphi & \text{ iff } & \exists s \in \msf{S}_A,\; \forall \rho \in \msf{O}(q,s),\; \\
			&&& \exists i \in \N,\; \rho[i:] \models \varphi \\
		\end{align*}
	\end{minipage}
	\begin{minipage}[t]{0.05\textwidth}
		\begin{align*}
			\mid \\
			\mid \\
			\mid \\
			\mid \\
			\mid \\
		\end{align*}
	\end{minipage}
	\begin{minipage}[t]{0.5\textwidth}
		\begin{align*}
			& q \models \fanBrOp{A} \lX \varphi & \text{ iff } & \exists s \in \msf{S}_A,\; \forall \rho \in \msf{O}(q,s),\; \\
			&&& \rho[1:] \models \varphi \\
			& q \models \varphi_1 \wedge \varphi_2 & \text{ iff } & q \models \varphi_1 \text{ and } q \models \varphi_2\\
			& q \models \fanBrOp{A} \lG \varphi & \text{ iff } & \exists s \in \msf{S}_A,\; \forall \rho \in \msf{O}(q,s),\; \\
			&&& \forall i \in \N,\; \rho[i:] \models \varphi \\
		\end{align*}
	\end{minipage}
	Furthermore:
	\begin{equation*}
		q \models \fanBrOp{A} (\varphi_1 \lU \varphi_2) \text{ iff } \exists s \in \msf{S}_A,\; \forall \rho \in \msf{O}(q,s),\; 
		\exists i \in \N,\; \rho[i:] \models \varphi_2,\; \forall j \leq i-1,\; \rho[j:] \models \varphi_1
	\end{equation*}
	Then, a concurrent game structure $C$ satisfies an $\ATL$-formula $\varphi$, i.e. $C \models \varphi$, if, for all $q \in I$, we have $q \models \varphi$. 
\end{definition}

Let us now state the corollary generalizing Corollary~\ref{coro:ctl_logic} to the case of $\ATL$-formulas evaluated on concurrent game structures.
\begin{corollary}
	\label{coro:atl_logic}
	Consider a non-empty set of propositions $\prop$, and any ATL-fragment $\msf{L}$ of $\msf{ATL}(\prop)$. Then, for all inputs $\mathcal{S} = (\mathcal{P},\mathcal{N})$ of the decision problem $\msf{PvLn}(\msf{L},\mathcal{C}(\prop))$, if there is a separating formula, there is one of size at most $2^{n}$, with $n := \sum_{C \in \mathcal{S}} |Q_C|$.
\end{corollary}
Before we prove this corollary, let us prove that it implies Corollary~\ref{coro:ctl_logic} (with $\CTL$-formulas). 
\begin{proof}[Proof of Corollary~\ref{coro:ctl_logic}]
	Consider an input $\mathcal{S} = (\mathcal{P},\mathcal{N})$ of the decision problem $\msf{PvLn}(\msf{L},\mathcal{K}(\prop))$. Any actionless Kripke structure $K = (Q,I,\delta,P,\pi)$ in $\mathcal{K}(\prop)$ can be seen as a concurrent game structures $C = (Q,I,k,P,\pi,d,\delta)$ with only one agent, i.e. $k = 1$. In this context, the $\ATL$-quantifier $\fanBrOp{A}$ corresponds to the $\CTL$-path quantifier $\exists$ if $1 \in A$, and to the $\CTL$-path quantifier $\forall$ if $1 \notin A$. The bound of Corollary~\ref{coro:atl_logic} can then be applied directly to the case of Kripke structures and $\CTL$-formulas.
\end{proof}

Let us now proceed to the proof of Corollary~\ref{coro:atl_logic}. 
\begin{proof}[Proof of Corollary~\ref{coro:atl_logic}]
	Let us apply Theorem~\ref{thm:size_separation}. Consider any concurrent game structure $C = (Q_C,I,k,P,\pi,d,\delta)$. We let $\msf{SEM}_C := 2^{Q_C}$, and for all $\ATL(\prop)$-formulas $\varphi \in \msf{Fm}_{\ATL(\prop)}$:
	\begin{equation*}
		\msf{sem}_C(\varphi) := \{ q \in Q_C \mid q \models \varphi \}
	\end{equation*}
	Let us show that the $(\ATL(\prop),C)$-pair $(\msf{SEM}_C,\msf{sem}_C)$ inductively captures the $\ATL(\prop)$-semantics. First of all, it captures the $\ATL(\prop)$-semantics for exactly the same reason that the $(\ML(\prop,\act),K)$-pair $\Theta_K$ (from Example~\ref{ex:Kripke_structures_SEM}) captures the $\ML(\prop,\act)$-semantics (recall Lemma~\ref{lem:modal_logic_pair_captures_semantics}). Let us now focus on the inductive property. Consider two $\ATL$-formulas $\varphi,\varphi'$ such that $\msf{sem}_{C}(\varphi) = \msf{sem}_{C}(\varphi')$. Then, for all $q \in Q_C$, and coalitions of agents $A \subseteq \llbracket 1,k \rrbracket$, we have:
	\begin{align*}
		q \models \fanBrOp{A} \lX \varphi & \Longleftrightarrow \exists s \in \msf{S}_A,\; \forall \rho \in \msf{O}(q,s),\; \rho[1:] \models \varphi \\ 
		& \Longleftrightarrow \exists s \in \msf{S}_A,\; \forall \rho \in \msf{O}(q,s),\; \rho[1:] \in \msf{sem}_C(\varphi) \\ 
		& \Longleftrightarrow \exists s \in \msf{S}_A,\; \forall \rho \in \msf{O}(q,s),\; \rho[1:] \in \msf{sem}_C(\varphi') \\ 
		& \Longleftrightarrow \exists s \in \msf{S}_A,\; \forall \rho \in \msf{O}(q,s),\; \rho[1:] \models \varphi' \\ 
		& \Longleftrightarrow \exists s \in \msf{S}_A,\; q \models \fanBrOp{A} \lX \varphi'
	\end{align*}
	This is similar for the temporal operator $\lF$:
	\begin{align*}
		q \models \fanBrOp{A} \lF \varphi & \Longleftrightarrow \exists s \in \msf{S}_A,\; \forall \rho \in \msf{O}(q,s),\; \exists i \in \N,\; \rho[i:] \models \varphi \\ 
		& \Longleftrightarrow \exists s \in \msf{S}_A,\; \forall \rho \in \msf{O}(q,s),\; \exists i \in \N,\; \rho[i:] \in \msf{sem}_C(\varphi) \\ 
		& \Longleftrightarrow \exists s \in \msf{S}_A,\; \forall \rho \in \msf{O}(q,s),\; \exists i \in \N,\; \rho[i:] \in \msf{sem}_C(\varphi') \\ 
		& \Longleftrightarrow \exists s \in \msf{S}_A,\; \forall \rho \in \msf{O}(q,s),\; \exists i \in \N,\; \rho[i:] \models \varphi' \\ 
		& \Longleftrightarrow \exists s \in \msf{S}_A,\; q \models \fanBrOp{A} \lF \varphi'
	\end{align*}
	The case of the temporal operator $\lG$ is identical to the case of the temporal $\lF$, except that $\exists i \in \N$ is replaced by $\forall i \in \N$. 
	
	Consider now four $\ATL$-formulas $\varphi_1,\varphi_2,\varphi_2',\varphi_1'$ such that $\msf{sem}_{C}(\varphi_1) = \msf{sem}_{C}(\varphi_1')$ and $\msf{sem}_{C}(\varphi_2) = \msf{sem}_{C}(\varphi_2')$. Then, for all $q \in Q_C$, we have:
	\begin{align*}
		q \models \fanBrOp{A} \varphi_1 \lU \varphi_2 & \Longleftrightarrow \exists s \in \msf{S}_A,\; \forall \rho \in \msf{O}(q,s),\; 
		\exists i \in \N,\; \rho[i:] \models \varphi_2,\; \forall j \leq i-1,\; \rho[j:] \models \varphi_1 \\
		& \Longleftrightarrow \exists s \in \msf{S}_A,\; \forall \rho \in \msf{O}(q,s),\; 
		\exists i \in \N,\; \rho[i:] \in \msf{sem}_C(\varphi_2),\; \forall j \leq i-1,\; \rho[j:] \models \msf{sem}_C(\varphi_1) \\
		& \Longleftrightarrow \exists s \in \msf{S}_A,\; \forall \rho \in \msf{O}(q,s),\; 
		\exists i \in \N,\; \rho[i:] \in \msf{sem}_C(\varphi_2'),\; \forall j \leq i-1,\; \rho[j:] \models \msf{sem}_C(\varphi_1') \\
		& \Longleftrightarrow \exists s \in \msf{S}_A,\; \forall \rho \in \msf{O}(q,s),\; 
		\exists i \in \N,\; \rho[i:] \models \varphi_2',\; \forall j \leq i-1,\; \rho[j:] \models \varphi_1' \\
		& \Longleftrightarrow \exists s \in \msf{S}_A,\; q \models \fanBrOp{A} \varphi_1' \lU \varphi_2' 
	\end{align*}
	
	We can therefore apply Remark~\ref{rmk:sufficient_condition} (since propositional operators straightforwardly satisfy the condition of this lemma) and obtain that the $(\ATL(\prop),C)$-pair $(\msf{SEM}_C,\msf{sem}_C)$ satisfies the inductive property. Therefore, by Theorem~\ref{thm:size_separation}, for all $\ATL$-fragments $\msf{L}$, for all inputs $\mathcal{S} = (\mathcal{P},\mathcal{N})$ of the decision problem $\msf{PvLn}(\msf{L},\mathcal{C}(\prop))$, if there is an $\msf{L}$-separating formula, there is one of size at most:
	\begin{equation*}
		\prod_{C \in S} |\msf{SEM}_{C}| = \prod_{C \in S} 2^{|Q_C|} = 2^{\sum_{C \in S} |Q_C|}
	\end{equation*}
\end{proof}

\section{Proof of Corollary~\ref{coro:pctl_logic}}
\label{proof:corollary_pctl} 
\begin{proof}
	Let us apply Theorem~\ref{thm:size_separation}. Consider any Markov chain $N = (Q_N,I_N,\mathbb{P}_N,P,\pi)$. We let $\msf{SEM}_N := 2^{Q_N}$, and for all $\PCTL(\prop)$-formulas $\varphi \in \msf{Fm}_{\PCTL(\prop)}$:
	\begin{equation*}
		\msf{sem}_N(\varphi) := \{ q \in Q_N \mid q \models \varphi \} \in \msf{SEM}_N
	\end{equation*}
	Let us show that the $(\PCTL(\prop),N)$-pair $(\msf{SEM}_N,\msf{sem}_N)$ inductively captures the $\PCTL(\prop)$-semantics. First of all, it captures the $\PCTL(\prop)$-semantics since, for all $\varphi \in \msf{Fm}_{\PCTL}$, we have $N \models \varphi$ if and only if $I \subseteq \msf{sem}_N(\varphi)$. Consider now the inductive property. Consider two $\PCTL(\prop)$-formulas $\varphi,\varphi'$ such that $\msf{sem}_{N}(\varphi) = \msf{sem}_{N}(\varphi')$. Then, for all $q \in Q_N$, $\bowtie \in \{\geq,>,\leq,<,=,\neq\}$ and $r \in \mathbb{Q}$, we have:
	\begin{align*}
		q \models \mathbb{P}_{\bowtie r}(\lX \varphi) \Longleftrightarrow \overline{\mathbb{P}_{N}}(q,Q \cdot \msf{sem}_N(\varphi)) \bowtie r \Longleftrightarrow \overline{\mathbb{P}_{N}}(q,Q \cdot \msf{sem}_N(\varphi')) \bowtie r \Longleftrightarrow q \models \mathbb{P}_{\bowtie r}(\lX \varphi')
	\end{align*}
	and:
	\begin{align*}
		q \models \mathbb{P}_{\bowtie r}(\lF \varphi) \Longleftrightarrow \overline{\mathbb{P}_{N}}(q,Q^* \cdot \msf{sem}_N(\varphi)) \bowtie r \Longleftrightarrow \overline{\mathbb{P}_{N}}(q,Q^* \cdot \msf{sem}_N(\varphi')) \bowtie r \Longleftrightarrow q \models \mathbb{P}_{\bowtie r}(\lX \varphi')
	\end{align*}
	and
	\begin{align*}
		q \models \mathbb{P}_{\bowtie r}(\lG \varphi) \Longleftrightarrow \overline{\mathbb{P}_{N}}(q,(\msf{sem}_N(\varphi))^\omega) \bowtie r \Longleftrightarrow \overline{\mathbb{P}_{N}}(q,(\msf{sem}_N(\varphi'))^\omega) \bowtie r \Longleftrightarrow q \models \mathbb{P}_{\bowtie r}(\lG \varphi')
	\end{align*}
	
	Consider now four $\PCTL$-formulas $\varphi_1,\varphi_2,\varphi_2',\varphi_1'$ such that $\msf{sem}_{N}(\varphi_1) = \msf{sem}_{N}(\varphi_1')$ and $\msf{sem}_{N}(\varphi_2) = \msf{sem}_{N}(\varphi_2')$. Then, for all $q \in Q_C$, we have:
	\begin{align*}
		q \models \mathbb{P}_{\bowtie r}(\varphi_1 \lU \varphi_2) & \Longleftrightarrow \overline{\mathbb{P}_{N}}(q,(\msf{sem}_N(\varphi_1))^* \cdot \msf{sem}_N(\varphi_2)) \bowtie r \\
		& \Longleftrightarrow \overline{\mathbb{P}_{N}}(q,(\msf{sem}_N(\varphi_1'))^* \cdot \msf{sem}_N(\varphi_2')) \bowtie r \\
		& \Longleftrightarrow q \models \mathbb{P}_{\bowtie r}(\varphi_1' \lU \varphi_2')
	\end{align*}
	
	We can therefore apply Remark~\ref{rmk:sufficient_condition} (since propositional operators straightforwardly satisfy the condition of this lemma) and obtain that the $(\PCTL(\prop),N)$-pair $(\msf{SEM}_N,\msf{sem}_N)$ satisfies the inductive property. Therefore, by Theorem~\ref{thm:size_separation}, for all $\PCTL$-fragments $\msf{L}$, for all inputs $\mathcal{S} = (\mathcal{P},\mathcal{N})$ of the decision problem $\msf{PvLn}(\msf{L},\mathcal{N}(\prop))$, if there is an $\msf{L}$-separating formula, there is one of size at most:
	\begin{equation*}
		\prod_{N \in S} |\msf{SEM}_{N}| = \prod_{N \in S} 2^{|Q_N|} = 2^{\sum_{N \in S} |Q_N|}
	\end{equation*}
\end{proof}

\section{Proof of Subsection~\ref{subsec:LTL_case_study}}

\subsection{There does not exist a pair that inductively captures the semantics}
\label{appen:discussion_inexistence}
\begin{proposition}
	Let $K$ denote the Kripke structure depicted in Figure~\ref{fig:Kripke_example}. There is no $(\LTL(\{a,b\}),K)$-pair that inductively capture the $\LTL(\{a,b\})$-semantics.
\end{proposition}
\begin{proof}
	Consider any $(\LTL(\{a,b\}),K)$-pair $\Theta_K = (\msf{SEM}_K,\msf{sem}_K)$ that satisfies the inductive property. Note that, by assumption, the set $\msf{SEM}_K$ is finite. 
	
	Now, for all $n \geq 1$, we let $\varphi_n^a := \lX^n a$ and $\varphi_n^b := \lX^n b$ be two size-$(n+1)$ $\LTL(\{a,b\})$-formulas. Since $\msf{SEM}_K$ is finite, there is some $1 \leq i < j$ such that $\msf{sem}_K(\varphi_i^a) = \msf{sem}_K(\varphi_j^a)$, which implies $\msf{sem}_K(\varphi_i^a \vee \varphi_i^b) = \msf{sem}_K(\varphi_i^a \vee \varphi_j^b)$. However, we have that:
	\begin{itemize}
		\item $q_1 \models_{\msf{s}} \varphi_i^a \vee \varphi_i^b$; indeed, for all paths $\rho \in \msf{Paths}(q_1)$, we have $\rho[i] \in \{q_2,q_3\}$, hence $\rho \models  \varphi_i^a \vee \varphi_i^b$.
		\item $q_1 \not\models_{\msf{s}} \varphi_i^a \vee \varphi_j^b$; indeed, letting $\rho := q_1 \cdot (q_3)^{i-1} \cdot q_2 \cdot q_3^\omega \in Q^\omega$, we have $\rho \in \msf{Paths}(q_1)$, $\rho[i] = q_2 \not\models a$ and $\rho[j] = q_3 \not\models b$, thus $\rho \not\models \varphi_i^a \vee \varphi_j^b$. 
	\end{itemize}
	It follows that the pair $\Theta_K$ does not capture the $\LTL$-semantics.
\end{proof}

\subsection{Proof of Proposition~\ref{prop:LTL_X_into_LTL_P}}
\label{proof:prop_LTL_X_into_LTL_P}
We say that two ($\msf{L}_{\lX}(\prop)$- or $\LTL_P(\prop)$-) formulas $\varphi,\varphi'$ are equivalent if for all $K \in \mathcal{T}(\prop)$, we have: $K \models \varphi$ if and only if $K \models \varphi'$. Let us now proceed to the proof of Proposition~\ref{prop:LTL_X_into_LTL_P}.
\begin{proof}
	For all $\varphi_1,\varphi_2 \in \msf{Fm}_{\msf{L}_{\lX}}$, we have $\lX (\varphi_1 \wedge \varphi_2)$ equivalent to $(\lX \varphi_1 \wedge \lX\varphi_2)$ and $\lX (\varphi_1 \vee \varphi_2)$ equivalent to $(\lX \varphi_1 \vee \lX\varphi_2)$. In addition, for all $\varphi_1,\varphi_2 \in \msf{Fm}_{\msf{L}_{\lX}}$, we have $\neg \lX \varphi$ equivalent to $\lX \neg \varphi$, $\neg (\varphi_1 \vee \varphi_2)$ equivalent to $(\neg \varphi_1 \wedge \neg \varphi_2)$, and $\neg (\varphi_1 \wedge \varphi_2)$ equivalent to $(\neg \varphi_1 \vee \neg \varphi_2)$. Therefore, for all $\msf{L}_{\lX}$-formulas $\varphi$, there is some formula $\overline{\varphi} = \bigwedge_{i = 1}^n \overline{\varphi_i}$ that is equivalent to $\varphi$ such that $1 \leq n \in \N$ and for all $1 \leq i \leq n$, we have $\overline{\varphi_i} := \bigvee_{1 \leq j \leq n_i} \lX^{k_{i,j}} x_{i,j}$, with, for all $1 \leq j \leq n_i$, $k_{i,j} \in \N$ and $x_{i,j} \in \{ p,\neg p \mid p \in \prop \}$. Let $1 \leq i \leq n$ and assume without loss of generality that we have $k_{i,1} \leq k_{i,2} \leq \ldots \leq k_{i,n_i}$. For all $2 \leq j \leq n_i$, let $l_{i,j} := k_{i,j} - k_{i,j-1} \geq 0$. Then, the formula $\varphi_{i}' := \lX^{k_{i,1}} (x_{i,1} \vee (\lX^{l_{i,2}} (x_{i,2} \vee (\lX^{l_{i,3}} (x_{i,3} \vee \ldots (x_{i,n_i-1} \vee \lX^{l_{i,n_i}} x_{i,n_i}))))))$ is equivalent to the formula $\overline{\varphi_i}$ and is a $\LTL_P$-formula. Therefore, the formula $\varphi' := \bigwedge_{i = 1}^n \varphi_i'$ is a $\LTL_P$-formula that is equivalent to $\varphi$.
\end{proof}

\subsection{Proof of Corollary~\ref{coro:ltl_P_logic}}
\label{proof:coro_ltl_P_logic}
\begin{proof}
	Let us apply Theorem~\ref{thm:size_separation}. Consider an actionless non-blocking Kripke structure $K = (Q,I,\delta,P,\pi) \in \mathcal{T}(\prop)$. Contrary to the other logics, $\LTL_P$ has two types of formulas $\tau$ and $\tau_P$, corresponding to the lines $\varphi$ and $\varphi_P$ respectively is the grammar defining the logic $\LTL_P$. Thus, we need to define the set of semantic values for the types $\tau$ and $\tau_P$. In both cases, it is equal to $2^{Q_K}$, i.e. we let $\msf{SEM}_K(\tau) := 2^{Q_K}$ and $\msf{SEM}_K(\tau_P) := 2^{Q_K}$ (these sets are typed so that their intersection is empty). For all $x \in \{\tau,\tau_P\}$ and $\LTL_P(\prop)$-formulas $\varphi \in \msf{Fm}_{\LTL(\prop)}(x)$, we let:
	\begin{equation*}
		\msf{sem}_K(\varphi) := \{ q \in Q_K \mid q \models_{\msf{s}} \varphi \} \in \msf{SEM}_K(x)
	\end{equation*}
	Let us show that the $(\LTL_P(\prop),K)$-pair $(\msf{SEM}_K,\msf{sem}_K)$ inductively captures the $\LTL_P(\prop)$-semantics. It straightforwardly captures the $\LTL_P(\prop)$-semantics. Let us now focus on the inductive property. It is also straightforward that the $(\LTL_P(\prop),K)$-pair $(\msf{SEM}_K,\msf{sem}_K)$ satisfies this property when considering only formulas in $\msf{Fm}_{\LTL_P(\prop)}(\tau_P)$. Consider now some $\LTL_P$-formula $\varphi \in \msf{Fm}_{\LTL_P(\prop)}(\tau)$. We have that:
	\begin{itemize}
		\item Straightforwardly, $\msf{sem}_K(\lX \varphi) = \{ q \in Q \mid \delta(q) \subseteq \msf{sem}_K(\varphi) \}$.
		\item Furthermore, letting $S := \{ q \in Q \mid \forall \rho \in \msf{Paths}(q),\; \forall i \in \N,\; \rho[i] \in \msf{sem}_K(\varphi) \}$, we have $\msf{sem}_K(\lG \varphi) = S$. Indeed, consider any $q \in \msf{sem}_K(\lG \varphi)$ and any $\rho \in \msf{Paths}(q)$. Then, by definition $\rho \models \lG \varphi$. Consider any $i \in \N$ and $\rho' \in \msf{Paths}(\rho[i])$. Since $\rho[:i-1] \cdot \rho' \in \msf{Paths}(q)$, and thus $\rho[:i-1] \cdot \rho' \models \lG \varphi$, it follows that $\rho' \models \varphi$. This holds for all $\rho' \in \msf{Paths}(\rho[i])$, thus $\rho[i] \in \msf{sem}_K(\varphi)$. This holds for all $i \in \N$ and $\rho \in \msf{Paths}(q)$. Thus, $q \in S$. In fact, we have $\msf{sem}_K(\lG \varphi) \subseteq S$. Consider now some $q \in S$. Let $\rho \in \msf{Paths}(q)$ and $i \in \N$. By definition of $S$, we have $\rho[i] \in \msf{sem}_K(\varphi)$, thus $\rho[i:] \models \varphi$. As this holds for all $i \in \N$, it follows that $\rho \models \lG \varphi$. Hence, $q \in \msf{sem}_K(\lG \varphi)$. In fact, $S \subseteq \msf{sem}_K(\lG \varphi)$. 
		%\item $\msf{sem}_K(\lG \varphi)$ is the largest subset $S \subseteq \msf{sem}_K(\varphi)$ such that, for all $q \in S$, we have $\delta(q) \subseteq S$.  
	\end{itemize}
	
	Furthermore, for all $\LTL_P$-formulas $\varphi_1,\varphi_2$, we straightforwardly have that: 
	\begin{equation*}
		\msf{sem}_K(\varphi_1 \wedge \varphi_2) = \msf{sem}_K(\varphi_1) \cap \msf{sem}_K(\varphi_2)
	\end{equation*}
	
	In addition, consider some $\LTL_P$-formula $\varphi \in \msf{Fm}_{\LTL_P(\prop)}(\tau_P)$. Letting $T := \{ q \in Q \mid \forall \rho \in \msf{Paths}(q),\; \exists i \in \N,\; \rho[i] \in \msf{sem}_K(\varphi) \}$, we have $\msf{sem}_K(\lF \varphi) = T$. Indeed, consider any $q \in \msf{sem}_K(\lF \varphi)$ and any $\rho \in \msf{Paths}(q)$. Then, by definition $\rho \models \lF \varphi$. Hence, there is some $i \in \N$ such that $\rho[i:] \models \varphi$. Since $\varphi$ is of type $\tau_P$, this is equivalent $\rho[i] \in \msf{sem}_K(\varphi)$. This holds for all $\rho \in \msf{Paths}(q)$. Thus, $q \in T$. In fact, we have $\msf{sem}_K(\lF \varphi) \subseteq T$. Consider now some $q \in T$. Let $\rho \in \msf{Paths}(q)$. By definition of $T$, there is some $i \in \N$ such that $\rho[i] \in \msf{sem}_K(\varphi)$, thus $\rho[i:] \models \varphi$. Hence, $\rho \models \lF \varphi$. Hence, $q \in \msf{sem}_K(\lF \varphi)$. In fact, $T \subseteq \msf{sem}_K(\lF \varphi)$. 
	
	%We have that $\msf{sem}_K(\lF \varphi)$ is the least subset $\msf{sem}_K(\varphi) \subseteq S \subseteq Q$ such that, for all $q \in Q$, if $\delta(q) \subseteq S$ then $q \in S$.  
	
	Furthermore, consider some $\LTL_P$-formula $\varphi'$ of type $\tau$. We have the following. 
	\begin{itemize}
		\item We have $\msf{sem}_K(\varphi \vee \varphi') = \msf{sem}_K(\varphi) \cup \msf{sem}_K(\varphi')$. Indeed, consider any $q \in \msf{sem}_K(\varphi \vee \varphi')$. If $q \notin \msf{sem}_K(\varphi)$, then, since $\varphi$ is of type $\tau_P$, for all $\rho \in \msf{Paths}(q)$, we have $\rho \not\models \varphi$, and thus $\rho \models \varphi'$. That is $q \in \msf{sem}_K(\varphi')$. Hence, we have $\msf{sem}_K(\varphi \vee \varphi') \subseteq \msf{sem}_K(\varphi) \cup \msf{sem}_K(\varphi')$. The reverse inclusion is straightforward.
		\item Letting $V := \{ q \in Q \mid \forall \rho \in \msf{Paths}(q),\; \exists j \in \N,\; \rho[j] \in \msf{sem}_K(\varphi),\; \forall i \leq j-1 \in \N,\; \rho[i] \in \msf{sem}_K(\varphi') \}$, we have $\msf{sem}_K(\varphi' \lU \varphi) = V$. Indeed, consider any $q \in \msf{sem}_K(\varphi' \lU \varphi)$ and any $\rho \in \msf{Paths}(q)$. Then, by definition $\rho \models \varphi' \lU \varphi$. Consider the least $j \in \N$ such that $\rho[j:] \models \varphi$ and for all $i \leq j-1$, $\rho[i:] \models \varphi'$. Since $\varphi$ is of type $\tau_P$, we have $\rho[j] \in \msf{sem}_K(\varphi)$. Furthermore, consider any $i \leq j-1$ and $\rho' \in \msf{Paths}(\rho[i])$. We have $\rho[:i-1] \cdot \rho' \in \msf{Paths}(q)$, thus $\rho[:i-1] \cdot \rho' \models \varphi' \lU \varphi$ and, for all $k \leq i-1$, we have $\rho[k:] \not \models \varphi$, since $\varphi$ is of type $\tau_P$. Hence, we have $\rho' \models \varphi'$. As this holds for all $\rho' \in \msf{Paths}(\rho[i])$, it follows that $\rho[i] \in \msf{sem}_K(\varphi')$. This holds for all $i \leq j-1$ and for all $\rho \in \msf{Paths}(q)$. Thus, $q \in V$. In fact, we have $\msf{sem}_K(\varphi' \lU \varphi) \subseteq V$. Consider now some $q \in V$. Let $\rho \in \msf{Paths}(q)$. By definition of $V$, there is some $j \in \N$ such that $\rho[j] \in \msf{sem}_K(\varphi)$, and for all $i \leq j$, $\rho[i] \in \msf{sem}_K(\varphi)$. It follows that $\rho[j:] \models \varphi$ and, for all $i \leq j$, $\rho[i:] \models \varphi'$. Hence, $\rho \models \varphi' \lU \varphi$. Hence, $q \in \msf{sem}_K(\varphi' \lU \varphi)$. In fact, $V \subseteq \msf{sem}_K(\varphi' \lU \varphi)$. 
	\end{itemize}
	
	Therefore, by Theorem~\ref{thm:size_separation}, for all $\LTL_P(\prop)$-fragments $\msf{L}$, for all inputs $\mathcal{S} = (\mathcal{P},\mathcal{N})$ of the decision problem $\msf{PvLn}(\msf{L},\mathcal{T}(\prop))$, if there is an $\msf{L}$-separating formula, there is one of size at most:
	\begin{equation*}
		\prod_{K \in S} |\msf{SEM}_{K}(\tau)| + \prod_{K \in S} |\msf{SEM}_{K}(\tau_P)| = 2 \cdot \prod_{K \in S} 2^{|Q_K|} = 2^{\sum_{K \in S} |Q_K| +1}
	\end{equation*}
\end{proof}

\subsection{Discussion on Proposition~\ref{prop:model_checking}}
\label{proof:prop_model_checking}

In \cite{DBLP:journals/tocl/BaulandM0SSV11}, it is shown that the model checking problem of $\LTL$-formulas on Kripke structures with the operator $\lX$ is $\msf{NP}$-hard. It has to be noted that the semantics considered in \cite{DBLP:journals/tocl/BaulandM0SSV11} is different from what we consider in this paper since we have $q \models_{\msf{s}} \varphi$ if and only if all paths from $q$ satisfy $\varphi$, while in \cite{DBLP:journals/tocl/BaulandM0SSV11} it is required that there exists a path from $q$ that satisfies $\varphi$. However, we believe that, if we considered this other semantic, up to reversing the roles of $\wedge$ and $\vee$, and of $\lF$ and $\lG$ (and not using the operator $\lU$), we would obtain the same results.

In addition, what we show with Proposition~\ref{prop:model_checking} is that, for the logic $\LTL_P$, the model checking problem can be decided in polynomial time. This does not imply that the model checking problem for the logic $\msf{L}_{\lX}$ (i.e. the fragment of $\LTL$ allowing only $\lX$ as temporal operator) can be decided in polynomial time. This is due to the fact that the translation of an $\msf{L}_{\lX}$-formula into an equivalent $\LTL_P$-formula --- as we consider in Proposition~\ref{prop:LTL_X_into_LTL_P} --- may induce an exponential blow up in size. 
\begin{proof}	
	Consider an $\LTL_P$-formula $\varphi$ and an actionless non-blocking Kripke structure $K$. Let us consider the $(\LTL_P(\prop),K)$-pair $(\msf{SEM}_K,\msf{sem}_K)$ from the proof of Corollary~\ref{coro:ltl_P_logic}. To decide if the formula $\varphi$ accepts the structure $K$, it suffices to iteratively compute the set $\msf{sem}_K(\psi) \subseteq Q$ for each sub-formula $\psi \in \msf{Sub}(\varphi)$. Each one of these subsets can be computed in time polynomial in $Q$ as it amounts to compute the set described in the proof of Corollary~\ref{coro:ltl_P_logic}: the operators $\lor,\wedge,\lX$ are straightforwardly handled. As for the other operators, the semantic value can be computed via fixed-point procedure. 
	\begin{itemize}
		\item For all $\LTL_P$-formulas $\varphi$ of type $\tau$, the set $\msf{sem}_K(\lG \varphi)$ corresponds to the largest subset $S \subseteq \msf{sem}_K(\varphi)$ such that, for all $q \in S$, we have $\delta(q) \subseteq S$.
		\item For all $\LTL_P$-formulas $\varphi$ of type $\tau_P$, the set $\msf{sem}_K(\lF \varphi)$  corresponds to the least subset $\msf{sem}_K(\varphi) \subseteq S \subseteq Q$ such that, for all $q \in Q$, if $\delta(q) \subseteq S$, then $q \in S$.
		\item For all $\LTL_P$-formulas $\varphi$ of type $\tau_P$ and $\LTL_P$-formulas $\varphi'$ of type $\tau$, the set $\msf{sem}_K(\varphi' \lU \varphi)$ corresponds to the least subset $\msf{sem}_K(\varphi) \subseteq S \subseteq Q$ such that, for all $q \in \msf{sem}_K(\varphi')$, if $\delta(q) \subseteq S$, then $q \in S$.
	\end{itemize}
\end{proof}
Note that the procedure described above is very classical to decide the model checking for e.g. $\LTL$-formulas on words, $\CTL$-formulas on Kripke structures, $\ATL$-formulas on concurrent game structures, $\PCTL$-formulas on Markov chains, etc. However, it cannot be used for the full logic $\LTL$ evaluated on Kripke structure, as the full logic $\LTL$ does not have such an inductive behavior. 

\section{Proof of Corollary~\ref{coro:finite_automaton}}
\label{proof:corollary_automaton}
To be able to apply Theorem~\ref{thm:size_separation} to prove this corollary, we define below a logic whose formulas are finite words which are evaluated on finite automaton.
\begin{definition}
	\label{def:finite_words_on_atutomaton}
	Consider a non-empty alphabet $\Sigma$. We define the logic $\msf{FW}(\Sigma)$ corresponding to the finite words in $\Sigma^*$. The $\msf{FW}(\Sigma)$-syntax is as follows, with $a \in \Sigma$ a letter:
	\begin{equation*}
		\msf{u} ::= \varepsilon \mid \msf{u} \cdot a
	\end{equation*}
	The models $\mathcal{FA}(\Sigma)$ on which $\msf{FW}(\Sigma)$-formulas are evaluated is the set of finite automata whose alphabet is equal to $\Sigma$, as defined below.
	\begin{equation*}
		\mathcal{FA}(\Sigma) := \{ \mathcal{A} = (Q,\Pi,I,\delta,F) \mid \Pi = \Sigma \}
	\end{equation*}
	
	Given a finite automaton $A \in \mathcal{FA}(\Sigma)$, and a $\msf{FW}(\Sigma)$-formula $\msf{u}$, we define when $\msf{u}$ is satisfied by the automaton $A$ via the corresponding finite word in $\Sigma^*$. Specifically, we inductively define a function $f_{\msf{FW}(\Sigma)}: \msf{Fm}_{\msf{FW}(\Sigma)} \to \Sigma^*$ as follows: 
	\begin{align*}
		f_{\msf{FW}(\Sigma)}(\varepsilon) & := \varepsilon \\
		f_{\msf{FW}(\Sigma)}(\msf{u} \cdot a) & := f_{\msf{FW}(\Sigma)}(\msf{u}) \cdot a
	\end{align*}
	
	Then, a finite automaton $A \in \mathcal{FA}(\Sigma)$ satisfies an $\msf{FW}(\Sigma)$-formula $\msf{u}$, i.e. $A \models \msf{u}$, if the automaton $A$ accepts the finite word $f_{\msf{FW}(\Sigma)}(\msf{u}) \in \Sigma^*$.
\end{definition}

This definition satisfies the lemma below.
\begin{lemma}
	\label{lem:logic_words}
	Consider a non-empty alphabet $\Sigma$. The function $f_{\msf{FW}(\Sigma)}: \msf{Fm}_{\msf{FW}(\Sigma)} \to \Sigma^*$ is a bijection such that, for all $\msf{u} \in \msf{Fm}_{\msf{FW}(\Sigma)}$, we have $\msf{sz}(\msf{u}) - 1 = |f_{\msf{FW}(\Sigma)}(\msf{u})|$. 
	
	For all $u \in \Sigma^*$, we denote by $\tilde{u} \in \msf{Fm}_{\msf{FW}(\Sigma)}$ the $\msf{FW}(\Sigma)$-formula such that $f_{\msf{FW}(\Sigma)}(\tilde{u}) = u$.
\end{lemma}
\begin{proof}
	This can be proved straightforwardly by induction on $\msf{FW}(\Sigma)$-formulas.
\end{proof}

We can then apply Theorem~\ref{thm:size_separation} to this finite-word logic to obtain a corollary of which Corollary~\ref{coro:finite_automaton} is a straightforward consequence.
\begin{lemma}
	\label{lem:finite_words_logic}
	Consider a non-empty alphabet $\Sigma$. For all $\mathcal{FA}(\Sigma)$-samples $\mathcal{S} = (\mathcal{P},\mathcal{N})$, if there is a $\msf{FW}(\Sigma)$-formula that is $\mathcal{S}$-separating, then there is one such $\msf{FW}(\Sigma)$-formula of size at most $2^{n}$, with $n := \sum_{\mathcal{A} \in \mathcal{S}} |Q_\mathcal{A}|$.
\end{lemma}
Before we proceed to the proof of this lemma, let us use it to prove Corollary~\ref{coro:finite_automaton}.
\begin{proof}[Proof of Corollary~\ref{coro:finite_automaton}]
	Consider a pair $(\mathcal{P},\mathcal{N})$ of finite sets of finite automata, and assume that there exists a finite word $u \in \Sigma^*$ accepted by all automata in $\mathcal{P}$ and rejected by all automata in $\mathcal{N}$. By Lemma~\ref{lem:logic_words}, the $\msf{FW}(\Sigma)$-formula $\tilde{u} \in \msf{Fm}_{\msf{FW}(\Sigma)}$ is such that, for all $A \in \mathcal{P}$, we have $A \models \tilde{u}$, and for all $A \in \mathcal{N}$, we have $A \not\models \tilde{u}$. Thus, the sample $\mathcal{S} = (\mathcal{P},\mathcal{N})$ is a positive instance of the decision problem $\msf{PvLn}(\msf{FW}(\Sigma),\mathcal{FA}(\Sigma))$. Hence, by Lemma~\ref{lem:finite_words_logic}, there is some $\msf{FW}(\Sigma)$-formula $\tilde{u} \in \msf{Fm}_{\msf{FW}(\Sigma)}$, for some $u \in \Sigma^*$, that is $\mathcal{S}$-separating and of size at most $2^{n}$, with $n := \sum_{\mathcal{A} \in \mathcal{S}} |Q_\mathcal{A}|$. By Lemma~\ref{lem:logic_words}, the word $u \in \Sigma^*$ is accepted by all the automata in $\mathcal{P}$, rejected by all the automata in $\mathcal{N}$, and such that $|u| \leq 2^n - 1$.
\end{proof}

Let us now proceed to the proof of Lemma~\ref{lem:finite_words_logic}. 
\begin{proof}[Proof of Lemma~\ref{lem:finite_words_logic}]
	Let us apply Theorem~\ref{thm:size_separation}. Consider an automaton $A \in \mathcal{FA}(\Sigma)$. We let $\msf{SEM}_A := 2^Q$, and for all $\msf{FW}(\Sigma)$-formulas $\tilde{u} \in \msf{Fm}_{\msf{FW}(\Sigma)}$ (for $u \in \Sigma^*$):
	\begin{equation*}
		\msf{sem}_A(\tilde{u}) := \{ q \in Q \mid \text{there is an $A$-run $\rho$ on $u$ such that }\head{\rho} = q\}
	\end{equation*}
	Let us show that the $(\msf{FW}(\Sigma),A)$-pair $(\msf{SEM}_A,\msf{sem}_A)$ inductively captures the $\msf{FW}(\Sigma)$-semantics. It clearly captures the $\msf{FW}(\Sigma)$-semantics since we have $A \models \tilde{u}$ if and only if there is an $A$-run $\rho$ on $u$ such that $\head{\rho} \in F$. As for the inductive property, we have that, for all $\msf{FW}(\Sigma)$-formulas $\tilde{u} \in \msf{Fm}_{\msf{FW}(\Sigma)}$ and letters $a \in \Sigma$, we have:
	\begin{align*}
		\msf{sem}_A(\tilde{u \cdot a}) & = \{ q \in Q \mid \text{there is an $A$-run $\rho$ on $u \cdot a$ such that }\head{\rho} = q\} \\
		& = \{ q \in Q \mid \exists \rho \in Q^{|u\cdot a|+1},\;  \forall i \in \llbracket 0,|\rho|-2 \rrbracket,\; \rho[i+1] \in \delta(\rho[i],(u\cdot a)[i]), \\
		& \phantom{= \{ \}q \in Q \mid \exists \rho \in Q^{|u\cdot a|+1},\; } \rho[0] \in I,\; \head{\rho} = q \} \\
		& = \{ q \in Q \mid \exists \rho \in Q^{|u|+1},\;  \forall i \in \llbracket 0,|\rho|-2 \rrbracket,\; \rho[i+1] \in \delta(\rho[i],u[i]), \\
		& \phantom{= \{ \}q \in Q \mid \exists \rho \in Q^{|u\cdot a|+1},\; } \rho[0] \in I,\; q \in \delta(\head{\rho},a)\} \\
		& = \{ q \in Q \mid \exists q' \in \msf{sem}_A(\tilde{u}),\; q \in \delta(q',a) \} =  \bigcup_{q \in \msf{sem}_A(\tilde{u})} \delta(q,a)
	\end{align*}
	%
	%\begin{equation*}
	%	\msf{sem}_A(\tilde{u} \cdot a) = \{ q \in Q \mid \exists q' \in \msf{sem}_A(\tilde{u}),\; q \in \delta(q',a) \}
	%\end{equation*}
	Hence, by Remark~\ref{rmk:sufficient_condition}, the pair the $(\msf{FW}(\Sigma),A)$-pair $(\msf{SEM}_A,\msf{sem}_A)$ inductively captures the $\msf{FW}(\Sigma)$-semantics. Therefore, by Theorem~\ref{thm:size_separation}, $\mathcal{S} = (\mathcal{P},\mathcal{N})$ of the decision problem $\msf{PvLn}(\msf{FW}(\Sigma),\mathcal{FA}(\Sigma))$, if there is an $\msf{FW}(\Sigma)$-formula that is $\mathcal{S}$-separating, there is one of size at most:
	\begin{equation*}
		\prod_{A \in \mathcal{S}} |\msf{SEM}_{A}| = \prod_{A \in \mathcal{S}} 2^{|Q_A|} = 2^{\sum_{A \in \mathcal{S}} |Q_A|}
	\end{equation*}
\end{proof}

\section{Proof of Corollary~\ref{coro:parity_automaton}}
\label{proof:corollary_parity_automaton}
The case of parity automata and ultimately periodic words is similar to the case of finite automata and finite words, with several additional difficulties. 

We first tackle the case of periodic words on parity automata, and we then deduce Corollary~\ref{coro:parity_automaton} (which deals with ultimately periodic words) with the help of Corollary~\ref{coro:finite_automaton}. Let us first define a periodic words logic. 	
\begin{definition}
	\label{def:periodic_words_on_parity_atutomaton}
	Consider a non-empty alphabet $\Sigma$. We define the logic $\msf{PW}(\Sigma)$ of periodic words $w = v^\omega \in \Sigma^\omega$. The $\msf{PW}(\Sigma)$-syntax is given by the grammar below, with $a \in \Sigma$ a letter: 		
	\begin{align*}
		\msf{v} & ::= a \mid \msf{v} \cdot a
	\end{align*}
	%More formally, the syntax is defined by $(\{\tau_{\msf{u}},\tau_{\msf{v}}\},\{\tau_{\msf{v}}\},\msf{Op},\tau,T)$ with $\msf{Op}_0 := \{ \varepsilon \mid a \in \Sigma \}$, $\msf{Op}_1 := \{ \msf{Inf},\msf{App}^{\msf{w}}_a,\msf{App}^{\msf{v}}_a\mid a \in \Sigma \}$, and for all $a \in \Sigma$: $\tau_{a} = \tau_{\msf{App}^{\msf{v}}_a} := \tau_{\msf{v}}$, $\tau_{\msf{App}^{\msf{w}}_a} = \tau_{\msf{Inf}} := \tau_{\msf{w}}$, and $T(\msf{Inf},1) = T(\msf{App}^{\msf{v}}_a,1) = \{\tau_{\msf{v}}\}$, and $T(\msf{App}^{\msf{w}}_a,1) = \{\tau_{\msf{w}}\}$. Thus, the ultimately periodic word $a_1 \cdot a_2 \cdot a_3 \cdot (a_4 \cdot a_5)^\omega$ is represented by the formula $\msf{App}_{a_1}^{\msf{w}}( \msf{App}_{a_2}^{\msf{w}}( \msf{App}_{a_3}^{\msf{w}}( \msf{Inf}(\msf{App}_{a_4}^{\msf{v}}(a_5)))))$
	
	The models $\mathcal{PA}(\Sigma)$ on which $\msf{PW}(\Sigma)$-formulas are evaluated is the set of parity automata defined below.
	\begin{equation*}
		\mathcal{PA}(\Sigma) := \{ \mathcal{A} = (Q,\Pi,I,\delta,\pi) \mid \Pi = \Sigma \}
	\end{equation*}
	
	Then, given a parity automaton $A \in \mathcal{PA}(\Sigma)$ and a $\msf{PW}(\Sigma)$-formula $\msf{v}$, we define when $\msf{v}$ is satisfied by $A$. To do so, we are going to consider the periodic word in $\Sigma^{\omega}$ corresponding to the formula $\msf{v}$. First, we inductively define a function $f_{\msf{PW}(\Sigma)}: \msf{Fm}_{\msf{PW}(\Sigma)} \rightarrow \Sigma^+$ as follows: 
	\begin{align*}
		f_{\msf{PW}(\Sigma)}(a) & := a \\
		f_{\msf{PW}(\Sigma)}(\msf{v} \cdot a) & := 	f_{\msf{PW}(\Sigma)}(\msf{v}) \cdot a
	\end{align*}
	Then, we let $\Sigma^{+\omega} := \{ v^\omega \mid v \in \Sigma^+ \}$ denote the set of  periodic words, and we let $g_{\msf{PW}(\Sigma)}: \msf{Fm}_{\msf{PW}(\Sigma)} \to \Sigma^{+\omega}$ be such that, for all $\msf{v} \in \msf{Fm}_{\msf{PW}(\Sigma)}$, we have $g_{\msf{PW}(\Sigma)}(\msf{v}) := (f_{\msf{PW}(\Sigma)}(\msf{v}))^\omega$. 
	
	Then, a parity automaton $A \in \mathcal{PA}(\Sigma)$ satisfies an $\msf{PW}(\Sigma)$-formula $\msf{v}$, i.e. $A \models \msf{v}$, if the periodic word $g_{\msf{PW}(\Sigma)}(\msf{v})$ is accepted by the parity automaton $A$. 
\end{definition}

This definition satisfies the lemma below.
\begin{lemma}
	\label{lem:logic_periodic_words}
	Consider a non-empty alphabet $\Sigma$. The function $f_{\msf{PW}(\Sigma)}: \msf{Fm}_{\msf{PW}(\Sigma)} \to \Sigma^+$ is a bijection and, for all $\msf{v} \in \msf{Fm}_{\msf{PW}(\Sigma)}$, we have $\msf{sz}(\msf{v}) =  |f_{\msf{PW}(\Sigma)}(\msf{v})|$. 
	
	For all $v \in \Sigma^+$, we denote by $\overline{v} \in \msf{Fm}_{\msf{PW}(\Sigma)}$ the $\msf{PW}(\Sigma)$-formula for which $f_{\msf{PW}(\Sigma)}(\overline{v}) = v$.
\end{lemma}
\begin{proof}
	This can be proved straightforwardly by induction on $\msf{Fm}_{\msf{PW}(\Sigma)}$.
\end{proof}

We can then apply Theorem~\ref{thm:size_separation} to this periodic-word logic to obtain a lemma from which we will be able to deduce Corollary~\ref{coro:parity_automaton}.
\begin{lemma}
	\label{lem:periodic_words_logic}
	Consider a non-empty alphabet $\Sigma$. For all $\mathcal{PA}(\Sigma)$-samples $\mathcal{S} = (\mathcal{P},\mathcal{N})$, if there is a $\mathcal{S}$-separating $\msf{PW}(\Sigma)$-formula, there is one of size at most $2^{k}$, with $k := \sum_{A \in \mathcal{S}} |Q_A|^2 \cdot n_A$.
\end{lemma}
\begin{proof}
	Let us apply Theorem~\ref{thm:size_separation}. Consider a parity automaton $A \in \mathcal{PA}(\Sigma)$. We let $\msf{SEM}_A := \{ Q \to 2^{Q \times \pi(Q)} \}$, and for all $\msf{PW}(\Sigma)$-formulas $\overline{v} \in \msf{Fm}_{\msf{FW}(\Sigma)}$ (for $v \in \Sigma^+$), we let $\msf{sem}_A(\overline{v}): Q \to 2^{Q \times \pi(Q)}$ be such that, for all $q \in Q$:
	\begin{align*}
		\msf{sem}_A(\overline{v})(q) := \{ (q',n) \in Q \times \Interval{0,n_A} \mid & \exists \rho \in Q^{|v|+1}:\; \forall i \in \Interval{0,|\rho|-2},\; \rho[i+1] \in \delta(\rho[i],v[i]), \\
		& \rho[0] = q,\; \head{\rho} = q', \; \max_{0 \leq i \leq |\rho|-1} \pi(\rho[i]) = n \}
	\end{align*}
	
	Let us show that the $(\msf{PW}(\Sigma),A)$-pair $(\msf{SEM}_A,\msf{sem}_A)$ inductively captures the $\msf{PW}(\Sigma)$-semantics. Let us first show that it satisfies the inductive property. For all $\msf{PW}(\Sigma)$-formulas $\overline{v} \in \msf{Fm}_{\msf{PW}(\Sigma)}$ and letters $a \in \Sigma$, we have $\msf{sem}_A(\overline{v} \cdot a): Q \to 2^{Q \times \pi(Q)}$ such that, for all $q \in Q$
	\begin{align*}
		\msf{sem}_A(\overline{v} \cdot a)(q) = & \{ (q',n') \in Q \times \Interval{0,n_A} \mid \exists (q'',n'') \in \msf{sem}_A(\overline{v}):
		q' \in \delta(q'',a),\; n' = \max(\pi(q'),n'') \}
	\end{align*}
	Hence, by Remark~\ref{rmk:sufficient_condition}, the pair the $(\msf{PW}(\Sigma),A)$-pair $(\msf{SEM}_A,\msf{sem}_A)$ satisfies the inductive property. Let us now focus on capturing the $\msf{FW}(\Sigma)$-semantics. Given any $v \in \Sigma^+$, we let $\msf{Priorities}_A(v) := \{ n \in \pi(Q) \mid \text{ there is an }A\text{-run }\rho\text{ on }v^\omega\text{ s.t. } n := \max \{ k \mid \forall i \in \N,\; \exists j \geq i,\; \pi(\rho[j]) = k \}  \}$. Note that $v^\omega$ is accepted by $A$ if and only if there is some even integer in $\msf{Priorities}_A(v)$. Now, given any $h: Q \to 2^{Q \times \pi(Q)}$ and $q \in Q$, we let:
	\begin{align*}
		\msf{Sequences}_A(h,q) := \{ & (q_0,n_0) \cdot (q_1,n_1) \cdots \in (Q \times \pi(Q))^\omega \mid \\ & (q_0,n_0) \in h(q),\; \forall i\in \N,\; (q_{i+1},n_{i+1}) \in h(q_i) \}
	\end{align*}
	and 
	\begin{align*}
		\msf{Priorities}_A(h) := \{ n \in \pi(Q) \mid & \; \exists q \in I,\; \exists (q_0,n_0) \cdot (q_1,n_1) \cdots \in \msf{Sequences}_A(h,q): \\
		%& \phantom{m}\rho = \body{\rho_0} \cdot \body{\rho_1} \cdot \body{\rho_2} \cdots,\; \\
		& \; n = \max \{ k \mid \forall i \in \N,\; \exists j \geq i,\; n_j = k \} \}
	\end{align*}
	Then, for all $v \in \Sigma^+$, let us show that:
	\begin{equation*}
		\msf{Priorities}_A(v) = \msf{Priorities}_A(\msf{sem}_A(\overline{v}))
	\end{equation*} 
	
	Let $t := |v| \geq 1$. Consider %some $n \in \msf{Priorities}_A(v)$ and 
	an $A$-run $\rho$ on $v^\omega$.  %We let $n_0 := \max \{ \pi(\rho[j]) \mid 0 \leq j \leq t-1 \} \in \pi(Q)$ and f
	For all $m \in \N$, we let $n_m := \max \{ \pi(\rho[j]) \mid m \cdot t \leq j \leq (m+1) \cdot t \} \in \pi(Q)$. By definition of an $A$-run and of $\msf{sem}_A(\overline{v})$, we have: 
	\begin{equation*}
		\forall m \in \N,\; (\rho[(m+1) \cdot t],n_m) \in \msf{sem}_A(\overline{v})(\rho[m \cdot t])
	\end{equation*}
	Thus, $((\rho[(m+1) \cdot t],n_m))_{m \in \N} \in \msf{Sequences}_A(\msf{sem}_A(\overline{v}),\rho[0])$, with $\rho[0] \in I$. Furthermore, we have:
	\begin{equation*}
		\max \{ k \mid \forall i \in \N,\; \exists j \geq i,\; \pi(\rho[j]) = k \} = \max \{ k \mid \forall i \in \N,\; \exists j \geq i,\; n_j = k \}
	\end{equation*}
	Since this holds for all $A$-runs $\rho$ on $v^\omega$, it follows that $\msf{Priorities}_A(v) \subseteq \msf{Priorities}_A(\msf{sem}_A(\overline{v}))$.
	
	Reciprocally, consider some $((q_m,n_m))_{m \in \N} \in \msf{Sequences}_A(\msf{sem}_A(\overline{v}),q)$, for some $q \in I$. By definition, for all $m \in \N$, there is some $\rho^m \in Q^{t+1}$ such that: 
	\begin{equation*}
		\rho^m[0] = q_m,\; \head{\rho^m} = q_{m+1},\; \forall 0 \leq i \leq t-1,\; \rho^m[i+1] \in \delta(\rho^m[i],v[i]),\; \max_{0 \leq i \leq t} \pi(\rho[i]) = n_m  
	\end{equation*}
	Let $\rho := \body{\rho^0} \cdot \body{\rho^1} \cdot \body{\rho^2} \cdots \in Q^\omega$. By definition, $\rho$ is an $A$-run on $v^\omega$ such that:
	\begin{equation*}
		\max \{ k \mid \forall i \in \N,\; \exists j \geq i,\; \pi(\rho[j]) = k \} = \max \{ k \mid \forall i \in \N,\; \exists j \geq i,\; n_j = k \}
	\end{equation*}
	Therefore, we have: $\msf{Priorities}_A(\msf{sem}_A(\overline{v})) \subseteq \msf{Priorities}_A(v)$. We obtain the desired equality.
	
	Therefore, for all $\overline{v},\overline{v}' \in \msf{Fm}_{\msf{PW}(\Sigma)}$ such that $\msf{sem}_A(\overline{v}) = \msf{sem}_A(\overline{v}')$, we have $\msf{Priorities}_A(v) = \msf{Priorities}_A(v')$, and thus $A\models \overline{v}$ if and only if $A\models \overline{v}'$. Therefore, the $(\msf{PW}(\Sigma),A)$-pair $(\msf{SEM}_A,\msf{sem}_A)$ captures the $\msf{PW}(\Sigma)$-semantics. In fact, the $(\msf{PW}(\Sigma),A)$-pair $(\msf{SEM}_A,\msf{sem}_A)$ inductively captures the $\msf{PW}(\Sigma)$-semantics. Therefore, by Theorem~\ref{thm:size_separation}, for all $\mathcal{PA}(\Sigma)$-samples $\mathcal{S} = (\mathcal{P},\mathcal{N})$, if there is a $\msf{PW}(\Sigma)$-formula that is $\mathcal{S}$-separating, there is one of size at most:
	\begin{equation*}
		\prod_{A \in \mathcal{S}} |\msf{SEM}_{A}| = \prod_{A \in \mathcal{S}} |\{ Q_A \to 2^{Q_A \times \pi(Q)}| = \prod_{A \in \mathcal{S}} (2^{|Q_A| \times n_A})^{|Q_A|} = 2^{\sum_{A \in \mathcal{S}} |Q_A|^2 \cdot n_A}
	\end{equation*}
\end{proof}

Let us now use this lemma to establish Corollary~\ref{coro:parity_automaton}.
Let us now use this lemma to establish Corollary~\ref{coro:parity_automaton}.
\begin{proof}[Proof of Corollary~\ref{coro:parity_automaton}]
	Consider a pair $(\mathcal{P},\mathcal{N})$ of finite sets of parity automata, and assume that there exists an ultimately periodic word $u \cdot v^\omega \in \Sigma^{\omega}$ accepted by all automata in $\mathcal{P}$ and rejected by all automata in $\mathcal{N}$. For all $A \in \mathcal{P} \cup \mathcal{N}$, we denote $A$ by $(Q_A,\Sigma,I_A,\delta_A,\pi_A)$, and:
	\begin{itemize}
		\item We define the subset of states $I_{A'} \subseteq Q_A$ as follows:
		\begin{itemize}
			\item If $A \in \mathcal{P}$, we consider some state $q_A \in Q_A$ such that: there exists a finite $A$-run $\rho$ on $u$ such that $\head{\rho} = q_A$, and there exists an infinite $A$-run $\rho'$ on $v^\omega$ starting in $q_A$ whose maximum integer seen infinitely often is even. We let $I_{A'} := \{q_A\}$.
			\item If $A \in \mathcal{N}$, we let:
			\begin{equation*}
				I_{A'} := \{ q \in Q \mid v^\omega\text{ is not accepted by the parity automaton }(Q_A,\Sigma,\{q\},\delta_A,\pi_A) \}
			\end{equation*}
		\end{itemize}
		\item We define the parity automaton:
		\begin{align*}
			A' := (Q_A,\Sigma,I_A',\delta_A,\pi_A) \in \mathcal{PA}(\Sigma)
		\end{align*}
		\item We define the finite automaton:
		\begin{align*}
			A_u :=
			\begin{cases}
				(Q_A,\Sigma,I_A,\delta_A,I_A') \in \mathcal{FA}(\Sigma) & \text{ if } A \in \mathcal{P} \\
				(Q_A,\Sigma,I_A,\delta_A,Q \setminus I_A') \in \mathcal{FA}(\Sigma) & \text{ if } A \in \mathcal{N} \\
			\end{cases}
		\end{align*}
	\end{itemize}
	Now, we let $\mathcal{P}' := \{ A' \mid A \in \mathcal{P} \}$ and $\mathcal{N}' := \{ A' \mid A \in \mathcal{N} \}$. Then, we have that the periodic word $v^\omega \in \Sigma^{\omega}$ is accepted by all automata in $\mathcal{P}'$ and rejected by all automata in $\mathcal{N}'$.	By Lemma~\ref{lem:logic_periodic_words}, the $\msf{PW}(\Sigma)$-formula $\overline{v} \in \msf{Fm}_{\msf{PW}(\Sigma)}$ is such that, for all $A' \in \mathcal{P}'$, we have $A' \models \overline{v}$, and for all $A' \in \mathcal{N}'$, we have $A' \not\models \overline{v}$. %Thus, the sample $\mathcal{S}' = (\mathcal{P}',\mathcal{N}')$ is a positive instance of the decision problem $\msf{PvLn}(\msf{PW}(\Sigma),\mathcal{PA}(\Sigma))$. 
	Hence, by Lemma~\ref{lem:periodic_words_logic}, there is some $\msf{PW}(\Sigma)$-formula $\overline{v'}$, for some $v' \in \Sigma^+$, that is $\mathcal{S}'$-separating and of size at most $2^{k}$, with $k := \sum_{A' \in \mathcal{S}} |Q_{A'}|^2 \cdot n_{A'} = \sum_{A \in \mathcal{S}} |Q_{A}|^2 \cdot n_{A}$. By Lemma~\ref{lem:logic_periodic_words}, the word $(v')^\omega \in \Sigma^\omega$ is accepted by all the automata in $\mathcal{P}'$, rejected by all the automata in $\mathcal{N}'$. 
	
	Furthermore, we let $\mathcal{P}_u := \{ A_u \mid A \in \mathcal{P} \}$ and $\mathcal{N}_u := \{ A_u \mid A \in \mathcal{N} \}$. Consider some automaton $A \in \mathcal{P} \cup \mathcal{N}$. Since the ultimately periodic word $u \cdot v^\omega$ is accepted by $A$ if and only if $A \in \mathcal{P}$, it follows that the finite word $u$ is accepted by $A_u$ if and only if $A \in \mathcal{P}$. Thus, by Corollary~\ref{coro:finite_automaton}, we have that there is a word $u' \in \Sigma^*$ that is accepted by all the finite automata in $\mathcal{P}_u$ and rejected by all the finite automata in $\mathcal{N}_u$, and of size at most $2^{n}-1$, for $n := \sum_{A \in \mathcal{S}} |Q_{A}|$. 
	
	Consider now the ultimately periodic word $w' := u' \cdot (v')^\omega$. We have $|w'| \leq 2^{n} + 2^{k} - 1$. Furthermore, for all $A \in \mathcal{P} \cup \mathcal{N}$:
	\begin{itemize}
		\item If $A \in \mathcal{P}$, then there exists a finite $A$-run $\rho$ on $u'$ such that $\head{\rho} = q_A$, since $u'$ is accepted by $A_u$. Furthermore, there is an infinite $A$-run on $(v')^\omega$, starting at $q_A$, whose maximum integer seen infinitely often is even, since $(v')^\omega$ is accepted by $A'$. Thus, the word $w'$ is accepted by $A$.
		\item If $A \in \mathcal{N}$, then all finite $A$-runs $\rho$ on $u'$ are such that $\head{\rho} \in I_A'$, since $u'$ is rejected by $A_u$. Furthermore, all infinite $A$-runs on $(v')^\omega$, starting from any state in $I_A'$, are with a maximum integer seen infinitely often odd, since $(v')^\omega$ is rejected by $A'$. Thus, the word $w'$ is rejected by $A$.
	\end{itemize}
	Overall, the word $w'$ is accepted by all the parity automata in $\mathcal{P}$, rejected by all the parity automata in $\mathcal{N}$, and such that $|w'| \leq 2^{n} + 2^{k} - 1$.
\end{proof}		

\section{Proof of Proposition~\ref{prop:lower_bound_ltl}}
\label{proof:lower_bound_ltl}
Let us first define below the samples, parameterized by an integer $n \in \N$, that we will consider to establish Proposition~\ref{prop:lower_bound_ltl}. %On the way, we also introduce again the notations used in the proof sketch given in the main part of the paper.
\begin{definition}
	\label{def:samples_ltl_lower_bound}
	For all $n \geq 1$, we let $p_n \in \N$ denote the $n$-th prime number, thus $p_1 = 2$, $p_2 = 3$, etc. Note that we have $\sum_{l = 1}^j p_{l} \sim (j^2)/(2 \cdot \log j)$, while $\prod_{l = 1}^j p_{l} = e^{(1 + o(1)) \cdot j \cdot \log j}$ \cite[VII.27,VII.35]{sandor2004handbook}.
	
	Furthermore, consider some $x \neq y \in \prop$. For all $n \in \N$, we let $w_n(x) := (\{y\} \cdot \{y\} \cdots \{y\} \cdot \{x\})^\omega \in \mathcal{IW}(\prop)$ be such that $|w_n| = n$. 
	
	Consider now some $n,j \in \N$. We let: 
	\begin{itemize}
		\item $\mathcal{P}_n^j(x) := \{ w_{p_i}(x)[j:] \mid i \in \Interval{1,n}\} \cup \{ w_4'(x)[j:] \}$ with $w_4'(x) := \{y\} \cdot \{x\} \cdot w_4(x)$;
		\item $\mathcal{N}_n^j(x) := \{w_4(x)[j:]\}$
	\end{itemize}
	We also let $\mathcal{S}^{\wedge,j}_n(x) := (\mathcal{P}_n^j(x),\mathcal{N}_n^j(x))$ and $\mathcal{S}^{\lor,j}_n(s) := (\mathcal{N}_n^j(x),\mathcal{P}_n^j(x))$.
	%
	%We also let $\msf{Sync}_n(a) := \min \{ j \in \N \mid \forall w \in \mathcal{P}_n^j(a),\; a \in w[0] \}$.
\end{definition}

Let us first establish a few simple properties that the above definition satisfies.
\begin{lemma}
	\label{lem:lower_bound_ltl_sample_works_first}
	Consider some $x \neq y \in \prop$. Let $n \in \N$. 
	\begin{itemize}
		\item The least integer $k_n \in \N$ such that for all $w \in \mathcal{P}_n^{k_n}(x)$, we have $x \in w[0]$ and for all $w \in \mathcal{N}_n^{k_n}(x)$, we have $x \notin w[0]$ is equal to $k_n := \prod_{j = 1}^n p_j - 1$.  
		\item For all $j \in \N$, if $y \notin w_4(x)[j]$, then $y \notin w_2(x)[j]$.
	\end{itemize}
\end{lemma}
\begin{proof}
	By definition, for all $j \in \N$, we have $w_n(x)[j] = \{x\}$ if $n \mid j+1$ (i.e $n$ is a divisor of $j+1$) and $w_n(x)[j] = \{y\}$ otherwise. Hence, we have:
	\begin{equation*}
		\forall l \in \Interval{1,n},\; x \in w_{p_l}[j] \Longleftrightarrow \forall l \in \Interval{1,n},\; p_l \mid j+1 \Longleftrightarrow \prod_{l = 1}^n \; p_l \mid j+1
	\end{equation*} 
	Hence, for all $j \in \N$ such that for all $w \in \mathcal{P}_n^{k_n}(x)$, we have $x \in w[j]$, we have, for all $l \in \Interval{1,n}$, $x \in (w_{p_l}[j:])[0] = w_{p_l}[j]$, and thus $j+1 \geq \prod_{l = 1}^n \; p_l$, i.e. $j \geq k_n$. Furthermore, for all $j \in \N$, we have $x \in (\{y\} \cdot \{x\} \cdot w_4)^\omega[j]$ if and only if $j+1 \mod 4 = 2$. In addition, since $k_n+1$ is even and not divisible by 4, it follows that $x \in  w_4'[k_n]$; and similarly $x \notin w_4[k_n]$. 
	
	On the other hand, for all $j \in \N$, if $y \notin w_4(x)[j]$, then $x \in w_4(x)[j]$, and $4 \mid j+1$, thus $2 \mid j+1$ and $x \in w_2(x)[j]$, hence $y \notin w_2(x)[j]$.
\end{proof}

The above definition satisfies another property, a little harder to establish, that is crucial to the proof of Proposition~\ref{prop:lower_bound_ltl}.
\begin{lemma}
	\label{lem:lower_bound_ltl_sample_works_second}
	Consider a monotone fragment $\msf{L}$ of $\LTL(\prop)$ with $\vee \notin \msf{Op}$%that disallows the operator $\lor$
	. Consider also some $x \neq y \in \prop$. Let $n \in \N$ and $j < k_n \in \N$. Assume that there exists an $\msf{L}$-formula $\varphi$ that is $\mathcal{S}^{\wedge,j}_n(x)$-separating. In that case, there is a sub-formula $\psi \in \msf{Sub}(\varphi)$ such that:
	\begin{itemize}
		\item $\msf{sz}(\varphi) \geq \msf{sz}(\psi) + 1$;
		\item the $\msf{L}$-formula $\psi$ is $\mathcal{S}^{\wedge,j+1}_n(x)$-separating.
	\end{itemize}
\end{lemma}
\begin{proof}
	Let us prove the result by exhaustion:
	\begin{itemize}
		%\item Assume that $\varphi \in \{\top,\bot\}$. Then, it cannot be $\mathcal{S}^{\wedge,j}_n(x)$-separating.
		\item Assume that $\varphi \in \{p,\bar{p}\}$. Then, since $j < k_n$ and by Lemma~\ref{lem:lower_bound_ltl_sample_works_first}, we have that $\varphi$ cannot possibly be $\mathcal{S}^{\wedge,j}_n(x)$-separating (the case $\varphi = x$ follows from the first item, the case $\varphi = y$ follows from the second one).
		\item Assume that $\varphi = \varphi_1 \wedge \varphi_2$. Then, let $\psi \in \{\varphi_1,\varphi_2\}$ be an $\msf{L}$-formula such that $w_4(j)[:] \not\models \psi$. Then, we have $\msf{sz}(\psi) < \msf{sz}(\varphi)$ and $\psi$ is $\mathcal{S}_n^{\wedge,j}(x)$-separating. Thus, we can apply (by induction) the  argument to the formula $\psi$ to obtain an $\msf{L}$-formula $\psi' \in \msf{Sub}(\psi) \subseteq \msf{Sub}(\varphi)$ satisfying the assumptions of the lemma.
		\item Assume that $\varphi = \lX \varphi'$. Then, by definition, we have $\msf{sz}(\varphi') < \msf{sz}(\varphi)$ and $\varphi'$ is $\mathcal{S}_n^{\wedge,j+1}(x)$-separating.
		\item Assume that $\varphi = \lF \varphi'$ or $\varphi = \lG \varphi'$. Since $\{ (w_4(x)[j:])[k:] \mid k \in \N \} = \{ (w_4'(x)[j:])[k:] \mid k \in \N \}$, it follows that either $\varphi$ accepts reject both $w_4'(x)[j:]$ and $w_4(x)[j:]$ or rejects both $w_4'(x)[j:]$ and $w_4(x)[j:]$, thus it is not $\mathcal{S}_n^{\wedge,j}(x)$ separating.
	\end{itemize}
	Overall, whenever $\varphi$ is $\mathcal{S}_n^{\wedge,j}(x)$-separating, it is possible to extract a sub-formula $\psi \in \msf{Sub}(\varphi)$ satisfying the conditions of the lemma.
\end{proof}

Finally, let us prove a final lemma before we proceed to the proof of Proposition~\ref{prop:lower_bound_ltl}. This lemma states that the functions defined below behave like negations. 
\begin{definition}
	For all $* \in \{\wedge,\lor\}$, we denote by $\LTL_{*}(\prop)$ the $\LTL(\prop)$-monotone fragment that allows exactly the $\LTL$-operators $\{p,\bar{p},*,\lX,\lF,\lG\}$. 
	
	For all $* \neq *' \in \{\wedge,\lor\}$, we define by induction the function $f^{\neg}_{* \to *'}: \msf{Fm}_{\LTL_{*}(\prop)} \rightarrow \msf{Fm}_{\LTL_{*'}(\prop)}$ such that:
	\begin{align*}
		f^{\neg}_{* \to *'}(p) & := \bar{p} \\
		f^{\neg}_{* \to *'}(\bar{p}) & := p \\
		f^{\neg}_{* \to *'}(\varphi_1 * \varphi_2) & := f^{\neg}_{* \to *'}(\varphi_1) *' f^{\neg}_{* \to *'}(\varphi_2) \\
		f^{\neg}_{* \to *'}(\lX \varphi) & := \lX f^{\neg}_{* \to *'}(\varphi) \\
		f^{\neg}_{* \to *'}(\lF \varphi) & := \lG f^{\neg}_{* \to *'}(\varphi) \\
		f^{\neg}_{* \to *'}(\lG \varphi) & := \lF f^{\neg}_{* \to *'}(\varphi) 
	\end{align*}
\end{definition}
\begin{lemma}
	\label{lem:ltl_function_neg}
	Consider some $x \neq y \in \prop$ and any word $w \in \cup_{n,j \in \N} \; \mathcal{S}_n^j(x)$. For all $* \neq *' \in \{\wedge,\lor\}$ and $\LTL_{*}(\prop)$-formulas $\varphi$, we have $\msf{sz}(\varphi) = \msf{sz}(f^{\neg}_{* \to *'}(\varphi))$ and:
	\begin{equation*}
		w \models \varphi \Longleftrightarrow w \not\models f^{\neg}_{* \to *'}(\varphi)
	\end{equation*} 
\end{lemma}
\begin{proof}
	The size equality can be proved straightforwardly by induction. Then, the word $w \in (2^\prop)^\omega$ is such that, for all $i \in \N$, we have $p \in w[i]$ if and only if $\bar{p} \notin w[i]$. The result of the lemma then follows from the classical $\LTL$-equivalences recalled below. For all $\LTL(\prop)$-formulas $\varphi,\varphi_1,\varphi_2$ and word $w \in (2^\prop)^\omega$, we have: 
	\begin{align*}
		w \not\models \lX \varphi & \Longleftrightarrow w \models \lX \neg \varphi \\ 
		%w \not\models \lF \varphi & \Longleftrightarrow w \models \lG \neg \varphi \\ 
		w \not\models \lG \varphi & \Longleftrightarrow w \models \lF \neg \varphi \\ 
		w \not\models \varphi_1 \wedge \varphi_2 & \Longleftrightarrow w \models \neg \varphi_1 \vee \neg \varphi_2 
		%w \not\models \varphi_1 \vee \varphi_2 & \Longleftrightarrow w \models \neg \varphi_1 \wedge \neg \varphi_2 \\ 
	\end{align*}
\end{proof}

We can now proceed to the proof of Proposition~\ref{prop:lower_bound_ltl}.
\begin{proof}
	Consider some $x \neq y \in \prop$. Consider some $i \in \N$. We let $t_i := \sum_{w \in \mathcal{S}_i^{\wedge,0}(x)} |w| = \sum_{l=1}^i p_i + 10$. Since $\sum_{l=1}^i p_i \sim (i^2)/(2\cdot \log i)$, it follows that there is some $n_0 \in \N$ such that, for all $i \geq n_0$, we have $t_i \leq i^2$. Furthermore, since we have $k_i + 1 = \prod_{l=1}^i p_i = e^{(1 + o(1)) \cdot i \cdot \log i}$, if follows that there is some $n_0' \in \N$ such that, for all $i \geq n_o'$, we have $k_i + 1 \geq 2^{i}$. Let $n_1 := \max(n_0,n_0')$. 
	
	Consider now any monotone $\LTL(\prop)$-fragment $\msf{L}$ such that $\vee \notin \msf{Op}$. Let $x \neq y \in \prop$ be such that $x \in \msf{Op}$. Consider any  $n \in \N$, and let $i := \max(n,n_1)$ and $\mathcal{S} := \mathcal{S}_{i}^{\wedge,0}(x)$. Let us show that the minimal size of an $\msf{L}$-formula that is $\mathcal{S}$-separating is $k_i+1$. Since we have $i \leq t_i \leq i^2$ and $k_i + 1 \geq 2^{i} \geq 2^{\sqrt{t_i}}$, this will show the result for monotone fragments without $\vee$. 
	
	First of all, consider the $\msf{L}$-formula $\varphi := (\lX)^{k_i} x$ of size $\msf{sz}(\varphi) = k_i + 1$. By Lemma~\ref{lem:lower_bound_ltl_sample_works_first}, it is $\mathcal{S}$-separating.
	
	Consider now any $\mathcal{S}$-separating $\msf{L}$-formula $\varphi$. By iteratively applying Lemma~\ref{lem:lower_bound_ltl_sample_works_second} to $j \in \Interval{0,k_i-1}$, we obtain that there is some $\msf{L}$-formula $\psi \in \msf{Sub}(\varphi)$ such that $\msf{sz}(\varphi) \geq \msf{sz}(\psi) + k_i$ (and $\psi$ is $\mathcal{S}_i^{\wedge,k_i}$-separating). Since $\msf{sz}(\psi) \geq 1$, it follows that $\msf{sz}(\varphi) \geq k_i+1$. 
	
	Let us now consider a monotone $\LTL(\prop)$-fragment $\msf{L}$ such that $\vee \in \msf{Op}$, and thus $\wedge \notin \msf{Op}$. Let $x \neq y \in \prop$ be such that $x \in \msf{Op}$. Consider any $n \in \N$, and let $i := \max(n,n_1)$ and $\mathcal{S} := \mathcal{S}_{i}^{\vee,0}(y)$. First, the $\msf{L}$-formula $\varphi := (\lX)^{k_i} x$ of size $\msf{sz}(\varphi) = k_i + 1$ is $\mathcal{S}$-separating by Lemma~\ref{lem:lower_bound_ltl_sample_works_first}. Furthermore, for all $\msf{L}$-formulas $\varphi$, by Lemma~\ref{lem:ltl_function_neg}, we have $\varphi \in \msf{Fm}_{\LTL_{\vee}(\prop)}$, $f^{\neg}_{\vee \to \wedge}(\varphi) \in \msf{Fm}_{\LTL_{\wedge}(\prop)}$, and $\msf{sz}(\varphi) = \msf{sz}(f^{\neg}_{\vee \to \wedge}(\varphi))$. Consider any $\msf{L}$-formula $\varphi$ that is $\mathcal{S}$-separating. We have that the $\LTL_{\wedge}(\prop)$-formula $f^{\neg}_{\vee \to \wedge}(\varphi)$ is $\mathcal{S}_{i}^{\wedge,0}(y)$, and thus of size at least $k_i + 1$ (as showed above). Hence, $\varphi$ is also of size at least $k_i + 1$. The proposition follows.
\end{proof}

\section{Proof of Proposition~\ref{prop:lower_bound_ml}}
\label{proof:lower_bound_ml}
In all of this section, we consider the size-5 alphabet $\Sigma := \{a,b,c,d,e\}$ and the same size-5 set of actions $\act := \Sigma$. Furthermore, for all $\alpha \in \act$, we denote the operator $\langle \alpha \rangle^{\geq 1}$ simply by $\langle \alpha \rangle$. We first define below the samples, parameterized by an integer $n \in \N$, that we will consider to establish Proposition~\ref{prop:lower_bound_ml}. 
\begin{definition}
	\label{def:samples_ml_lower_bound}
	For all $n \in \N$, we let $A_n = (Q_n,\Sigma,I_n,\delta_n,F_n)$ denote the finite automaton of Theorem~\ref{thm:automtaton_least_word_exponential}. For all $Z \subseteq Q_n$, and $\alpha \in \Sigma$, we let $\delta(Z,\alpha) := \cup_{q \in Z} \delta(q,\alpha)$. Then, we define inductively on words $u \in \Sigma^*$ the set of states $\delta_n^*(Z,u) \subseteq Q_n$ to which there is $u$-labeled path from some state in $Z$. Formally:
	\begin{align*}
		\delta_n^*(Z,\varepsilon) & := Z \subseteq Q_n\\
		\forall u \in \Sigma^*,\; \alpha \in \Sigma,\; \delta_n^*(Z,u \cdot \alpha) & := \delta(\delta_n^*(Z,u),\alpha) \subseteq Q_n
	\end{align*}
	
	For all $Z \subseteq Q_n$ and $P \subseteq \{p\}$, we let $K_n^Z(P) \in \mathcal{K}(\prop,\act)$ denote the Kripke structure $K_n^Z(P) = (Q_n,Z,\act,\delta_n,\prop,\pi_{n}^P)$ mimicking the automaton $A_n$ with $Z$ as set of initial states, and such that for all $q \in Q_n$, we have:
	\begin{align*}
		\pi_{n}^P(q) := \begin{cases}
			\{p\} \setminus P & \text{ if }q \in Q_n \setminus F_n\\
			P & \text{ if }q \in F_n
		\end{cases}
	\end{align*}
	We also let $K(P) := (\{q^{\msf{idle}},q_n^{\msf{win}}\},\{q^{\msf{idle}}\},\act,\delta,\prop,\pi^P)$ be such that:
	\begin{itemize}
		\item For all $\alpha \in \act$:
		\begin{align*}
			\delta(q^{\msf{idle}},\alpha) & := \{q^{\msf{idle}},q^{\msf{loop}}\} \\
			\delta(q^{\msf{loop}},\alpha) & := \{q_n^{\msf{win}}\}
		\end{align*}
		\item For all $q \in Q_n'$:
		\begin{align*}
			\pi^P(q^{\msf{idle}}) & := P \\
			\pi^P(q^{\msf{loop}}) & := \{p\} \setminus P
		\end{align*}
	\end{itemize}
	
	Consider now some $n \in \N$ and $Z \subseteq Q_n$. We let: $\mathcal{S}^{Z,[\cdot]}_n := (\{K_n^Z(\emptyset)\},\{K(\emptyset)\})$ and $\mathcal{S}^{Z,\langle\cdot\rangle}_n := (\{K(\{p\})\},\{K_n^{q}(\{p\}) \mid q \in Z\})$.
	
	\iffalse
	%For all $P_1.P_2 \subseteq \{p\}$, f
	For all $Z \subseteq Q_n$, we let $K_n^Z \in \mathcal{K}(\prop,\act)$ denote the Kripke structure $K_n^Z(P_1,P_2) = (Q_n,Z,\act,\delta_n,\prop,\pi_{n}^{(P_1,P_2)})$ mimicking the automaton $A_n$ with $Z$ as set of initial states, and such that for all $q \in Q_n$, we have:
	\begin{align*}
		\pi_{n}^{(P_1,P_2)}(q) := \begin{cases}
			P_1 & \text{ if }q \in Q_n \setminus F_n\\
			P_2 & \text{ if }q \in F_n
		\end{cases}
	\end{align*}
	We also let $K(P_1,P_2) := (\{q^{\msf{idle}},q_n^{\msf{win}}\},\{q^{\msf{idle}}\},\act,\delta,\prop,\pi^{(P_1,P_2)})$ be such that:
	\begin{itemize}
		\item For all $\alpha \in \act$:
		\begin{align*}
			\delta(q^{\msf{idle}},\alpha) & := \{q^{\msf{idle}},q^{\msf{loop}}\} \\
			\delta(q^{\msf{loop}},\alpha) & := \{q_n^{\msf{win}}\}
		\end{align*}
		\item For all $q \in Q_n'$:
		\begin{align*}
			\pi^{(P_1,P_2)}(q^{\msf{idle}}) & := P_2 \\
			\pi^{(P_1,P_2)}(q^{\msf{loop}}) & := P_1
		\end{align*}
	\end{itemize}
	
	Consider now some $n \in \N$ and $Z \subseteq Q_n$. We let: $\mathcal{S}^{Z,[\cdot]}_n := (\mathcal{P}_n^Z(\{p\},\emptyset),\mathcal{N}_n(\{p\},\emptyset))$, and $\mathcal{S}^{Z,\langle\cdot\rangle}_n := (\mathcal{N}_n(\emptyset,\{p\}),\mathcal{P}_n^Z(\emptyset,\{p\}))$.
	%\begin{itemize}
	%	\item $\mathcal{P}_n^Z(x) := \{ K_n^Z(x) \}$;
	%	\item $\mathcal{N}_n(x) := \{ K(x) \}$;
	%	\item $\mathcal{S}^{Z,[\cdot]}_n(x) := (\mathcal{P}_n^Z(x),\mathcal{N}_n(x))$, and $\mathcal{S}^{Z,\langle\cdot\rangle}_n(x) := (\mathcal{N}_n(x),\mathcal{P}_n^Z(x))$.
	%\end{itemize}
	\fi
\end{definition}

This definition of $\mathcal{S}^{Z,[\cdot]}_n$ satisfies the lemma below.
\begin{lemma}
	\label{lem:lower_bound_ml_sample_works}
	Consider an $\ML(\prop,\act)$-fragment $\msf{L}$ such that such that $\vee,\neg \notin \msf{Op}$. Let $n \in \N$, and $Z \subseteq Q_n$ such that $Z \cap F_n \neq \emptyset$. Assume that there exists an $\msf{L}$-formula $\varphi$ that is $\mathcal{S}^{Z,[\cdot]}_n$-separating. In that case, there is a sub-formula $\psi \in \msf{Sub}(\varphi)$ and some $\alpha \in \Sigma$ such that:
	\begin{itemize}
		\item $\msf{sz}(\varphi) \geq \msf{sz}(\psi) + 1$;
		\item the $\msf{L}$-formula $\psi$ is $\mathcal{S}^{\delta(Z,\alpha),[\cdot]}_n$-separating.
	\end{itemize}
\end{lemma}
\begin{proof}
	Let us prove the result by exhaustion:
	\begin{itemize}
		%\item Assume that $\varphi \in \{\top,\bot\}$. Then, it cannot be $\mathcal{S}^{Z,[\cdot]}_n(x)$-separating.
		\item Assume that $\varphi = p$. Since $Z \cap F_n \neq \emptyset$, an initial state of $K_n^Z(\emptyset)$ is labeled by $\emptyset$, hence $p \not\models K_n^Z(\emptyset)$, which is a positive model.
		\item Assume that $\varphi = \varphi_1 \wedge \varphi_2$. Then, let $\psi \in \{\varphi_1,\varphi_2\}$ be an $\msf{L}$-formula such that $K(\emptyset) \not\models \psi$. Then, we have $\msf{sz}(\psi) < \msf{sz}(\varphi)$ and $\psi$ is $\mathcal{S}_n^{Z,[\cdot]}$-separating. Thus, we can apply (by induction) the  argument to the formula $\psi$ to obtain an $\msf{L}$-formula $\psi' \in \msf{Sub}(\psi) \subseteq \msf{Sub}(\varphi)$ satisfying the assumptions of the lemma.
		\item Assume that $\varphi = \langle \alpha \rangle \varphi'$ for some $\alpha \in \act$.
		Since the initial state $q^{\msf{idle}}$ of the Kripke structure $K(\emptyset)$ has as $\alpha$-successor the state $q^{\msf{loop}}$ labeled by $\{p\}$, which is thus satisfied by all $\msf{L}$-formulas, it follows that the formula $\varphi$ cannot reject $K(\emptyset)$.
		\item Assume that $\varphi = [\alpha] \varphi'$ for some $\alpha \in \act$. As mentioned above, the state $q^{\msf{loop}}$ in the Kripke structure $K(\emptyset)$ is satisfied by all $\msf{L}$-formulas. Furthermore, we have $\delta(q^{\msf{idle}},\alpha) = \{q^{\msf{idle}},q^{\msf{loop}}\}$. It follows that $q^{\msf{idle}} \not\models \varphi'$. In addition, for all $q \in Z$, since $q \models \varphi$, for all $q' \in \delta(q,\alpha)$, we have $q' \models \varphi'$. Therefore, the $\msf{L}$-formula $\varphi'$ is  $\mathcal{S}^{\delta(Z,\alpha),[\cdot]}_n$-separating and such that $\msf{sz}(\varphi') + 1 = \msf{sz}(\varphi)$.
	\end{itemize}
	Overall, whenever $\varphi$ is $\mathcal{S}_n^{Z,[\cdot]}$-separating, it is possible to extract a sub-formula $\psi \in \msf{Sub}(\varphi)$ satisfying the conditions of the lemma.		
\end{proof}

In a somewhat symmetrical manner, we prove the result below about the sample $\mathcal{S}^{Z,\langle\cdot\rangle}_n$.
\begin{lemma}
	\label{lem:lower_bound_ml_sample_works_vee}
	Consider an $\ML(\prop,\act)$-fragment $\msf{L}$ such that $\wedge,\lnot \notin \msf{Op}$. Let $n \in \N$, and $Z \subseteq Q_n$ such that $Z \cap F_n \neq \emptyset$. Assume that there exists an $\msf{L}$-formula $\varphi$ that is $\mathcal{S}^{Z,\langle\cdot\rangle}_n$-separating. In that case, there is a sub-formula $\psi \in \msf{Sub}(\varphi)$ and some $\alpha \in \Sigma$ such that:
	\begin{itemize}
		\item $\msf{sz}(\varphi) \geq \msf{sz}(\psi) + 1$;
		\item the $\msf{L}$-formula $\psi$ is $\mathcal{S}^{\delta(Z,\alpha),\langle\cdot\rangle}_n$-separating.
	\end{itemize}
\end{lemma}
\begin{proof}
	Let us prove the result by exhaustion:
	\begin{itemize}
		\item Assume that $\varphi = p$. Since there is some $q \in Z \cap F_n$, it follows that the initial state of $K_n^{\{q\},\langle\cdot\rangle}(\{p\})$ is labeled by $p$, hence $p \models K_n^{\{q\},\langle\cdot\rangle}(\{p\})$, which is a negative model.
		\item Assume that $\varphi = \varphi_1 \vee \varphi_2$. Then, let $\psi \in \{\varphi_1,\varphi_2\}$ be an $\msf{L}$-formula such that $K(\{p\}) \models \psi$. Then, we have $\msf{sz}(\psi) < \msf{sz}(\varphi)$ and $\psi$ is $\mathcal{S}_n^{Z,\langle\cdot\rangle}$-separating. Thus, we can apply (by induction) the  argument to the formula $\psi$ to obtain an $\msf{L}$-formula $\psi' \in \msf{Sub}(\psi) \subseteq \msf{Sub}(\varphi)$ satisfying the assumptions of the lemma.
		\item Assume that $\varphi = \langle \alpha \rangle \varphi'$ for some $\alpha \in \act$. Since $\pi^{\{p\}}(q^{\msf{loop}}) = \emptyset$ and $\delta(q^{\msf{loop}},\alpha') = \{q^{\msf{loop}}\}$ for all $\alpha' \in \act$, it follows that no $\msf{L}$-formula satisfies the state $q^{\msf{loop}}$. Hence, since $\delta(q^{\msf{idle}},\alpha) = \{q^{\msf{idle}},q^{\msf{loop}}\}$, it follows that $q^{\msf{idle}} \models \varphi'$. In addition, for all $q \in Z$, in the Kripke structure $K_n^{\{q\},\langle\cdot\rangle}(\{p\})$, since $q \not\models \varphi$, for all $q' \in \delta(q,\alpha)$, we have $q' \not\models \varphi'$. Therefore, the $\msf{L}$-formula $\varphi'$ is  $\mathcal{S}^{\delta(Z,\alpha),\langle\cdot\rangle}_n$-separating and such that $\msf{sz}(\varphi') + 1 = \msf{sz}(\varphi)$.
		\item Assume that $\varphi = [\alpha] \varphi'$ for some $\alpha \in \act$. As mentioned above, the state $q^{\msf{loop}}$ in the Kripke structure $K(\{p\})$ is not satisfied by any $\msf{L}$-formula. Furthermore, we have $\delta(q^{\msf{idle}},\alpha) = \{q^{\msf{idle}},q^{\msf{loop}}\}$. It follows that the formula $\varphi$ cannot accept the structure $K(\emptyset)$.
	\end{itemize}
	Overall, whenever $\varphi$ is $\mathcal{S}_n^{Z,\langle\cdot\rangle}$-separating, it is possible to extract a sub-formula $\psi \in \msf{Sub}(\varphi)$ satisfying the conditions of the lemma.		
\end{proof}

Let us now proceed to the proof of Proposition~\ref{prop:lower_bound_ml}.
\begin{proof}
	Consider some $i \in \N$. We let $t_i := \sum_{K \in \mathcal{S}_i^{I_i,[\cdot]}} |Q_K| = \sum_{K \in \mathcal{S}_i^{I_i,\langle\cdot\rangle}} |Q_K| = 25i+113$ (since $I_i$ is a singleton). Let $n_0 \in \N$ be such that, for all $i \geq n_0$, we have $x_i := (2^i-1)\cdot (i+1) +1 \geq 2^{i+113/25}$.  
	
	Consider any $n \in \N$, and let $i := \max(n,n_0)$. Assume that $\vee \notin \msf{Op}'$ (resp. $\wedge \notin \msf{Op}'$). Let $\mathcal{S} := \mathcal{S}_i^{I_i,[\cdot]}(x)$ (resp. $\mathcal{S} := \mathcal{S}_i^{I_i,\langle\cdot\rangle}(x)$). Let us show that the minimal size of an $\msf{L}'$-formula that is $\mathcal{S}$-separating is $x_i+1$. Since we have $i \leq t_i = 25i+113$ and $x_i \geq 2^{t_i/25}$, the result follows.
	
	By Theorem~\ref{thm:automtaton_least_word_exponential}, the least size of a word $u_i = u_i^0 \cdots u_i^k \in \Sigma^*$ such that $\delta^*(I_i,u_i) \subseteq Q_i \setminus F_i$ is $k+1 = x_i$. Now, consider the $\msf{L}'$-formula $\varphi := [u_i^0] [u_i^1] \cdots [u_i^k] p$ (resp. $\varphi := \langle u_i^0\rangle \langle u_i^1\rangle \cdots \langle u_i^k\rangle p$) of size $\msf{sz}(\varphi) = |u_i| +1 = x_i + 1$. By definition of $\mathcal{S}$, this formula is $\mathcal{S}$-separating. 
	
	Consider now any $\mathcal{S}$-separating $\msf{L}'$-formula $\varphi$. By iteratively applying Lemma~\ref{lem:lower_bound_ml_sample_works} (resp. Lemma~\ref{lem:lower_bound_ml_sample_works_vee}), we obtain that there is some $\msf{L}'$-formula $\psi \in \msf{Sub}(\varphi)$ and a word $u \in \Sigma^*$ such that $\delta^*(I_i,u) \subseteq Q_i \setminus F_i$ and $\msf{sz}(\varphi) \geq \msf{sz}(\psi) + |u|$% (and $\psi$ is $\mathcal{S}_i^{\delta(I_i,u)}$-separating)
	. Since $\msf{sz}(\psi) \geq 1$, and by Theorem~\ref{thm:automtaton_least_word_exponential}, it follows that $\msf{sz}(\varphi) \geq x_i+1$. The result follows.
\end{proof}

\end{document}